\documentclass{CSML}

\def\dOi{12(1:4)2016}
\lmcsheading%
{\dOi}
{1--46}
{}
{}
{Feb.~20, 2015}
{Mar.~\phantom09, 2016}
{}

\ACMCCS{[{\bf Theory of computation}]: Models of
  computation---Computability---Lambda calculus; Models of
  computation---Abstact machines; Logic---Proof theory\,/\,Linear
  logic\,/\,Equational logic and rewriting; Computational complexity
  and cryptography---Complexity theory and logic}

\keywords{lambda calculus, cost models, linear logic, invariance,
  standardization, sharing, operational semantics}

\usepackage{ifthen}

\newboolean{talk}
\setboolean{talk}{false}
\newboolean{paper}
\setboolean{paper}{false}

\newboolean{IEEEstyle}
\setboolean{IEEEstyle}{false}
\newboolean{lipicsstyle}
\setboolean{lipicsstyle}{false}
\newboolean{eptcsstyle}
\setboolean{eptcsstyle}{false}
\newboolean{sigplanstyle}
\setboolean{sigplanstyle}{false}

\newboolean{needstheorems}
\setboolean{needstheorems}{false}
\newboolean{withimages}
\setboolean{withimages}{false}
\newboolean{withproofs}
\setboolean{withproofs}{true}

\newboolean{french}
\setboolean{french}{false}

\ifthenelse{\boolean{talk}}{}{\usepackage[utf8]{inputenc}}
\ifthenelse{\boolean{french}}{\usepackage[french]{babel}}{
  \ifthenelse{\boolean{lipicsstyle}}{}{\usepackage[english]{babel}}}
\usepackage{amssymb}
\usepackage{amsmath}
\usepackage{graphicx}
\usepackage{bussproofs}
\usepackage{fancybox}
\usepackage{cmll}
\usepackage{stmaryrd}
\usepackage{multirow}
\ifthenelse{\boolean{talk}}{\usepackage{color}}{
  \ifthenelse{\boolean{lipicsstyle}}{\usepackage{color}}{\usepackage[usenames]{xcolor}}}
\ifthenelse{\boolean{IEEEstyle}}{
\usepackage[bookmarks={false}]{hyperref}}{
\usepackage{hyperref}}
\usepackage{enumitem}

\ifthenelse{\boolean{IEEEstyle}}{
\setboolean{needstheorems}{true}}{}

\ifthenelse{\boolean{eptcsstyle}}{
\setboolean{needstheorems}{true}}{}

\ifthenelse{\boolean{sigplanstyle}}{
\setboolean{needstheorems}{true}}{}

\ifthenelse{\boolean{needstheorems}}{

\usepackage{amsthm}

    \newtheorem{theorem}{Theorem}[section]
    \newtheorem{lemma}[theorem]{Lemma}
    \newtheorem{corollary}[theorem]{Corollary}
    \newtheorem{proposition}[theorem]{Proposition}
    \newtheorem{definition}[theorem]{Definition}
    \newtheorem{remark}[theorem]{Remark}
}{}

\newcommand{\myproof}[1]{
\ifthenelse{\boolean{withproofs}}{#1}{}}

\newcommand{\withproofs}[1]{
\ifthenelse{\boolean{withproofs}}{#1}{}}

\newcommand{\withoutproofs}[1]{
\ifthenelse{\boolean{withproofs}}{}{#1}}

\newcommand{\tm}{t}
\newcommand{\tmtwo}{u}
\newcommand{\tmthree}{r}
\newcommand{\tmfour}{p}
\newcommand{\tmfive}{q}
\newcommand{\tmsix}{s}

\newcommand{\var}{x}
\newcommand{\vartwo}{y}
\newcommand{\varthree}{z}

\newcommand{\rootRew}[1]{\mapsto_{#1}}
\newcommand{\Rew}[1]{\rightarrow_{#1}}
\newcommand{\Rewn}[1]{\rightarrow_{#1}^*}

\renewcommand{\to}{\Rew{}}
\newcommand{\rtob}{\rootRew{\beta}}
\newcommand{\tob}{\Rew{\beta}}

\newcommand{\lssym}{{\tt ls}}

\newcommand{\db}{{\tt dB}}

\newcommand{\msym}{{\tt m}}
\newcommand{\esym}{{\tt e}}

\newcommand{\ctxholep}[1]{[#1]}
\newcommand{\ctxhole}{\ctxholep{\cdot}}

\newcommand{\ctx}{C}
\newcommand{\ctxtwo}{D}
\newcommand{\ctxthree}{E}
\newcommand{\ctxfour}{F}
\newcommand{\ctxp}[1]{\ctx\ctxholep{#1}}

\newcommand{\apctx}{A}

\newcommand{\nbvctxtwo}[1]{\nbvctxtwo{#1}}

\newcommand{\sctx}{L}
\newcommand{\sctxtwo}{L'}
\newcommand{\sctxthree}{L''}

\newcommand{\sctxp}[1]{\sctx\ctxholep{#1}}

\newcommand{\sctxtwop}[1]{\sctxtwo\ctxholep{#1}}
\newcommand{\sctxthreep}[1]{\sctxthree\ctxholep{#1}}

\newcommand{\eqalpha}{\equiv_\alpha}
\newcommand{\defeq}{:=}
\newcommand{\grameq}{::=}

\newcommand{\isub}[2]{\{#1/#2\}}
\newcommand{\esub}[2]{[#1/#2]}

\newcommand{\rtodb}{\rootRew{\db}}

\newcommand{\todbp}[1]{\Rew{\db#1}}
\newcommand{\todb}{\todbp{}}

\newcommand{\rtols}{\rootRew{\lssym}}

\newcommand{\tols}{\Rew{\lssym}}

\newcommand{\llbrace}{\{ \kern -0.27em \vert}
\newcommand{\rrbrace}{\vert \kern -0.27em \}}

\renewcommand{\l}{\lambda}
\newcommand{\ie}{{\em i.e.}}
\newcommand{\eg}{{\em e.g.}}
\newcommand{\ih}{{\em i.h.}}
\newcommand{\fv}[1]{{\tt fv}(#1)}

\newcommand{\deff}[1]{\emph{#1}}

\newcommand{\ben}[1]{{\color{red} {#1}}}

\newcommand{\ignore}[1]{}

\newcommand{\myinput}[1]{\ifthenelse{\boolean{withimages}}{\input{#1}}{}}

\newcommand{\mellies}{{Melli{\`e}s}}
\newcommand{\levy}{{L{\'e}vy}}

\newcommand{\linlogic}{linear logic}

\newcommand{\set}[1]{\{#1\}}

\newcommand{\nat}{\mathbb{N}}

\newcommand{\setone}{S}
\newcommand{\settwo}{R}

\newcommand{\pns}{proof nets}

\newcommand{\tom}{\Rew{\msym}}
\newcommand{\toe}{\Rew{\esym}}

\newcommand{\size}[1]{|#1|}

\newcommand{\code}{\underline{\tm}}
\newcommand{\codetwo}{\underline{\tmtwo}}

\newenvironment{varitemize}
{
\begin{list}{\labelitemi}
{\setlength{\itemsep}{0pt}
 \setlength{\topsep}{0pt}
 \setlength{\parsep}{0pt}
 \setlength{\partopsep}{0pt}
 \setlength{\leftmargin}{15pt}
 \setlength{\rightmargin}{0pt}
 \setlength{\itemindent}{0pt}
 \setlength{\labelsep}{5pt}
 \setlength{\labelwidth}{10pt}
}}
{
 \end{list} 
}

\newenvironment{varenumerate}{\begin{enumerate}}{\end{enumerate}}
\newenvironment{varvarenumerate}{\begin{enumerate}}{\end{enumerate}}
\newenvironment{varvarvarenumerate}{\begin{enumerate}}{\end{enumerate}}

\setenumerate{itemsep=0pt,topsep=0pt,parsep=0pt,partopsep=0pt,leftmargin=20pt,rightmargin=0pt,labelsep=5pt,itemindent=0pt,labelwidth=20pt}

\usepackage{tikz}
\usetikzlibrary{matrix,arrows}
\usetikzlibrary{decorations.shapes}
\usetikzlibrary{decorations.text}
\usetikzlibrary{decorations.pathmorphing}
\usetikzlibrary{decorations.markings}

\tikzset{
node distance=1.3cm, auto,
every node/.style={font=\tiny },
ocenter/.style={baseline={([yshift=-.5ex, xshift=-.5ex]current bounding box)}},  
labelBeginAbove/.style={postaction={decorate,decoration={markings,mark=at position 0 with {\node[inner sep= 0.6pt, above=1pt]{\tiny #1};}} } },
labelBeginBelow/.style={postaction={decorate,decoration={markings,mark=at position 0 with {\node[inner sep= 0.6pt, below=1pt]{\tiny #1};}}}},
labelEndAbove/.style={postaction={decorate,decoration={markings,mark=at position 1 with {\node[inner sep= 0.6pt, above=1pt]{\tiny #1};}}}},
labelEndBelow/.style={postaction={decorate,decoration={markings,mark=at position 1 with {\node[inner sep= 0.6pt, below=1pt]{\tiny #1};}}}},
labelEndRight/.style={postaction={decorate,decoration={markings,mark=at position 1 with {\node[inner sep= 0.6pt, right=1pt]{\tiny #1};}}}},
labelEndLeft/.style={postaction={decorate,decoration={markings,mark=at position 1 with {\node[inner sep= 0.6pt, left=1pt]{\tiny #1};}}}}
}

\newcommand{\nodeHorDist}{2cm}
\newcommand{\nodeVerDist}{1cm}

\newcommand{\commDiagramRed}[8]{
\begin{tikzpicture}[ocenter]
	\node (s) {\normalsize #1};
  \node at (s.center)  [right=1.7*\nodeHorDist](s1){\normalsize #2};
  \node at (s.center)  [below=\nodeVerDist](s2) {\normalsize #3};
  \node at (s1|-s2) (t) {\normalsize #4};
  
  \draw[->] (s) to node {#5} (s1);
  \draw[->] (s2) to node {#6} (t);
  \draw[->] (s) to node {#7} (s2);
  \draw[->] (s1) to node {#8} (t);
\end{tikzpicture} 
}

\renewcommand{\ctxholep}[1]{\langle #1\rangle}
\newcommand{\ctxtwop}[1]{\ctxtwo\ctxholep{#1}}
\newcommand{\ctxthreep}[1]{\ctxthree\ctxholep{#1}}

\newcommand{\esmeas}[1]{|#1|_{[\cdot]}}

\newcommand{\unfsym}{\rotatebox[origin=c]{-90}{$\rightarrow$}}
\newcommand{\unf}[1]{#1\unfsym\,}
\newcommand{\relunf}[2]{\unf{#1}_{#2}}

\newcommand{\opt}{useful}

\newcommand{\deriv}{\rho}
\newcommand{\derivtwo}{\tau}

\newcommand{\sctximp}{\hat\sctx}

\newcommand{\prefix}{\prec_p}
\newcommand{\outin}{\prec_O}
\newcommand{\leftright}{\prec_L}
\newcommand{\leftout}{\prec_{\lo}}

\newcommand{\lo}{LO}
\newcommand{\lou}{\lo U}
\newcommand{\ilou}{i\lou}

\newcommand{\lob}{\lo$\beta$}
\newcommand{\ilob}{i\lob}

\newcommand{\redex}{R}
\newcommand{\redextwo}{Q}
\newcommand{\redexthree}{P}
\newcommand{\sizedb}[1]{|#1|_{\db}}
\newcommand{\sizels}[1]{|#1|_{\ls}}

\newcommand{\tostrat}{\leadsto}
\newcommand{\toes}{\leadsto_{X}}
\newcommand{\reflemma}[1]{Lemma~\ref{l:#1}}
\newcommand{\reflemmap}[2]{Lemma~\ref{l:#1}.\ref{p:#1-#2}}
\newcommand{\reflemmaeqp}[2]{{\ref{l:#1}.\ref{p:#1-#2}}}

\newcommand{\refcorollaryp}[2]{Corollary~\ref{c:#1}.\ref{p:#1-#2}}

\newcommand{\refthm}[1]{Theorem~\ref{thm:#1}}
\newcommand{\reftm}[1]{Theorem~\ref{tm:#1}}
\newcommand{\refprop}[1]{Proposition~\ref{prop:#1}}
\newcommand{\refsect}[1]{Sect.~\ref{sect:#1}}

\newcommand{\refeq}[1]{(\ref{eq:#1})}

\newcommand{\refcoro}[1]{Corollary~\ref{coro:#1}}

\newcommand{\refdef}[1]{Definition~\ref{def:#1}}
\newcommand{\refrem}[1]{Remark~\ref{rem:#1}}
\newcommand{\refpoint}[1]{Point~\ref{p:#1}}

\newcommand{\ls}{\lssym}

\newcommand{\bctx}{B}
\newcommand{\bctxp}[1]{\bctx\ctxholep{#1}}
\newcommand{\bctxtwo}{B'}

\newcommand{\ap}[2]{#1[#2]}
\renewcommand{\esub}[2]{[#1{\shortleftarrow}#2]}
\renewcommand{\isub}[2]{\{#1{\shortleftarrow}#2\}}

\newcommand{\gctx}{C}
\newcommand{\gctxtwo}{D}
\newcommand{\gctxthree}{E}

\newcommand{\gctxp}[1]{\gctx\ctxholep{#1}}

\newcommand{\gctxthreep}[1]{\gctxthree\ctxholep{#1}}

\renewcommand{\ctx}{S}
\renewcommand{\ctxtwo}{S'}
\renewcommand{\ctxthree}{S''}
\renewcommand{\ctxfour}{S'''}

\newcommand{\lsc}{LSC}
\newcommand{\toblo}{\Rew{\mbox{\tiny LO}\beta}}
\newcommand{\tolo}{\Rew{\mbox{\tiny LO}}}
\newcommand{\tolou}{\Rew{\mbox{\tiny LOU}}}

\renewcommand{\ben}[1]{#1}

\renewcommand{\ap}[2]{#1#2}

\newcommand{\la}[1]{\lambda#1.}

\newcommand{\tolhs}{\Rew{\tt lhs}}

\newcommand{\str}{\Rew{\R}} 
\newcommand{\nos}[2]{\#_{#1}(#2)}
\newcommand{\strp}[3]{#1^{#2}\!\!(#3)}

\renewcommand{\redex}{R} 
\renewcommand{\redextwo}{P} 
\renewcommand{\redexthree}{Q} 

\newcommand{\R}{\mathcal{T}}

\usepackage{hyperref}

\begin{document}

\setlength{\pdfpageheight}{\paperheight}
\setlength{\pdfpagewidth}{\paperwidth}

\title[(Leftmost-Outermost) Beta Reduction is Invariant, Indeed]{(Leftmost-Outermost)\\ Beta Reduction is Invariant, Indeed}

\newcommand{\ugo}[1]{\ben{#1}}
\newcommand{\midd}{\; \; \mbox{\Large{$\mid$}}\;\;}

\author[B.~Accattoli]{Beniamino Accattoli\rsuper a}
\address{{\lsuper a}INRIA \&\ LIX \'Ecole Polytechnique}
\email{beniamino.accattoli@inria.fr}

\author[U.~Dal Lago]{Ugo Dal Lago\rsuper b}
\address{{\lsuper b}Dipartimento di Informatica - Scienza e Ingegneria
Università degli Studi di Bologna}
\email{dallago@cs.unibo.it}

\begin{abstract}
Slot and van Emde Boas' weak invariance thesis states that
\emph{reasonable} machines can simulate each other within a
polynomial overhead in time. Is $\l$-calculus a reasonable machine?
Is there a way to measure the computational complexity of a $\l$-term?
This paper presents the first complete positive answer to this
long-standing problem. Moreover, our answer is completely
machine-independent and based on a standard notion in the theory of
$\l$-calculus: the length of a leftmost-outermost derivation to normal
form is an invariant, \ie\ reasonable, cost model. Such a
theorem cannot be proved by directly relating $\l$-calculus with
Turing machines or random access machines, because of
the \emph{size-explosion problem}: there are terms that in a linear
number of steps produce an exponentially large output. The first step
towards the solution is to shift to a notion of evaluation for which
the length and the size of the output are linearly related. This is
done by adopting the \emph{linear substitution calculus} (LSC), a
calculus of explicit substitutions modeled after
\linlogic\ \pns\ and admitting a decomposition of leftmost-outermost
derivations with the desired property. Thus, the LSC is invariant with
respect to, say, random access machines. The second step is to show
that the LSC is invariant with respect to the $\l$-calculus. The
size explosion problem seems to imply that this is not possible:
having the same notions of normal form, evaluation in the LSC is
exponentially longer than in the $\l$-calculus. We solve such
an \emph{impasse} by introducing a new form of shared normal form and
shared reduction, called \emph{useful}. Useful
evaluation produces a compact, shared representation of the
normal form, by avoiding those steps that only unshare the output
without contributing to $\beta$-redexes, \ie\ the steps that cause the
blow-up in size. The main technical contribution of the paper is
indeed the definition of useful reductions and the thorough analysis
of their properties.
\end{abstract}

\maketitle

\section*{Introduction}

Theoretical computer science is built around algorithms, computational
models, and machines: an algorithm describes a solution to a problem
with respect to a fixed computational model, whose role is to provide
a handy abstraction of concrete machines. The choice of the model
reflects a tension between different needs. For complexity analysis,
one expects a neat relationship between the primitives of the model
and the way in which they are effectively implemented.  In this
respect, random access machines are often taken as the reference
model, since their definition closely reflects the von Neumann
architecture. The specification of algorithms unfortunately lies at
the other end of the spectrum, as one would like them to be as
machine-independent as possible. In this case programming languages
are the typical model. Functional programming languages, thanks to
their higher-order nature, provide very concise and abstract
specifications. Their strength is also their weakness: the abstraction
from physical machines is pushed to a level where it is no longer
clear how to measure the complexity of an algorithm. Is there a way in
which such a tension can be resolved?

The tools for stating the question formally are provided by complexity
theory and by Slot and van Emde Boas' invariance thesis
\cite{DBLP:conf/stoc/SlotB84}, which stipulates when any Turing
complete computational model can be considered reasonable:
\begin{center}
  \emph{Reasonable} computational models simulate each other\\ with
  polynomially bounded overhead in time,\\ and constant factor overhead
  in space.
\end{center}
The \emph{weak} invariance thesis is the variant where the requirement
about space is dropped, and it is the one we will actually work with
in this paper (alternatively called \emph{extended}, \emph{efficient},
\emph{modern}, or \emph{complexity-theoretic} Church(-Turing)
thesis). The idea behind the thesis is that for reasonable models the
definition of every polynomial or super-polynomial class such as
$\mathbf{P}$ or $\mathbf{EXP}$ does not rely on the chosen model. On
the other hand, it is well-known that sub-polynomial classes depend
very much on the model. A first refinement of our question then is:
are functional languages invariant with respect to standard models
like random access machines or Turing machines?

Invariance results have to be proved via an appropriate measure of
time complexity for programs, \ie\ a \emph{cost model}. The natural
measure for functional languages is the \emph{unitary} cost model,
\ie\ the number of evaluation steps. There is, however, a subtlety.
The evaluation of functional programs, in fact, depends very much on
the evaluation strategy chosen to implement the language, while the
reference model for functional languages, the $\l$-calculus, is so
machine-independent that it does not even come with a deterministic
evaluation strategy.  And which strategy, if any, gives us the most
natural, or \emph{canonical} cost model (whatever that means)?  These
questions have received some attention in the last decades.  The
number of optimal parallel $\beta$-steps (in the sense of
\levy\ \cite{thesislevy}) to normal form has been shown \emph{not} to
be a reasonable cost model: there exists a family of terms that
reduces in a polynomial number of parallel $\beta$-steps, but whose
intrinsic complexity is
non-elementary~\cite{DBLP:conf/icfp/LawallM96,DBLP:conf/popl/AspertiM98}. If
one considers the number of \emph{sequential} $\beta$-steps (in a
given strategy, for a given notion of reduction), the literature
offers some partial positive results, all relying on the use of
sharing (see below for more details).

Sharing is indeed a key ingredient, for one of the issues here is due
to the \emph{representation of terms}. The ordinary way of
representing terms indeed suffers from the \emph{size-explosion
  problem}: even for the most restrictive notions of reduction
(e.g. Plotkin's weak reduction), there is a family of terms
$\{\tm_n\}_{n\in\nat}$ such that $|\tm_n|$ is linear in $n$, $\tm_n$
evaluates to its normal form in $n$ steps, but at the $i$-th step a term
of size $2^i$ is copied, producing a normal form of size exponential
in $n$. Put differently, an evaluation sequence of linear length can
possibly produce an output of exponential size. At first sight, then,
there is no hope that evaluation lengths may provide an invariant cost
model. The idea is that such an \emph{impasse} can be avoided by
sharing common subterms along the evaluation process, in order to
keep the representation of the output compact, \ie\ polynomially
related to the number of evaluation steps. But is appropriately
managed sharing enough? The answer is positive, at least for certain
restricted forms of reduction: the number of steps is already known to
be an invariant cost model for weak
reduction~\cite{DBLP:conf/fpca/BlellochG95,DBLP:conf/birthday/SandsGM02,DBLP:journals/tcs/LagoM08,DBLP:journals/corr/abs-1208-0515}
and for head reduction~\cite{DBLP:conf/rta/AccattoliL12}.

If the problem at hand consists in computing the \emph{normal form} of
an arbitrary $\lambda$-term, however, no positive answer was known, to the best of our knowledge, before our result. We
believe that not knowing whether the $\lambda$-calculus in its full
generality is a reasonable machine is embarrassing for the
$\lambda$-calculus community. In addition, this problem is relevant
in practice: proof assistants often need to check whether two terms
are convertible, itself a problem usually reduced to the computation of normal forms.

In this paper, we give a positive answer to the question above, by
showing that leftmost-outermost (\lo, for short) reduction \emph{to
  normal form} indeed induces an invariant cost model. Such an
evaluation strategy is \emph{standard}, in the sense of the
standardization theorem, one of the central theorems in the theory of
$\l$-calculus, first proved by Curry and Feys
\cite{curry1958combinatory}. The relevance of our cost model is given
by the fact that \lo\ reduction is an abstract concept from rewriting
theory which at first sight is totally unrelated to complexity
analysis. \ugo{Moreover, the underlying computational model is
very far from traditional, machine-based models like Turing machines
and RAMs.}

Another view on this problem comes in fact from rewriting theory
itself. It is common practice to specify the operational semantics of
a language via a rewriting system, whose rules always employ some form
of substitution, or at least of copying, of subterms. Unfortunately,
this practice is very far away from the way languages are
implemented, as actual interpreters perform copying in a very
controlled way (see, \eg,
\cite{Wad:SemPra:71,Pey:ImplFunProgLang:87}). This discrepancy induces
serious doubts about the relevance of the computational model. Is
there any theoretical justification for copy-based models, or more
generally for rewriting theory as a modeling tool? In this paper we
give a very precise answer, formulated within rewriting theory
itself. \ben{A second contribution of the paper, indeed, is 
  a rewriting analysis of the technique used to prove the invariance
  result.}

As in our previous work \cite{DBLP:conf/rta/AccattoliL12}, we prove
our result by means of the \emph{linear substitution calculus} (see
also \cite{DBLP:conf/rta/Accattoli12,non-standard-preprint}), a simple
calculus of explicit substitutions (ES, for short) introduced by
Accattoli and Kesner, that arises from linear logic and graphical
syntaxes and it is similar to calculi studied by de
Bruijn~\cite{deBruijn87}, Nederpelt~\cite{Ned92}, and
Milner~\cite{DBLP:journals/entcs/Milner07}. A peculiar feature of the
linear substitution calculus (\lsc) is the use of rewriting rules
\emph{at a distance}, \ie\ rules defined by means of contexts, that
are used to closely mimic reduction in linear logic proof nets. Such a
framework---whose use does not require any knowledge of these
areas---allows an easy management of sharing and, in contrast to
previous approaches to ES, admits a theory of standardization and a
notion of \lo\ evaluation~\cite{non-standard-preprint}. The proof of
our result is based on a fine quantitative study of the
relationship between \lo\ derivations for the $\l$-calculus and a
variation over \lo\ derivations for the \lsc. Roughly, the latter
avoids the size-explosion problem while keeping a polynomial
relationship with the former.

Invariance results usually have two directions,
while we here study only one of them, namely that the $\l$-calculus
can be efficiently simulated by, say, Turing machines. The missing
half is a much simpler problem already solved in
\cite{DBLP:conf/rta/AccattoliL12}: there is an encoding of Turing
machines into $\l$-terms such that their execution is simulated by weak
head $\beta$-reduction with only a linear overhead. \ugo{The result
on tree Turing machines from~\cite{DBLP:conf/fpca/FrandsenS91} is 
not immediately applicable here, being formulated on a different,
more parsimonious, cost model.}

\emph{On Invariance, Complexity Analysis, and Some Technical Choices.}  
Before proceeding, let us stress some crucial points:
\begin{varenumerate}
\item 
  \emph{ES Are Only a Tool}. Although ES are an \emph{essential
  tool} for the proof of our result, the \emph{result itself} is
  about the usual, pure, $\lambda$-calculus. In particular, the
  invariance result can be used without any need to care about ES: we
  are allowed to measure the complexity of problems by simply bounding
  the \emph{number} of \lo\ $\beta$-steps taken by any $\lambda$-term
  solving the problem.
\item 
  \emph{Complexity Classes in the $\l$-Calculus}. The main
  consequence of our invariance result is that every polynomial or
  super-polynomial class, like $\mathbf{P}$ or $\mathbf{EXP}$, can be
  defined using $\l$-calculus (and \lo\ $\beta$-reduction) instead of
  Turing machines.
\item 
  \emph{Our Cost Model is Unitary}. An important point is that our
  cost model is \emph{unitary}, and thus attributes a constant cost to
  any \lo\ step. One could argue that it is always possible to reduce
  $\lambda$-terms on abstract or concrete machines and take that
  number of steps as \emph{the} cost model. First, such a measure of
  complexity would be very machine-dependent, against the very essence
  of $\l$-calculus. Second, these cost models invariably attribute a
  more-than-constant cost to any $\beta$-step, making the measure much
  harder to use and analyze. It is not evident that a computational
  model enjoys a unitary invariant cost model. As an example, if
  multiplication is a primitive operation, random access machines need
  to be endowed with a \emph{logarithmic} cost model in order to
  obtain invariance.
\item \ben{
  \emph{LSC vs. Graph-Reduction vs. Abstract Machines}. The LSC has
  been designed as a graph-free formulation of the representation of
  $\lambda$-calculus into linear logic proof nets. As such, it can be
  seen as equivalent to a graph-rewriting formalism. While employing
  graphs may slightly help in presenting some intuitions, terms are
  much simpler to define, manipulate, and formally reason about. In
  particular the detailed technical development we provide would
  simply be out of scope if we were using a graphical
  formalism. Abstract machines are yet another formalism that could
  have been employed, that also has a tight relationship with the LSC,
  as shown by Accattoli, Berenbaum, and Mazza in
  \cite{DBLP:conf/icfp/AccattoliBM14}. We chose to work with the LSC
  because it is more abstract than abstract machines and both more apt
  to formal reasoning and closer to the $\l$-calculus (no translation
  is required) than graph-rewriting.}
 \item \ben{\emph{Proof Strategy}. While the main focus of the paper
   is the invariance result, we also spend much time providing an
   abstract decomposition of the problem and a more general rewriting
   analysis of the LSC and of useful sharing. Therefore, the proof
   presented here is not the simplest possible one. We believe,
   however, that our study is considerably more informative than the
   shortest proof.}
\end{varenumerate}
The next section explains why the problem at hand is hard, and in
particular why iterating our previous results on head reduction
\cite{DBLP:conf/rta/AccattoliL12} does not provide a solution.\medskip

\emph{Related Work}. In the literature invariance results for the weak
call-by-value $\l$-calculus have been proved three times,
independently. First, by Blelloch and Greiner
\cite{DBLP:conf/fpca/BlellochG95}, while studying cost models for
parallel evaluation. Then by Sands, Gustavsson and Moran
\cite{DBLP:conf/birthday/SandsGM02}, while studying speedups for
functional languages, and finally by Martini and the second author
\cite{DBLP:journals/tcs/LagoM08}, who addressed the invariance thesis
for the $\l$-calculus. The latter also proved invariance for the weak
call-by-name $\l$-calculus
\cite{DBLP:journals/corr/abs-1208-0515}. Invariance of head reduction
has been shown by the present authors, in previous work
\cite{DBLP:conf/rta/AccattoliL12}. \ugo{The problem of an invariant cost
model for the ordinary $\l$-calculus is discussed by Frandsen and
Sturtivant \cite{DBLP:conf/fpca/FrandsenS91}, and then by Lawall and
Mairson \cite{DBLP:conf/icfp/LawallM96}. Frandsen and Sturtivant's
proposal consists in taking the number of \emph{parallel} $\beta$-steps
to normal form as the cost of reducing any term. A negative result about the
invariance of such cost model has been proved by Asperti and
Mairson \cite{DBLP:conf/popl/AspertiM98}.} When only first order
symbols are considered, Dal Lago and Martini, and independently
Avanzini and Moser, proved some quite general results through graph
rewriting~\cite{DBLP:journals/corr/abs-1208-0515,DBLP:conf/rta/AvanziniM10},
itself a form of sharing.\medskip

This paper is a revised and extended version of
\cite{DBLP:conf/csl/AccattoliL14}, to which it adds explanations and
the proofs that were omitted. It differs considerably with respect to
both \cite{DBLP:conf/csl/AccattoliL14} and the associated technical
report \cite{EV}, as proofs and definitions have been improved and
simplified, partially building on the recent work by Accattoli and
Sacerdoti Coen in \cite{usef-constr}, where useful sharing is studied
in a call-by-value scenario.

\setcounter{tocdepth}{1}
After the introduction, in \refsect{hard} we explain why the problem
is hard by discussing the size-explosion problem. An abstract view of
the solution is given in \refsect{abstract-proof}. The sections in
between (2-6) provide the background, \ie\ definitions and basic
results, up to the introduction of useful reduction---at a first
reading we suggest to skip them. After the abstract view, in
\refsect{road-map} we explain how the various abstract requirements
are actually proved in the remaining sections (9-14), where the proofs
are. We put everything together in \refsect{summing-up}, and discuss
optimizations in \refsect{discussion}.

\section{Why is The Problem Hard?}
\label{sect:hard}
In principle, one may wonder why sharing is needed at all, or whether
a relatively simple form of sharing suffices. In this section, we will
show that sharing is unavoidable and that a new subtle notion of
sharing is necessary.

If we stick to explicit representations of terms, in which sharing is
not allowed, counterexamples to invariance can be designed in a fairly
easy way.  The problem is \emph{size-explosion}, or the existence of
terms of size $n$ that in $O(n)$ steps produce an output of size
$O(2^{n})$, and affects the $\l$-calculus as well as its weak and head
variants. The explosion is due to iterated \emph{useless} duplications
of subterms that are normal and whose substitution does not create new
redexes. For simple cases as weak or head reduction, turning to shared
representations of $\l$-terms and micro-step substitutions (\ie\ one
occurrence at the time) is enough to avoid size-explosion. For
micro-steps, in fact, the length of evaluation and the size of the output are
linearly related. A key point is that both micro-step weak and head
reduction stop on a compact representation of the weak or head normal
form.

In the ordinary $\l$-calculus, a very natural notion of evaluation to
normal form is \lo\ reduction. Unfortunately, turning to sharing and
micro-step \lo\ evaluation is not enough, because such a micro-step
simulation of $\beta$-reduction computes ordinary normal forms,
\ie\ it does not produce a compact representation, but the usual one,
whose size is sometimes exponential. In other words, size-explosion
reappears disguised as \emph{length-explosion}: for the size-exploding
family, indeed, micro-step evaluation to normal form is necessarily
exponential in $n$, because its length is linear in the size of the
output. Thus, the number of $\beta$-steps cannot be shown to be
invariant using such a simple form of sharing.

The problem is that evaluation should stop on a
compact---\ie\ not exponential---representation of the normal form, as
in the simpler cases, but there is no such notion. Our way out is the
definition of a variant of micro-step \lo\ evaluation that stops on a
minimal \emph{useful normal form}, that is a term with ES $\tm$
such that unfolding all the substitutions in $\tm$ produces a normal form,
\ie\ such that the duplications left to do are useless. In
\refsect{useful}, we will define \emph{useful} reduction, that will
stop on minimal useful normal forms and for which we will show
invariance with respect to both $\beta$-reduction and random access
machines.

In the rest of the section we discuss in detail the size-explosion
problem, recall the solution for the head case, and explain the
problem for the general case. Last, we discuss the role of standard
derivations.\medskip

\subsection{A Size-Exploding Family} The typical example of a term
that is useless to duplicate is a free variable\footnote{\emph{On open
    terms}: in the $\l$-calculus free variables are unavoidable
  because reduction takes place under abstractions. Even if one
  considers only globally closed terms, variable occurrences may look
  free \emph{locally}, as $\vartwo$ in
  $\l\vartwo.((\l\var.(\var\var))\vartwo)\tob
  \l\vartwo.(\vartwo\vartwo)$. This is why for studying the strong
  $\l$-calculus it is common practice to work with possibly open
  terms.  }, as it is normal and its substitution cannot create
redexes. Note that the same is true for the application $\var\var$ of
a free variable $\var$ to itself, and, iterating, for
$(\var\var)(\var\var)$, and so on. We can easily build a term
$\tmtwo_n$ of size $\size{\tmtwo_n} = O(n)$ that takes a free variable
$\var$, and puts its self application $\var\var$ as argument of a new
redex, that does the same, \ie\ it puts the self application
$(\var\var)(\var\var)$ as argument of a new redex, and so on, for $n$
times normalizing in $n$ steps to a complete binary tree of height $n$
and size $O(2^{n})$, whose internal nodes are applications and whose
$2^n$ leaves are all occurrences of $\var$. Let us formalize this
notion of \emph{variable tree} $\var^{@n}$ of height $n$:
\begin{center}$
\begin{array}{lllllllllllllll}
	\var^{@0} & \defeq & \var;\\
	\var^{@(n+1)} & \defeq & \var^{@n} \var^{@n}.
\end{array}$
\end{center}
Clearly, the size of variable trees is exponential in $n$, a routine
induction indeed shows $\size{\var^{@n}} = 2^{n+1}-1 = O(2^{n})$.  Now
let us define the family of terms $\{\tm_n\}_{n\geq 1}$ that in only
$n$ \lo\ steps blows up $\var $ into the tree $\var^{@n}$ :
\begin{center}
$\begin{array}{lllllllllllllll}
	\tm_1 & \defeq & \la{\vartwo_{1}}(\vartwo_{1}\vartwo_{1}) & = & \la{\vartwo_{1}}\vartwo_{1}^{@1};\\
	\tm_{n+1} & \defeq & \la{\vartwo_{n+1}}(\tm_{n}(\vartwo_{n+1}\vartwo_{n+1})) & = & \la{\vartwo_{n+1}}(\tm_{n}\vartwo_{n+1}^{@1}).
\end{array}$
\end{center}
Note that the size $\size{\tm_n}$ of $ \tm_n$ is
$O(n)$. Leftmost-outermost $\beta$-reduction is noted $\toblo$. The
next proposition proves size-explosion, \ie\ $\tm_n \var = \tm_n
\var^{@0} \toblo^n \var^{@n}$ (with $\size{\tm_n\var} = O(n)$ and
$\size{\var^{@n}} = O(2^{n})$ giving the explosion). The statement is
slightly generalized, in order to express it as a nice property over
variable trees.

\begin{proposition}[Size-Explosion]
  \label{prop:size-exp}
	$\tm_n \var^{@m} \toblo^n \var^{@(n+m)}$.
\end{proposition}

\proof
  By induction on $n$. Cases:
  \begin{varenumerate}
  \item \emph{Base Case}, \ie\ $n= 1$:
    $$\begin{array}{lllllllll}
    \tm_1 \var^{@m}& = & (\la{\vartwo_{1}}(\vartwo_{1}\vartwo_{1})) \var^{@m} &	\toblo & \var^{@m}\var^{@m} & = & \var^{@(m+1)}.
  \end{array}$$
    
  \item \emph{Induction Step}: 
    $$\begin{array}{llllllllllllll}			
    \tm_{n+1} \var^{@m}& = & (\la{\vartwo_{n+1}}(\tm_{n}(\vartwo_{n+1}\vartwo_{n+1}))) \var^{@m} & \toblo
    & \tm_n (\var^{@m}\var^{@m}) \\
    &&&= &\tm_n \var^{@(m+1)}  \\
    &&&\toblo^n& \var^{@(n+m+1)} & \mbox{(\ih)}.\rlap{\hbox to 44 pt{\hfill\qEd}}
  \end{array}$$
  \end{varenumerate}
\medskip

\noindent It seems that the unitary cost model---\ie\ the number of \lo\ $\beta$-steps---\emph{is not} invariant:
in a linear number of $\beta$-steps we reach an object which cannot
even be written down in polynomial time.\medskip

\subsection{The Head Case} 
The solution the authors proposed in~\cite{DBLP:conf/rta/AccattoliL12}
tames size-explosion in a satisfactory way when \emph{head} reduction
is the evaluation strategy (note that $\beta$-steps in
\refprop{size-exp} are in particular head steps). It uses sharing
under the form of explicit substitutions (ES), that amounts to extend
the language with an additional constructor noted
$\tm\esub\var\tmtwo$, that is an avatar of $\tt let$-expressions, to
be thought of as a sharing annotation of $\tmtwo$ for $\var$ in
$\tm$,\ben{ or as a term notation for the DAGs used in the
  graph-rewriting of $\l$-calculus (see
  \cite{DBLP:conf/csl/AccattoliG09})}. The usual, capture-avoiding,
and meta-level notion of substitution is instead noted
$\tm\isub\var\tmtwo$.

Let us give a sketch of how ES work for the head case. Formal details
about ES and the more general \lo\ case will be given in the
\refsect{shallow-calculus}. First of all, a term with sharing,
\ie\ with ES, can always be unshared, or unfolded, obtaining an
ordinary $\l$-term $\unf\tm$.

\begin{definition}[Unfolding]
\label{def:unfolding}
The \emph{unfolding}  $\unf\tm$ of a term with ES $\tm$ is given by:
\begin{align*}
  \unf{\var}&\defeq\var; &
  \unf{(\ap\tm \tmtwo)}&  \defeq \ap{\unf{\tm}} {\unf{\tmtwo}};\\
  \unf{(\l \var. \tm)} & \defeq  \l \var. \unf{\tm};&
  \unf{(\tm\esub{\var}{\tmtwo})} & \defeq  \unf{\tm}\isub{\var}{\unf{\tmtwo}}.
\end{align*} 
\end{definition}\medskip

\noindent Head $\beta$-reduction is $\beta$-reduction in a head context, \ie\
out of all arguments, possibly under abstraction (and thus involving
open terms). A head step $(\la\var\tm)\tmtwo \tob \tm\isub\var\tmtwo$
is simulated by
\begin{varenumerate}
\item
  \emph{Delaying Substitutions}: the substitution $\isub\var\tmtwo$ is
  delayed with a rule $(\la\var\tm)\tmtwo \todb \tm\esub\var\tmtwo$
  that introduces an explicit substitution. The name $\db$ stays for
  \emph{distant $\beta$} or \emph{$\beta$ at a distance}, actually
  denoting a slightly more general rule to be discussed in the next
  section\footnote{\ben{A more accurate explanation of the
      terminology: in the literature on ES the rewriting rule
      $(\la\var\tm)\tmtwo \to \tm\esub\var\tmtwo$ (that is the
      explicit variant of $\beta$) is often called $\tt B$ to
      distinguish it from $\beta$, and $\db$---that will be formally
      defined in \refsect{shallow-calculus}---stays for \emph{distant
        $\tt B$} (or \emph{$\tt B$ at a distance}) rather than
      \emph{distant $\beta$}.}}.
\item
  \emph{Linear Head Substitutions}: linear substitution $\tols$
  replaces a single variable occurrence with the associated shared
  subterm. Linear \emph{head} substitution $\tolhs$ is the variant
  that replaces only the head variable, for instance
  $(\var(\vartwo\var))\esub\var\tm\esub\vartwo\tmtwo \tolhs
  (\tm(\vartwo\var))\esub\var\tm\esub\vartwo\tmtwo$. Linear
  substitution can be seen as a reformulation with ES of de Bruijn's
  \emph{local $\beta$-reduction} \cite{deBruijn87}, and similarly its
  head variant is a reformulation of Danos and Regnier's presentation
  of linear head reduction \cite{Danos04headlinear}.
\end{varenumerate}\medskip
In particular, the size-exploding family $\tm_n \var$ is evaluated by
the following \emph{linear} head steps. For $n=1$ we have
$$\begin{array}{lllllllll}
  \tm_1\var & = & (\la{\vartwo_{1}}(\vartwo_{1}\vartwo_{1})) \var &	\todb & (\vartwo_{1}\vartwo_{1})\esub{\vartwo_{1}}\var & \tolhs & (\var\vartwo_{1})\esub{\vartwo_{1}}\var\\
\end{array}$$
Note that only the head variable has been replaced, and that evaluation requires one $\todb$ step and one $\tolhs$ step. For $n=2$, 
$$\begin{array}{lllllllll}
  \tm_2\var & = & (\la{\vartwo_2}((\la{\vartwo_{1}}(\vartwo_{1}\vartwo_{1})) (\vartwo_2\vartwo_2)))\var \\
  & \todb & ((\la{\vartwo_{1}}(\vartwo_{1}\vartwo_{1})) (\vartwo_2\vartwo_2))\esub{\vartwo_{2}}\var \\
  & \todb & (\vartwo_{1}\vartwo_{1})\esub{\vartwo_1}{\vartwo_2\vartwo_2}\esub{\vartwo_{2}}\var \\
  &\tolhs & ((\vartwo_2\vartwo_2)\vartwo_{1})\esub{\vartwo_1}{\vartwo_2\vartwo_2}\esub{\vartwo_{2}}\var \\
  &\tolhs & ((\var\vartwo_2)\vartwo_{1})\esub{\vartwo_1}{\vartwo_2\vartwo_2}\esub{\vartwo_{2}}\var. 
\end{array}$$
It is easily seen that 
$$
\tm_n\var \todb^n\tolhs^n ((\ldots(\var\vartwo_n)\ldots\vartwo_{2})\vartwo_{1})\esub{\vartwo_1}{\vartwo_2\vartwo_2}\esub{\vartwo_{2}}{\vartwo_3\vartwo_3}\ldots\esub{\vartwo_n}\var \ \ \ =: \tmthree_n
$$
As one can easily verify, the size of the linear head normal form
$\tmthree_n$ is linear in $n$, so that there is no size-explosion (the
number of steps is also linear in $n$). Moreover, the unfolding
$\unf{\tmthree_n}$ of $\tmthree_n$ is exactly $\var^{@n}$, so that the
linear head normal form $\tmthree_n$ is a compact representation of
the head normal form, \ie\ the expected result. Morally, in
$\tmthree_n$ only the left branch of the complete binary tree
$\var^{@n}$ has been unfolded, while the rest of the tree is kept
shared via explicit substitutions. Size-explosion is avoided by not
substituting in arguments at all.

Invariance of head reduction via LHR is obtained
in~\cite{DBLP:conf/rta/AccattoliL12} by proving that LHR correctly
implements head reduction up to unfolding within---crucially---a
quadratic overhead. This is how sharing is exploited to circumvent the
size-explosion problem: the length of head derivations is a reasonable
cost model even if head reduction suffers of size-explosion, because
the actual implementation is meant to be done via LHR and be only
polynomially (actually quadratically) longer. Note that---a
posteriori---we are allowed to forget about ES. They are an essential
tool for the proof of invariance. But once invariance is established,
one can provide reasonable complexity bounds by simply counting
$\beta$-steps in the $\l$-calculus, with no need to deal with ES.

Of course, one needs to show that turning to shared representations is
a reasonable choice, \ie\ that using a term with ES outside the
evaluation process does not hide an exponential overhead. Shared terms
can in fact be managed efficiently, typically tested for
\emph{equality of their unfoldings} in time polynomial (actually
quadratic \cite{DBLP:conf/rta/AccattoliL12}, or quasi-linear
\cite{DBLP:conf/icfp/GrabmayerR14}) in the size of the \emph{shared
  terms}. In \refsect{properties}, we will discuss another kind of
test on shared representations.\medskip

\subsection{Length-Explosion and Usefulness} 
It is clear that the computation of the full normal form $\var^{@n}$
of $\tm_n\var$, requires exponential work, so that the general case
seems to be hopeless. In fact, there is a notion of linear
\lo\ reduction $\tolo$ \cite{non-standard-preprint}, obtained by
iterating LHR on the arguments, that computes normal forms and it is
linearly related to the size of the output. However, $\tolo$ cannot be
polynomially related to the \lo\ strategy $\toblo$, because it
produces an exponential output, and so it necessarily takes an
exponential number of steps. In other words, size-explosion disguises
itself as \emph{length-explosion}. With respect to our example,
$\tolo$ extends LHR evaluation by unfolding the whole variable tree in
a \lo\ way,
$$\tm_n\var \todb^n\tols^{O(2^{n})} \var^{@n}\esub{\vartwo_1}{\vartwo_2\vartwo_2}\esub{\vartwo_{2}}{\vartwo_3\vartwo_3}\ldots\esub{\vartwo_n}\var$$
and leaving garbage $\esub{\vartwo_1}{\vartwo_2\vartwo_2}\esub{\vartwo_{2}}{\vartwo_3\vartwo_3}\ldots\esub{\vartwo_n}\var$ that may eventually be collected. Note the exponential number of steps.

Getting out of this \emph{cul-de-sac} requires to avoid useless
duplication. \ben{Essentially, only substitution steps that contribute
  to eventually obtain an unshared $\beta$-redex have to be done. The
  other substitution steps, that only unfold parts of the normal form,
  have to be avoided. Such a process then produces a minimal shared
  term whose unfolding is an ordinary normal form.  The tricky point
  is how to define, and then select in reasonable time, those steps
  that \emph{contribute to eventually obtain an unshared
    $\beta$-redex}. The definition of useful reduction relies on tests
  of certain partial unfoldings that have a inherent global nature,
  what in a graphical formalism can be thought of as the unfolding of
  the sub-DAG rooted in a given sharing node. Of course, computing
  unfoldings takes in general exponential time, so that an efficient
  way of performing such tests has to be found.}

The proper definition of useful reduction is postponed to
\refsect{useful}, but we discuss here how it circumvents
size-explosion. With respect to the example, useful reduction
evaluates $\tm_n\var$ to the useful normal form
$$
(\vartwo_1\vartwo_1)\esub{\vartwo_1}{\vartwo_2\vartwo_2}\esub{\vartwo_{2}}{\vartwo_3\vartwo_3}\ldots\esub{\vartwo_n}\var,
$$
that unfolds to the exponentially bigger result $\var^{@n}$. In
particular, our example of size-exploding family will be evaluated without
performing \emph{any duplication at all}, because the duplications
needed to compute the normal form are all useless.

Defining and reasoning about useful reduction requires some care. At
first sight, one may think that it is enough to evaluate a term $\tm$
in a \lo\ way, stopping as soon as a useful normal form is
reached. Unfortunately, this simple approach does not work, because
size-explosion may be caused by ES \emph{lying in between} two
$\beta$-redexes, so that \lo\ evaluation would unfold the exploding
substitutions anyway.

Moreover, it is not possible to simply define useless terms and avoid
their reduction. The reason is that usefulness and uselessness are
properties of substitution steps, not of subterms. Said differently,
whether a subterm is useful depends crucially on the context in which
it occurs. An apparently useless argument may become useful if plugged
into the right context. Indeed, consider the term $\tmtwo_n \defeq
(\la\var(\tm_n\var)) I$, obtained by plugging the size-exploding
family in the context $(\la\var\ctxhole) I$, that abstracts $\var$ and
applies to the identity $I \defeq \la\varthree\varthree$. By delaying
$\beta$-redexes we obtain:
$$\tmtwo_n \todb^{n+1} (\vartwo_1\vartwo_1)\esub{\vartwo_1}{\vartwo_2\vartwo_2}\esub{\vartwo_{2}}{\vartwo_3\vartwo_3}\ldots\esub{\vartwo_n}\var\esub\var{I}$$
Now---in contrast to the size-explosion case---it is useful to unfold
the whole variable tree $\var^{@n}$, because the obtained copies of
$\var$ will be substituted by $I$, generating exponentially many
$\beta$ steps, that compensate the explosion in size.
Our notion of \emph{useful} step will elaborate on this idea,
by computing \emph{contextual} unfoldings, to check if a substitution
step contributes (or will contribute) to some future $\beta$-redex. Of
course, we will have to show that such tests can be themselves
performed in polynomial time. 

\ben{It is also worth mentioning that the contextual nature of useful
  substitution implies that---as a rewriting rule---it is inherently
  \emph{global}: it cannot be first defined at top level
  (\ie\ locally) and then extended via a closure by evaluation
  contexts, because the evaluation context has to be taken into
  account in the definition of the rule itself. Therefore, the study
  of useful reduction is delicate at the technical level, as proofs by
  na\"ive induction on evaluation contexts usually do not
  work.}\medskip

\subsection{The Role of Standard Derivations.} 
Apart from the main result, we also connect the classic rewriting
concept of standard derivation with the problem under study. Let us
stress that such a connection is a \emph{plus}, as it is not needed to
prove the invariance theorem. We use it in the proof, but only to shed
a new light on a well-established rewriting concept, and not because
it is necessary.

The role of standard derivations is in fact twofold. On the one hand,
\lo\ $\beta$-derivations are standard, and thus our invariant cost
model is justified by a classic notion of evaluation, \emph{internal}
to the theory of the $\l$-calculus and not \emph{ad-hoc}. On the other
hand, the linear useful strategy is shown to be standard for the
LSC. Therefore, this notion, at first defined \emph{ad-hoc} to solve
the problem, turns out to fit the theory.

The paper contains also a general result about standard derivations
for the LSC. We show they have the \emph{subterm property}, \ie\ every
single step of a standard derivation $\deriv:\tm\to^*\tmtwo$ is
implementable in time linear in the size $\size\tm$ of the input. It
follows that the size of the output is linear in the length of the
derivation, and so there is no size-explosion. Such a connection
between standardization and complexity analysis is quite surprising,
and it is one of the signs that a new complexity-aware rewriting
theory of $\beta$-reduction is emerging.\medskip

At a first reading, we suggest to read \refsect{abstract-proof}, where
an abstract view of the solution is provided, right after this
section. In between (\ie\ sections 2-6), there is the necessary long
sequence of preliminary definitions and results. In particular,
\refsect{useful}, will define useful reduction.

\section{Rewriting}
For us, an \emph{(abstract) reduction system} is a pair $(T, \str)$ consisting of
a set $T$ of terms and a binary relation $\str$ on $T$ called a
\emph{reduction (relation)}. When $(\tm,\tmtwo) \in \str$ we write
$\tm \str \tmtwo$ and we say that $\tm$ \emph{$\R$-reduces} to
$\tmtwo$. The reflexive and transitive closure of $\str$ is written
$\Rewn{\R}$. Composition of relations is denoted by juxtaposition.
Given $k \geq 0$, we write $a \Rew{\R}^k b$ iff $a$ is $\R$-related to
$b$ in $k$ steps, \ie\ $a \Rew{\R}^0 b$ if $a=b$ and $a \Rew{\R}^{k+1}
b$ if $\exists\ c$ such that  $a \str c$ and $c \Rew{\R}^k b$.

A term $\tm \in R$ is a \emph{$\R$-normal form} if there is no
$\tmtwo\in R$ such that $\tm\str \tmtwo$. Given a deterministic
reduction system $(T, \str)$, and a term $\tm \in T$, the expression
$\nos{\str}{\tm}$ stands for the number of reduction steps necessary
to reach the $\str$-normal form of $\tm$ along $\str$, or $\infty$ if
 $\tm$ diverges. Similarly, given a natural number $n$, the expression
$\strp{\str}{n}{\tm}$ stands for the term $\tmtwo$ such that
$\tm\str^n\tmtwo$, if $n\leq\nos{\str}{\tm}$, or for the normal form
of $\tm$ otherwise.

\section{\texorpdfstring{$\l$}{lambda}-Calculus}
\subsection{Statics}
The syntax of the $\l$-calculus is given by the following grammar for
terms:
$$
\tm,\tmtwo,\tmthree,\tmfour\grameq\var\midd \l \var. \tm \midd \tm \tmtwo.
$$ 
We use $\tm\isub{\var}{\tmtwo}$ for the usual (meta-level)
notion of substitution. An abstraction $\l
\var. \tm$ binds $\var$ in $\tm$, and we
silently work modulo $\alpha$-equivalence of these bound variables,
\eg\ $(\la\vartwo(\var\vartwo))\isub\var\vartwo =
\la\varthree(\vartwo\varthree)$. We use $\fv\tm$ for the set of
free variables of $\tm$.

\emph{Contexts.} One-hole contexts are defined by:
$$
\gctx \grameq \ctxhole\midd \l \var. \gctx\midd \ap\gctx \tm \midd\ap\tm\gctx,
$$ 
and the plugging of a term $\tm$ into a context $\gctx$ is defined by
\[\begin{array}{cccccccccccc}
   \ctxhole\ctxholep\tm & \defeq & \tm &&& (\gctx \tmtwo)\ctxholep\tm & \defeq & \gctx\ctxholep\tm \tmtwo\\
   (\l\var.\gctx)\ctxholep\tm & \defeq & \l\var.\gctx\ctxholep\tm &&& (\tmtwo \gctx)\ctxholep\tm & \defeq &  \tmtwo \gctx\ctxholep\tm
  \end{array}\]
As usual, plugging in a context can capture variables,
\eg\ $(\la\vartwo(\ctxhole\vartwo))\ctxholep\vartwo =
\la\vartwo(\vartwo\vartwo)$. The plugging $\gctxp\gctxtwo$ of a
context $\gctxtwo$ into a context $\gctx$ is defined analogously. Plugging will be implicitly extended to all notions of contexts in the paper, always in the expected way.

\subsection{Dynamics}
We define $\beta$-reduction $\tob$ as follows:
$$\begin{array}{c@{\hspace{1.5cm}}c}
  \textsc{Rule at Top Level} & \textsc{Contextual closure} \\
	(\l\var.\tm)\tmtwo \rtob \tm\isub\var\tmtwo &
        \gctxp \tm \tob \gctxp \tmtwo \textrm{~~~if } \tm \rtob \tmtwo \\
\end{array}$$
The \emph{position} of a $\beta$-redex $\gctxp \tm \tob \gctxp \tmtwo$ is the context $\gctx$ in which it takes place. To ease the language, we will identify a redex with its position. A \deff{derivation} $\deriv:\tm\to^k\tmtwo$
is a finite, possibly empty, sequence of reduction steps, sometimes given as
$\gctx_1;\ldots; \gctx_k$, \ie\ as the sequence of positions of reduced
redexes. We write $\size\tm$ for the size of $\tm$ and $\size\deriv$ for the length of
$\deriv$. \medskip

\emph{Leftmost-Outermost Derivations.} 
The left-to-right outside-in order on redexes is expressed as an order
on positions, \ie\ contexts. Let us warn the reader about a possible
source of confusion. The \emph{left-to-right outside-in} order in the
next definition is sometimes simply called \emph{left-to-right} (or
simply \emph{left}) order. The former terminology is used when terms
are seen as trees (where the left-to-right and the outside-in orders
are disjoint), while the latter terminology is used when terms are
seen as strings (where left-to-right is a total order). While the
study of standardization for the \lsc~\cite{non-standard-preprint}
uses the string approach (and thus only talks about the
\emph{left-to-right} order and the \emph{leftmost} redex), here some
of the proofs require
a delicate analysis of the relative positions of redexes and so we
prefer the more informative tree approach and define the order
formally.

\begin{definition}\hfill
  \begin{varenumerate}
  \item 
    The \deff{outside-in order}: 
    \begin{varvarenumerate}
    \item 
      \emph{Root}: $\ctxhole\outin\gctx$ for every context
      $\gctx\neq\ctxhole$;
    \item 
      \emph{Contextual closure}: If $\gctx\outin\gctxtwo$ then
      $\gctxthreep\gctx\outin\gctxthreep\gctxtwo$ for any context
      $\gctxthree$.
    \end{varvarenumerate}    
  \item 
    The \deff{left-to-right order}: $\gctx\leftright\gctxtwo$ is defined by:
    \begin{varvarenumerate}
    \item 
      \emph{Application}: If $\gctx\prefix \tm$ and
      $\gctxtwo\prefix\tmtwo$ then
      $\ap\gctx\tmtwo\leftright\ap\tm\gctxtwo$;
    \item 
      \emph{Contextual closure}: If $\gctx\leftright\gctxtwo$ then
      $\gctxthreep\gctx\leftright\gctxthreep\gctxtwo$ for any context
      $\gctxthree$.
    \end{varvarenumerate}  
  \item 
    The \deff{left-to-right outside-in order}:
    $\gctx\leftout\gctxtwo$ if $\gctx\outin\gctxtwo$ or
    $\gctx\leftright\gctxtwo$:
  \end{varenumerate}
\end{definition}

The following are a few examples. For every context $\gctx$, it holds
that $\ctxhole\not\leftright\gctx$. Moreover,
\begin{align*}
  \ap{(\l\var.\ctxhole)}\tm
  &\outin 
  \ap{(\l\var.(\ap{\ctxhole}\tmtwo)})\tm;\\
  \ap{(\ap\ctxhole\tm)}\tmtwo
  &\leftright
  \ap{(\ap\tmthree\tm)}\ctxhole.
\end{align*}

\begin{definition}[\lo\ $\beta$-Reduction]
  Let $\tm$ be a $\l$-term and $\gctx$ a redex of $\tm$. $\gctx$ is the \deff{leftmost-outermost}  $\beta$-redex (\lob\ for short) of $\tm$ 
   if $\gctx\leftout\gctxtwo$ for
  every other $\beta$-redex $\gctxtwo$ of $\tm$. We
  write $\tm\toblo\tmtwo$ if a step reduces
  the \lo\ $\beta$-redex.
\end{definition}

\subsection{Inductive \lob\ Contexts} It is useful to have an inductive characterization of the contexts in which $\toblo$ takes place. We use the following terminology: a term is \emph{neutral} if it is $\beta$-normal and it is not of the form $\la\var\tmtwo$, \ie\ it is not an abstraction. 
 
 \begin{definition}[\ilob\ Context]
 \label{def:ilob-ctx}
 Inductive \lob\ (or \ilob) contexts are defined by induction as follows:
 \begin{center}
$\begin{array}{cccccc}
	\AxiomC{}
	\RightLabel{(ax-\ilob)}
	\UnaryInfC{$\ctxhole$ is \ilob}
	\DisplayProof
	&
	\AxiomC{$\gctx$ is \ilob}
	\AxiomC{$\gctx\neq\la\var\gctxtwo$}
	\RightLabel{(@l-\ilob)}
	\BinaryInfC{$\gctx \tm$ is \ilob}
	\DisplayProof \\\\
	
	\AxiomC{$\gctx$ is \ilob}
	\RightLabel{($\l$-\ilob)}
	\UnaryInfC{$\la\var\gctx$ is \ilob}
	\DisplayProof &

	  \AxiomC{$\tm$ is neutral}
	\AxiomC{$\gctx$ is \ilob}
	\RightLabel{(@r-\ilob)}
	\BinaryInfC{$\tm \gctx$ is \ilob}
	\DisplayProof 
\end{array}$
\end{center}
 \end{definition}

As expected,

\begin{lemma}[$\toblo$-steps are Inductively \lob]
\label{l:toblo-ilob}
Let $\tm$ be a $\l$-term and $\gctx$ a redex in $\tm$. $\gctx$ is the \lob\ redex in $\tm$ iff $\gctx$ is \ilob.
\end{lemma}

\begin{proof}
  The left-to-right implication is by induction on $\gctx$.  The
  right-to-left implication is by induction on the definition of
  \emph{$\gctx$ is \ilob}.
\end{proof}

\section{The Shallow Linear Substitution Calculus}
\label{sect:shallow-calculus}
\subsection{Statics}
The language of the \emph{linear substitution
  calculus} (\lsc\ for short) is given by the following grammar for
terms:
$$
\tm,\tmtwo,\tmthree,\tmfour\grameq\var\midd \l \var. \tm \midd \ap\tm \tmtwo\midd  \tm\esub\var\tmtwo.
$$ 
The constructor $\tm\esub{\var}{\tmtwo}$ is called an \emph{explicit
  substitution} (of $\tmtwo$ for $\var$ in $\tm$). Both $\l
\var. \tm$ and $\tm\esub{\var}{\tmtwo}$ bind $\var$ in $\tm$. In general, we assume a strong form of Barendregt's convention: any two bound or free variables have distinct names. We also
silently work modulo $\alpha$-equivalence of bound variables to preserve the convention,
\eg\ $(\var\vartwo)\esub\vartwo\tm\isub\var\vartwo =
(\vartwo\varthree)\esub\varthree{\tm\isub\var\vartwo}$ and $(\var\vartwo)\esub\vartwo\tm\isub\vartwo\tmtwo = (\var\varthree)\esub\varthree{\tm\isub\vartwo\tmtwo}$ where $\varthree$ is fresh. 

The operational semantics of the \lsc\ is parametric
in a notion of (one-hole) context. General contexts simply extend the contexts for $\l$-terms with the two cases for explicit substitutions:
$$
\gctx \grameq \ctxhole\midd \l \var. \gctx\midd \ap\gctx \tm \midd\ap\tm\gctx\midd\gctx\esub{\var}{\tm}\midd\tm\esub{\var}{\gctx},
$$ 
Along most of the paper, however, we will not need such a general
notion of context. In fact, our study takes a simpler form if the
operational semantics is defined with respect to \deff{shallow}
contexts, defined as (note the absence of the production
$\tm\esub{\var}{\ctx}$):
$$
\ctx,\ctxtwo,\ctxthree,\ctxfour \grameq \ctxhole\midd \l \var. \ctx\midd \ap\ctx \tm \midd\ap\tm\ctx\midd\ctx\esub{\var}{\tm}.
$$
In the following, whenever we refer to a \emph{context} without
further specification, it is implicitly assumed that it is a
\emph{shallow} context. We write $\ctx\prefix\tm$ if
there is a term $\tmtwo$ such that $\ctxp\tmtwo=\tm$, and call it the
\deff{prefix relation}.

A special class of contexts is that of
\deff{substitution contexts}: 
$$
\sctx\grameq\ctxhole\midd
\sctx\esub{\var}{\tm}.
$$
\ben{
\begin{remark}[$\alpha$-Equivalence for Contexts]
\label{rem:alpha-for-contexts}
While Barendregt's convention can always be achieved for terms, for
contexts the question is subtler. Plugging in a context $\ctx$,
indeed, is \emph{not} a capture-avoiding operation, so it is not
stable by $\alpha$-renaming $\ctx$, as renaming can change the set of
variables captured by $\ctx$ (if the hole of the context appears in
the scope of the binder). Nonetheless, taking into account both the
context $\ctx$ and the term $\tm$ to be plugged into $\ctx$, one can
always rename both the bound variable in $\ctx$ and its free
occurrences in $\tm$ and satisfy the convention. Said differently, the
contexts we consider are always obtained by splitting a term $\tm$ as
a subterm $\tmtwo$ and a context $\ctx$ such that $\ctxp\tmtwo = \tm$,
so we assume that $\tm$ has been renamed before splitting it into
$\ctx$ and $\tmtwo$, guaranteeing that $\ctx$ respects the
convention. In particular, we shall freely assume that in
$\tm\esub\var\tmtwo$ and $\ctx\esub\var\tmtwo$ there are no free
occurrences of $\var$ in $\tmtwo$, as this can always be obtained by
an appropriate $\alpha$-conversion.
\end{remark}
}

\subsection{Dynamics.} The (shallow) rewriting rules $\todb$
($\db$ = \emph{$\beta$ at a distance}) and $\tols$ ($\ls$ = linear
substitution) are given by:
\[\begin{array}{c@{\hspace{1.5cm}}c}
  \textsc{Rule at Top Level} & \textsc{(Shallow) Contextual closure} \\
	\ap{\sctxp{\l \var.\tm}} \tmtwo  \rtodb \sctxp{\tm\esub{\var}{\tmtwo}} &
        \ctxp \tm \todb \ctxp \tmtwo \textrm{~~~if } \tm \rtodb \tmtwo \\

  \ctxp{\var}\esub{\var}{\tmtwo} \rtols \ctxp{\tmtwo}\esub{\var}{\tmtwo} &
        \ctxp \tm \tols \ctxp \tmtwo \textrm{~~~if } \tm \rtols \tmtwo
\end{array}\]

\noindent and the union of $\todb$ and $\tols$ is simply noted $\to$.

\ben{Let us point out a slight formal abuse of our system: rule
  $\tols$ does not preserve Barendregt's convention (shortened BC), as
  it duplicates the bound names in $\tmtwo$, so BC is not stable by
  reduction. To preserve BC it would be enough to replace the target
  term with $\ctxp{\tmtwo^\alpha}\esub{\var}{\tmtwo}$, where
  $\tmtwo^\alpha$ is an $\alpha$-equivalent copy of $\tmtwo$ such that
  all bound names in $\tmtwo$ have been replaced by fresh and distinct
  names. Such a renaming can be done while copying $\tmtwo$ and thus
  does not affect the complexity of implementing $\tols$. In order to
  lighten this already technically demanding paper, however, we
  decided to drop an explicit and detailed treatment of
  $\alpha$-equivalence, and so we simply stick to
  $\ctxp{\tmtwo}\esub{\var}{\tmtwo}$, letting the renaming implicit.

The implicit use of BC also rules out a few degenerate rewriting
sequences. For instance, the following degenerated behavior is
\emph{not} allowed
\[\begin{array}{ccccccc}
    \var\esub\var{\var\var} &\tols &(\var\var)\esub\var{\var\var} &\tols &((\var\var)\var)\esub\var{\var\var} \tols \ldots
\end{array}\]
because the initial term does not respect BC. By $\alpha$-equivalence
we rather have the following evaluation sequence, ending on a normal
form
\[\begin{array}{ccccccc}
 \var\esub\var{\var\var} &\eqalpha &\vartwo\esub\vartwo{\var\var} &\tols &(\var\var)\esub\vartwo{\var\var}
\end{array}\]
Finally, BC implies that in $\tols$ the context $\ctx$ is assumed to
not capture $\var$, so that $(\l \var.\var)\esub{\var}{\vartwo}
\not\tols (\l \var.\vartwo)\esub{\var}{\vartwo}$.}

The just defined shallow fragment simply ignores garbage collection
(that in the \lsc\ can always be postponed
\cite{DBLP:conf/rta/Accattoli12}) and lacks some of the nice
properties of the LSC (obtained simply by replacing shallow contexts
by general contexts). Its relevance lies in the fact that it is the
smallest fragment implementing linear \lo\ reduction (see forthcoming
\refdef{linear-lo-red}). The following are examples of shallow steps:
\begin{align*}
  \ap{(\l\var.\var)}\vartwo&\todb\var\esub\var\vartwo;\\
  (\ap\var\var)\esub\var\tm&\tols(\ap\var\tm)\esub\var\tm;
\end{align*}
while the following are not 
\begin{align*}
  \tm\esub\varthree{\ap{(\l\var.\var)}\vartwo}&\todb\tm\esub\varthree{\var\esub\var\vartwo};\\
  \var\esub\var\vartwo\esub\vartwo\tm&\tols\var\esub\var\tm\esub\vartwo\tm.
\end{align*}
\ben{With respect to the literature on the LSC we slightly abuse the
  notation, as $\todb$ and $\tols$ are usually used for the
  unrestricted versions, while here we adopt them for their shallow
  variants. Let us also warn the reader of a possible source of
  confusion: in the literature there exists an alternative notation
  and terminology in use for the LSC, stressing the linear logic
  interpretation, for which $\todb$ is noted $\tom$ and called
  \emph{multiplicative} (cut-elimination rule) and $\tols$ is noted
  $\toe$ and called \emph{exponential}.}

Taking the external context into account, a substitution step has the
following \emph{explicit} form: $
\ctxtwop{\ctxp{\var}\esub{\var}{\tmtwo}} \tols
\ctxtwop{\ctxp{\tmtwo}\esub{\var}{\tmtwo}}$. We shall often use a
\emph{compact} form:

\[\begin{array}{c@{\hspace{1.5cm}}c}
  \multicolumn{2}{c}{\textsc{Linear Substitution in Compact Form}}\\
	\ctxthreep{\var}\tols\ctxthreep{\tmtwo} &
        \mbox{if }\ctxthree =
\ctxtwop{\ctx\esub{\var}{\tmtwo}}\\

\end{array}\]

Since every $\tols$ step has a unique compact form, and a shallow
context is the compact form of at most one $\tols$ step, it is natural
to use the compact context of a $\tols$ step as its position.

\begin{definition}
\label{def:redex-position}
Given a $\todb$-redex $\ctxp\tm \todb \ctxp\tmtwo$ with
$\tm\rtodb\tmtwo$ or a compact $\tols$-redex $\ctxp\var\tols\ctxp\tm$,
the \deff{position} of the redex is the context $\ctx$.
\end{definition}

As for $\l$-calculus, we identify a redex with its position, thus
using $\ctx,\ctxtwo,\ctxthree$ for redexes, and use
$\deriv:\tm\to^k\tmtwo$ for (possibly empty) derivations. We write
$\esmeas\tm$ for the number of substitutions in $\tm$ and
$\sizedb\deriv$ for the number of $\db$-steps in $\deriv$. \medskip

\subsection{Linear \lo\ Reduction}

We redefine the \lo\ order on contexts to accommodate ES.

\begin{definition}
  The following definitions are given with respect to general (not
  necessarily shallow) contexts, even if apart from \refsect{standard}
  we will use them only for shallow contexts.
  \begin{varenumerate}
  \item 
    The \deff{outside-in order}: 
    \begin{varvarenumerate}
    \item 
      \emph{Root}: $\ctxhole\outin\gctx$ for every context
      $\gctx\neq\ctxhole$;
    \item 
      \emph{Contextual closure}: If $\gctx\outin\gctxtwo$ then
      $\gctxthreep\gctx\outin\gctxthreep\gctxtwo$ for any context
      $\gctxthree$.
    \end{varvarenumerate}
    Note that $\outin$ can be seen as the prefix relation $\prefix$ on
      contexts.
  \item 
    The \deff{left-to-right order}: $\gctx\leftright\gctxtwo$ is defined by:
    \begin{varvarenumerate}
    \item 
      \emph{Application}: If $\gctx\prefix \tm$ and
      $\gctxtwo\prefix\tmtwo$ then
      $\ap\gctx\tmtwo\leftright\ap\tm\gctxtwo$;
    \item 
      \emph{Substitution}: If $\gctx\prefix \tm$ and
      $\gctxtwo\prefix\tmtwo$ then
      $\gctx\esub\var\tmtwo\leftright\tm\esub\var\gctxtwo$;
    \item 
      \emph{Contextual closure}: If $\gctx\leftright\gctxtwo$ then
      $\gctxthreep\gctx\leftright\gctxthreep\gctxtwo$ for any context
      $\gctxthree$.
    \end{varvarenumerate}  
  \item 
    The \deff{left-to-right outside-in order}:
    $\gctx\leftout\gctxtwo$ if $\gctx\outin\gctxtwo$ or
    $\gctx\leftright\gctxtwo$:
  \end{varenumerate}
\end{definition}

\noindent Some examples:
\[\begin{array}{cccc}
  \ap{(\l\var.\ctxhole)}\tm
  &\outin 
  (\l\var.(\ctxhole\esub\vartwo\tmtwo))\tm;\\
	\tm\esub\var\ctxhole
  &\outin 
  \tm\esub\var{\tmtwo\ctx};\\
  \ap{\tm\esub\var\ctx}\tmtwo
  &\leftright
  \ap{\tm\esub\var\tmthree}\ctxhole
  &\mbox{if $\ctx\prefix\tmthree$}.
\end{array}\]
Note that the outside-in order $\outin$ can be seen as the prefix relation $\prefix$ on contexts.

The next lemma guarantees that we defined a total order.
\begin{lemma}[Totality of $\leftout$]\label{l:lefttor-basic} 
  If $\gctx\prefix\tm$ and $\gctxtwo\prefix\tm$ then either
  $\gctx\leftout\gctxtwo$ or $\gctxtwo\leftout\gctx$ or
  $\gctx=\gctxtwo$.
\end{lemma}

\begin{proof}
  By induction on $\tm$. 
\end{proof}

Remember that we identify redexes with their position context and write $\ctx\leftout\ctxtwo$. We can now define \lo\ reduction in the LSC, first considered in \cite{non-standard-preprint}.

\begin{definition}[\lo\ Linear Reduction $\tolo$]
\label{def:linear-lo-red}
  Let $\tm$ be a term and $\ctx$ a redex of $\tm$. $\ctx$ is the
  \deff{leftmost-outermost} (\lo\ for short) redex of $\tm$ if $\ctx\leftout\ctxtwo$ for
  every other redex $\ctxtwo$ of $\tm$. We
  write $\tm\tolo\tmtwo$ if a step reduces
  the \lo\ redex.
\end{definition}

\emph{Technical Remark}. Note that one cannot define $\tolo$ as the
union of the two natural rules $\Rew{\lo\db}$ and $\Rew{\lo\ls}$,
reducing the \lo\ $\db$ and $\ls$ redexes, respectively. For example,
if $I \defeq \la\varthree\varthree$ then $(\var\vartwo)\esub\var I (I
\vartwo) \Rew{\lo\db} (\var\vartwo)\esub\var I
(\varthree\esub\varthree\vartwo)$, while we have
$(\var\vartwo)\esub\var I (I \vartwo) \tolo (I\vartwo)\esub\var I (I
\vartwo)$, because the \lo\ redex has to be chosen among both $\db$
and $\ls$ redexes. Therefore, we will for instance say \emph{given a
  $\tolo$ $\db$-step} and not \emph{given the \lo\ $\db$-step}.

\section{Unfoldings}
In \refsect{hard}, we defined the unfolding $\unf\tm$ of a term $\tm$
(\refdef{unfolding}, page \pageref{def:unfolding}). Here we extend it
in various ways. We first define context unfoldings, then we
generalize the unfolding (of both terms and contexts) relatively to a
context, and finally we unfold shallow derivations.

\subsection{Unfolding Contexts} Shallowness is crucial here: the
unfolding of a shallow context is still a context, because the hole
cannot be duplicated by unfolding, being out of all ES.  First of all,
we define \emph{substitution on (general) contexts}:
\begin{center}
$\begin{array}{rclllrclllllll}
	\ctxhole\isub\var\tmtwo 			& \defeq & \ctxhole\\
	(\la\vartwo\gctx)\isub\var\tmtwo 	& \defeq & \la\vartwo\gctx\isub\var\tmtwo
	&&&(\la\var\gctx)\isub\var\tmtwo 	& \defeq & \la\var\gctx\\
	(\tm \gctx)\isub\var\tmtwo 			& \defeq & \tm\isub\var\tmtwo \gctx\isub\var\tmtwo
	&&&(\gctx \tm)\isub\var\tmtwo 			& \defeq & \gctx\isub\var\tmtwo \tm\isub\var\tmtwo\\
	\gctx \esub\vartwo\tm\isub\var\tmtwo	& \defeq & \gctx\isub\var\tmtwo \esub\vartwo{\tm\isub\var\tmtwo}
	&&&\gctx \esub\var\tm\isub\var\tmtwo	& \defeq & \gctx\esub\var{\tm\isub\var\tmtwo}\\
	\tm \esub\vartwo\gctx\isub\var\tmtwo	& \defeq & \tm\isub\var\tmtwo \esub\vartwo{\gctx\isub\var\tmtwo}
	&&&\tm \esub\var\gctx\isub\var\tmtwo	& \defeq & \tm \esub\var{\gctx\isub\var\tmtwo}
\end{array}$
\end{center}
\ben{Note that the definition of $\ctx\isub\var\tmtwo$ assumes that the free variables of $\tmtwo$ are not captured by $\ctx$ (that means that for instance $\vartwo\notin\fv\tmtwo$ in $(\la\vartwo\gctx)\isub\var\tmtwo$). This can always be achieved by $\alpha$-renaming $\ctx$ (according to \refrem{alpha-for-contexts})}.

And then define \emph{context unfolding} $\unf\ctx$ as:
\begin{center}
$\begin{array}{lllllllllllll}
	\unf\ctxhole 				& \defeq & \ctxhole
	&&&\unf{(\la\var\ctx)} 		& \defeq & \la\var\unf{\ctx}\\
	\unf{(\tm \ctx)} 		& \defeq & \unf\tm \unf{\ctx}
	&&&\unf{(\ctx \tm)} 		& \defeq & \unf{\ctx} \unf\tm\\
	\unf{\ctx \esub\var\tm} 		& \defeq & \unf\ctx \isub\var{\unf\tm}
\end{array}$
\end{center}
We have the following properties.
\begin{lemma}
\label{l:ctx-unf-new} 
Let $\ctx$ be a shallow contexts. Then:
\begin{varenumerate}
\item\label{p:ctx-unf-new-one}
  $\unf{\ctx}$ is a context;
\item\label{p:ctx-unf-new-two}
  $\ctxp\tm\isub\var\tmtwo={\ctx\isub\var\tmtwo}\ctxholep{\tm\isub\var\tmtwo}$;
\end{varenumerate}
\end{lemma}
\begin{proof}
  By induction on $\ctx$.
\end{proof}

An important notion of context will be that of applicative context,
\ie\ of context whose hole is applied to an argument, and that if
plugged with an abstraction provides a $\db$-redex.
\begin{definition}[Applicative Context]
\label{def:ap-context}
An \deff{applicative context} is a context $\apctx \grameq
\ctxp{\sctx\tm}$, where $\ctx$ and $\sctx$ are a shallow and a
substitution context, respectively.
\end{definition}
Note that applicative contexts are \emph{not} made out of applications
only: $\tm (\l \var.(\ctxhole \esub\vartwo\tmtwo \tmthree))$ is an
applicative context.

\begin{lemma}
\label{l:ap-unf} 
Let $\ctx$ be a context. Then,
 \begin{varenumerate}
  \item\label{p:ap-unf-sub} $\ctx$ is applicative iff $\ctx\isub\var\tm$ is applicative;
  \item\label{p:ap-unf-unf}
  $\ctx$ is applicative iff $\unf\ctx$ is applicative.
 \end{varenumerate} 
\end{lemma}

\subsection{Relative Unfoldings}

\ben{Useful reduction will require a more general notion of unfolding
  and context unfolding. The usefulness of a redex, in fact, will
  depend crucially on the context in which it takes place. More
  precisely, it will depend on the unfolding of the term extended with
  the substitutions that the surrounding context can provide---this is
  the unfolding of a term \emph{relative to} a context. Moreover, relative
  unfoldings will also be needed for contexts.}

\begin{definition}[Relative Unfolding]
\ben{Let $\ctx$ be a (shallow) context (verifying, as usual, Barendregt's convention---see also the remark after this definition)}. The \emph{unfolding
$\relunf{\tm}{\ctx}$ of a term $\tm$ relative to $\ctx$} and the \emph{unfolding
$\relunf{\ctxtwo}{\ctx}$ of a (shallow) context $\ctxtwo$ relative to $\ctx$} are defined by:

\begin{center}
\begin{tabular}{cccccccc}
$\begin{array}{lllllllllllll}
	\relunf\tm\ctxhole 				& \defeq & \unf\tm\\
	\relunf\tm{\la\var\ctx} 		& \defeq & \relunf\tm \ctx\\
	\relunf\tm{\tmtwo \ctx} 		& \defeq & \relunf\tm \ctx\\
	\relunf\tm{\ctx \tmtwo} 		& \defeq & \relunf\tm\ctx\\
	\relunf\tm{\ctx \esub\var\tmtwo} 		& \defeq & \relunf{\tm}{\ctx} \isub\var{\unf\tmtwo}
\end{array}$ &&&&&

$\begin{array}{lllllllllllll}
	\relunf\ctxtwo\ctxhole 				& \defeq & \unf\ctxtwo\\
\relunf\ctxtwo{\la\var\ctx} 		& \defeq & \relunf\ctxtwo \ctx\\
	\relunf\ctxtwo{\tmtwo \ctx} 		& \defeq & \relunf\ctxtwo \ctx\\
	\relunf\ctxtwo{\ctx \tmtwo} 		& \defeq & \relunf\ctxtwo\ctx\\
	\relunf\ctxtwo{\ctx \esub\var\tmtwo} 		& \defeq & \relunf{\ctxtwo}{\ctx} \isub\var{\unf\tmtwo}
\end{array}$
\end{tabular}
\end{center}
\end{definition}

\noindent For instance, $\relunf{(\ap\var\vartwo)}{(\ap{\ctxhole\esub\vartwo\var}\tm)\esub\var{\l\varthree.(\ap\varthree\varthree)}}
  =
  {(\l\varthree.(\ap\varthree\varthree))}{(\l\varthree.(\ap\varthree\varthree))}$. 
\ben{  
  \begin{remark}[Relative Unfoldings and Barendregt's Convention]
   Let us point out that the definition of relative unfolding $\relunf\tm\ctx$ relies crucially on the use of Barendregt's convention for contexts (according to \refrem{alpha-for-contexts}). For contexts not respecting the convention, in fact, the definition does not give the intended result. For instance, 
   $$\relunf\var{(\la\var\ctxhole)\esub\var\vartwo} = \relunf\var{(\la\var\ctxhole)}\isub\var\vartwo = \relunf\var\ctxhole \isub\var\vartwo = \var\isub\var\vartwo = \vartwo$$
   while the result should clearly be $\var$.
  \end{remark}
  }
  We also state some further properties of relative unfolding, to be used in the proofs, and proved by easy inductions.

\begin{lemma}[Properties of Relative Unfoldings]
	\label{l:relunf-prop} 
	Let $\tm$ and $\tmtwo$ be terms and $\ctx$ and $\ctxtwo$ be shallow contexts.
	\begin{varenumerate}
		\item \label{p:relunf-prop-minusone} \emph{Well-Definedness}: $\relunf{\ctx}\ctxtwo$ is a context.
		
		\item \label{p:relunf-prop-zero} \emph{Commutation}: the following equalities hold \ben{(and in those on the right $\ctx$ is assumed to not capture $\var$)}
		
		\begin{center}$\begin{array}{rclccrcl}
		 \relunf{(\tm\tmtwo)}\ctx   	& = & \relunf\tm\ctx \relunf\tmtwo\ctx
		 &&& \relunf{(\la\var\tm)}\ctx  	& = & \relunf{\la\var\tm}\ctx\\
		 \relunf{\tm\isub\var\tmtwo}\ctx & = & \relunf{\tm\esub\var\tmtwo}\ctx
		 &&& \relunf{\tm\isub\var\tmtwo}\ctx & = & \relunf\tm\ctx \isub\var{\relunf\tmtwo\ctx}\\
		 \unf\ctx\isub{\var}{\unf\tm} 	& = & \unf{\ctx\isub{\var}{\unf\tm}}
		 &&& \relunf\tm{\ctx\esub\var\tmtwo} & = & \relunf{\tm\isub\var{\unf\tmtwo}}{\ctx\isub\var{\unf\tmtwo}}
		\end{array}$\end{center}

		\item \label{p:relunf-prop-one} \emph{Freedom}:
		if $\ctx$ does not capture any free variable of $\tm$ then $\relunf\tm\ctx = \unf\tm$.
		
		\item \label{p:relunf-prop-five} \emph{Relativity}:
		if $\unf\tm =\unf\tmtwo$ then $\relunf\tm\ctx = \relunf\tmtwo\ctx$ and if $\unf\tm$ has a $\beta$-redex then so does $\relunf\tm\ctx$.
		
		\item \label{p:relunf-prop-six} \emph{Splitting}:
		$ \relunf\tm{\ctxp\ctxtwo} = \relunf{\relunf\tm\ctxtwo}{\ctx}$.
		
		\item \label{p:relunf-prop-four} \emph{Factorization}:
		$ \relunf{\ctxtwop\tm}\ctx = \relunf\ctxtwo\ctx \ctxholep{\relunf\tm{\ctxp\ctxtwo}}$, and in particular 
		\begin{varvarenumerate}
		 \item \label{p:relunf-prop-four-a}$ \unf{\ctxp\tm} = \unf\ctx \ctxholep{\relunf\tm{\ctx}}$,
		 \item $ \relunf{\sctxp\tm}\ctx = \relunf\tm{\ctxp\sctx}$, and
		 \item if $\ctxtwo\prefix\tmtwo$ then $\unf{\ctxtwo}\prefix\unf{\tmtwo}$.
		\end{varvarenumerate}

	\end{varenumerate}
\end{lemma}

\subsection{Unfolding Derivations}
Given a derivation $\deriv:\tm\to^*\tmtwo$ in the \lsc, we often
consider the $\beta$-derivation
$\unf{\deriv}:\unf{\tm}\tob^*\unf{\tmtwo}$ obtained by projecting
$\deriv$ via unfolding. It is built in two steps.

As expected, linear substitution steps do not modify the unfolding, as the next lemma shows. \ben{Its proof is a nice application of the properties of contexts and (relative) unfoldings, allowed by the restriction to shallow contexts (the property is valid more generally for unrestricted $\ls$ steps, but we will need it only for the shallow ones).}

\begin{lemma}[$\tols$ Projects on $=$]
\label{l:ls-unfolding}
 If $\tm\tols\tmtwo$ then $\unf\tm = \unf\tmtwo$.
\end{lemma}

\ben{
\begin{proof}
 If $\tm \tols\tmtwo$ is a root step, \ie\ $\tm = \ctxp\var \esub\var\tmthree \rtols \ctxp\tmthree \esub\var\tmthree = \tmtwo$ then 
 \[\begin{array}{rllclc}
    \unf{\ctxp\var \esub\var\tmthree} & = & \unf{\ctxp\var} \isub\var{\unf\tmthree} \\
    & =_\reflemmaeqp{relunf-prop}{four-a} & \unf\ctx\ctxholep{\relunf\var\ctx} \isub\var{\unf\tmthree} \\
    & =_\reflemmaeqp{relunf-prop}{one} & \unf\ctx\ctxholep\var \isub\var{\unf\tmthree} \\
    & =_\reflemmaeqp{ctx-unf-new}{two} & \unf\ctx\isub\var{\unf\tmthree}\ctxholep{\unf\tmthree} \\
    & = & \unf\ctx\isub\var{\unf\tmthree}\ctxholep{\unf\tmthree\isub\var{\unf\tmthree}} \\
    & =_\reflemmaeqp{ctx-unf-new}{two} & \unf\ctx\ctxholep{\unf\tmthree}\isub\var{\unf\tmthree} \\
    & =_\reflemmaeqp{relunf-prop}{one} & \unf\ctx\ctxholep{\relunf\tmthree\ctx}\isub\var{\unf\tmthree} \\
    & =_\reflemmaeqp{relunf-prop}{four-a} & \unf{\ctxp\tmthree}\isub\var{\unf\tmthree}  \\
    & = & \unf{\ctxp\tmthree \esub\var\tmthree}
    \end{array}\]
 Suppose instead $\tm \tols\tmtwo$ because $\tm = \ctxtwop\tmthree \tols \ctxtwop\tmfour = \tmtwo$ with $\tmthree\rtols\tmfour$. By what we just proved we obtain $\unf\tmthree = \unf\tmfour$, that gives $\relunf\tmthree\ctxtwo = \relunf\tmfour\ctxtwo$ by \reflemmap{relunf-prop}{five}. Then $\unf{\ctxtwop\tmthree} =_\reflemmaeqp{relunf-prop}{four-a} = \unf\ctxtwo\ctxholep{\relunf\tmthree\ctxtwo} = \unf\ctxtwo\ctxholep{\relunf\tmfour\ctxtwo} =_\reflemmaeqp{relunf-prop}{four-a} = \unf{\ctxtwop\tmfour}$.
\end{proof}
}

Instead, $\db$-steps project to $\beta$-steps. Because of shallowness, we actually obtain a strong form of projection, as every $\db$-step projects on a single $\beta$-step. We are then allowed to identify $\db$ and $\beta$-redexes. 

\begin{lemma}[$\todb$ Strongly Projects on $\tob$]
	\label{l:db-projection}
	Let $\tm$ be a \lsc\ term and $\ctx:\tm\todb\tmtwo$ a $\db$-redex in $\tm$. Then $\unf\ctx: \unf\tm\tob\unf\tmtwo$.
	
\end{lemma}

\proof

Let $\tm=\ctxp {\tmthree}\todb\ctxp{\tmfour}=\tmtwo$ with $\tmthree\rtodb\tmfour$. We show that $\unf{\ctx}:\unf{\tm}=\unf{\ctx}\ctxholep{\relunf{\tmthree}{\ctx}}\tob\unf{\ctx}\ctxholep{\relunf{\tmfour}{\ctx}}=\unf{\tmtwo}$ with $\relunf{\tmthree}{\ctx}\rtob\relunf{\tmfour}{\ctx}$ ($\rtodb$ and $\rtob$ stands for $\db/\beta$-reduction at top level). By induction on $\ctx$. Cases:
\begin{varenumerate}
	\item \emph{Empty context $\ctx=\ctxhole$}. Notation: given $\sctx=\esub{\var_1}{\tmthree_1}\ldots\esub{\var_n}{\tmthree_n}$ we denote with $\sctximp$ the list of implicit substitutions $\isub{\var_1}{\tmthree_1}\ldots\isub{\var_n}{\tmthree_n}$.  Then $\tm=\tmthree=\sctxp{\l\var.\tmfive}\tmsix$, $\unf{\tm}=\relunf{\tmthree}{\ctx} = \unf\tmthree$, $\tmtwo=\tmfour$, $\unf{\tmtwo}=\relunf{\tmfour}{\ctx} = \unf\tmfour$, and
		\begin{center}
			\commDiagramRed{$\tm=\sctxp{\l\var.\tmfive}\tmsix$            }{
											$\sctxp{\tmfive\esub\var\tmsix}=\tmtwo$		       }{
											$\unf{\tm}=(\l\var.\unf{\tmfive}\sctximp) (\unf{\tmsix})   $}{
											$\unf{\tmfive}\sctximp \isub\var{\unf{\tmsix}}=\unf{\tmfive}\isub\var{\unf{\tmsix}}\sctximp   =\unf{\tmtwo}$ }{
											$\db$}{$\beta$}{\unfsym}{\unfsym}
		\end{center}
		where the first equality in the South-East corner is given by the fact that $\var$ does not occur in $\sctx$ and the variables on which $\sctximp$ substitutes do not occur in $\tmsix$, as is easily seen by looking at the starting term. Thus the implicit substitutions $\sctximp$ and $ \isub\var{\unf{\tmsix}}$ commute.

	\item \emph{Abstraction $\ctx=\l\var.\ctxtwo$}. It follows immediately by the \ih
	
	\item \emph{Left of an application $\ctx=\ctxtwo\tmfive$}. By \reflemmap{relunf-prop}{zero} we know that $\unf{\tm}=\unf{\ctxtwo}\ctxholep{\relunf{\tmthree}{\ctxtwo}}\unf{\tmfive}$. Using the \ih\ we derive the following diagram:		
		\begin{center}
			\commDiagramRed{$
											\ctxtwop\tmthree\tmfive
											$}{$
											\ctxtwop\tmfour\tmfive
											$}{$
											\unf{\ctxtwo}\ctxholep{\relunf{\tmthree}{\ctxtwo}}\unf{\tmfive}
											$}{$
											\unf{\ctxtwo}\ctxholep{\relunf{\tmfour}{\ctxtwo}}\unf{\tmfive}
											$}{$
											\db$}{$\beta$}{\unfsym}{\unfsym}
		\end{center}

\item \emph{Right of an application $\ctx=\tmfive\ctxtwo$}. By \reflemmap{relunf-prop}{zero} we know that $\unf{\tm}=\unf{\tmfive}\unf{\ctxtwo}\ctxholep{\relunf{\tmthree}{\ctxtwo}}$. Using the \ih\ we derive the following diagram:		
		\begin{center}
			\commDiagramRed{$
											\tmfive\ctxtwop\tmthree
											$}{$
											\tmfive\ctxtwop\tmfour
											$}{$
											\unf{\tmfive}\unf{\ctxtwo}\ctxholep{\relunf{\tmthree}{\ctxtwo}}
											$}{$
											\unf{\tmfive} \unf{\ctxtwo}\ctxholep{\relunf{\tmfour}{\ctxtwo}}
											$}{$
											\db$}{$\beta$}{\unfsym}{\unfsym}
		\end{center}
	\item \emph{Substitution $\ctx=\ctxtwo\esub\var\tmfive$}. By \ih\ $\unf{\ctxtwo}\ctxholep{\relunf{\tmthree}{\ctxtwo}}\tob \unf{\ctxtwo}\ctxholep{\relunf{\tmfour}{\ctxtwo}}$ with $\relunf{\tmthree}{\ctxtwo} \rtob \relunf{\tmfour}{\ctxtwo}$. In general, from the definition of substitution it follows $\tmsix\rtob\tmsix'$ implies $\tmsix\isub\var{\tmsix''}\rtob\tmsix'\isub\var{\tmsix''}$. Therefore, $\relunf{\tmthree}{\ctxtwo}\isub\var{\unf\tmfive} \rtob \relunf{\tmfour}{\ctxtwo}\isub\var{\unf\tmfive}$. The following calculation concludes the proof:
$$
\begin{array}{llllllll}
\unf{\tm}&=&
\unf{\ctxtwo}\ctxholep{\relunf{\tmthree}{\ctxtwo}}\isub\var{\unf\tmfive}\\
&=_{\reflemmaeqp{ctx-unf-new}{two}}&\unf{\ctxtwo}\isub\var{\unf\tmfive}\ctxholep{\relunf{\tmthree}{\ctxtwo}\isub\var{\unf\tmfive}}\\
&\tob &\unf{\ctxtwo}\isub\var{\unf\tmfive}\ctxholep{\relunf{\tmfour}{\ctxtwo}\isub\var{\unf\tmfive}}\\
&=_{\reflemmaeqp{ctx-unf-new}{two}}&\unf{\ctxtwo}\ctxholep{\relunf{\tmfour}{\ctxtwo}}\isub\var{\unf\tmfive}\unf{\tmtwo}\rlap{\hbox
                                     to 154 pt{\hfill\qEd}}                                    
\end{array}
$$
\end{varenumerate}

%

\section{Useful Derivations}
\label{sect:useful}
In this section we define useful reduction, a constrained, optimized reduction, that will be the key to the invariance
theorem. The idea is that an optimized substitution step $\ctxp\var\esub\var\tmtwo \tols \ctxp\tmtwo\esub\var\tmtwo$ takes place only if it contributes to eventually obtain an unshared (\ie\ shallow) $\beta/\db$-redex. \emph{Absolute} usefulness can be of two kinds. 

\begin{varenumerate}
	\item \emph{Duplication}: a step can \emph{duplicate} $\db$-redexes, as in 
	$$\ctxp\var\esub\var{(\la\vartwo\tmthree)\tmfour} \tols 
	\ctxp{(\la\vartwo\tmthree)\tmfour}\esub\var{(\la\vartwo\tmthree)\tmfour}$$
	\item \emph{Creation}: it can \emph{create} a new $\db$-redex with its context, if it substitutes an abstraction in an applicative context, as in  
	$$\ctxp{\sctxp\var\tmtwo}\esub\var{\la\vartwo\tm} \tols
	\ctxp{\sctxp{\la\vartwo\tm}\tmtwo}\esub\var{\la\vartwo\tm}$$
\end{varenumerate}

\noindent There is also a subtler \emph{relative} usefulness to $\db$-redexes. A substitution step may indeed put just a piece of what later, with further substitutions, will become a $\db$-redex. Accommodating relative usefulness requires to generalize the duplication and the creation cases to contexts and relative unfoldings. 

Let us give some examples. The following step
$$
(\tm\var)\esub\var{\ap\vartwo\vartwo} \tols
(\tm(\vartwo\vartwo))\esub\var{\ap\vartwo\vartwo}
$$
is useless. However, in an appropriate context it is relatively useful. For instance, let us put it in $\ctxhole\esub\vartwo{\l\varthree.\varthree}$, obtaining a case of \emph{relatively useful duplication},
\begin{equation}
\label{eq:relative-dup}
(\ap\tm\var)\esub\var{\vartwo\vartwo}\esub\vartwo{\l\varthree.\varthree}
\tols
(\ap\tm{(\ap\vartwo\vartwo)})\esub\var{\vartwo\vartwo}\esub\vartwo{\l\varthree.\varthree}
\end{equation}
Note that the step, as before, does not duplicate a $\db$-redex. Now, however, evaluation will continue and turn the substituted copy of $\vartwo\vartwo$ into a $\db$-redex, as follows
$$
(\tm{(\vartwo\vartwo)})\esub\var{\vartwo\vartwo}\esub\vartwo{\la\varthree\varthree}
\tols
(\tm{((\la\varthree\varthree)\vartwo)})\esub\var{\vartwo\vartwo}\esub\vartwo{\la\varthree\varthree}
$$
We consider the step in \refeq{relative-dup} a case of relative duplication because $\vartwo\vartwo$ contains a $\beta$-redex up to relative unfolding in its context, as we have
$\relunf{(\vartwo\vartwo)}{\ctxhole\esub\vartwo{\l\varthree.\varthree}}=\ap{(\l\varthree.\varthree)}{(\l\varthree.\varthree)}$, and thus duplicating $\vartwo\vartwo$ duplicates a $\beta$-redex, up to unfolding. 

Similarly, a case of \emph{relatively useful creation} is given by:
$$
(\var\tm)\esub\var{\vartwo}\esub\vartwo{\l\varthree.\varthree}
\tols
(\vartwo\tm)\esub\var{\vartwo}\esub\vartwo{\l\varthree.\varthree}
$$
Again, the step itself does not create a $\db$-redex, but---up to unfolding---it substistutes an abstraction, because $\relunf\vartwo{\ctxhole\esub\vartwo{\l\varthree.\varthree}} = \l\varthree.\varthree$, and the context is \emph{applicative} (note that a context is applicative iff it is applicative up to unfolding, by \reflemma{ap-unf}).

The actual definition of useful reduction captures at the same time absolute and relative cases by means of relative unfoldings.

\begin{definition}[Useful/Useless Steps and Derivations]
\label{def:useful-redex}
  A \deff{useful} step is either a $\db$-step or a $\ls$-step
  $\ctxp{\var}\tols \ctxp{\tmthree}$ (in compact form)
  such that:
  \begin{varenumerate}
  \item \emph{Relative Duplication}:
    either $\relunf{\tmthree}{\ctx}$ contains a $\beta$-redex, or
  \item \emph{Relative Creation}:
    $\relunf{\tmthree}{\ctx}$ is an abstraction and $\ctx$ is applicative.
  \end{varenumerate}
  A \deff{useless} step is a $\ls$-step that is not useful. A
  \deff{useful derivation} (resp. \deff{useless derivation}) is a
  derivation whose steps are useful (resp. useless).
\end{definition}

\ben{Note that a useful normal form, \ie\ a term that is normal for useful reduction, is not necessarily a normal form. For instance, the reader can now verify that the compact normal form we discussed in \refsect{hard}, namely
$$
(\vartwo_1\vartwo_1)\esub{\vartwo_1}{\vartwo_2\vartwo_2}\esub{\vartwo_{2}}{\vartwo_3\vartwo_3}\ldots\esub{\vartwo_n}\var,
$$
is a useful normal form, but not a normal form.}

As a first sanity check for useful reduction, we show that as long as there are useful substitutions steps to do, the unfolding is not $\tob$-normal.

\begin{lemma}
\label{l:useful-prom-on-beta}
 If $\ctxp\var \tols \ctxp\tm$ is useful then $\unf{\ctxp\var} = \unf{\ctxp\tm}$ has a $\beta$-redex.
\end{lemma}

\begin{proof}
The equality $\unf{\ctxp\var} = \unf{\ctxp\tm}$ holds in general for $\tols$-steps (\reflemma{ls-unfolding}). For the existence of a $\beta$-redex, note that $\unf{\ctxp\tm} = \unf\ctx\ctxholep{\relunf\tm\ctx}$ by \reflemmap{relunf-prop}{four-a}, and that $\unf\ctx$ applicative iff $\ctx$ is applicative by \reflemmap{ap-unf}{unf}. Then by relative duplication or relative creation there is a $\tob$-redex in $\unf{\ctxp\var}$. 
\end{proof}

We can finally define the strategy that will be shown to implement \lo\ $\beta$-reduction within a polynomial overhead.

\begin{definition}[\lo\ Useful Reduction $\tolou$]
  Let $\tm$ be a term and $\ctx$ a redex of $\tm$. $\ctx$ is the
  \deff{leftmost-outermost useful} (\lou\ for short) redex of $\tm$ if $\ctx\leftout\ctxtwo$ for
  every other useful redex $\ctxtwo$ of $\tm$. We
  write $\tm\tolou\tmtwo$ if a step reduces
  the \lou\ redex.
\end{definition}

\ben{\subsection{On Defining Usefulness via Residuals.} Note that
  \opt\ steps concern future creations of $\beta$-redexes and yet
  circumvent the explicit use of residuals, relying on relative
  unfoldings only. It would be interesting, however, to have a
  characterization based on residuals. We actually spent time
  investigating such a characterization, but we decide to leave it to
  future work. We think that it is informative to know the reasons,
  that are the following:
\begin{varenumerate}
\item
  a definition based on residuals is not required for the final result
  of this paper;
\item
  the definition based on relative unfoldings is preferable, as it
  allows the complexity analysis required for the final result;
\item
  we believe that the case studied in this paper, while certainly
  relevant, is not enough to develop a general, abstract theory of
  usefulness. We feel that more concrete examples should first be
  developed, for instance in call-by-value and call-by-need scenarios,
  and comparing weak and strong variants, extending the language with
  continuations or pattern matching, and so on. The complementary
  study in \cite{usef-constr}, indeed, showed that the weak
  call-by-value case already provides different insights, and that
  useful sharing as studied here is only an instance of a more general
  concept;
\item
  we have a candidate characterization of useful reduction using
  residuals, for which however one needs sophisticated rewriting
  theory. It probably deserves to be studied in another paper.
  Our candidate characterization relies on a less rigid order between
  redexes of the LSC than the total order $\leftout$ considered here,
  namely the partial \emph{box order} $\prec_{\tt box}$ studied in
  \cite{non-standard-preprint}. Our conjecture is that an $\tols$
  redex $\ctx$ is useful iff it is shallow and
  \begin{varvarenumerate}
  \item
    there is a (not necessarily shallow) $\todb$ redex $\gctx$ such
    that $\ctx\prec_{\tt box}\gctx$, or
  \item
    $\ctx$ creates a shallow $\todb$ redex, or
  \item
    there is a (not necessarily shallow) $\tols$ redex $\gctx$ such
    that $\ctx\prec_{\tt box}\gctx$ and there exists a residual
    $\gctxtwo$ of $\gctx$ after $\ctx$ that is useful.
  \end{varvarenumerate}
  Coming back to the previous point, we feel that such an abstract
  characterization---assuming it holds---is not really satisfying, as
  it relies too much on the concrete notion of shallow redex. It is
  probably necessary to abstract away from a few cases in order to
  find the right notion. An obstacle, however, is that the rewriting
  theory developed in \cite{non-standard-preprint} has yet to be
  adapted to call-by-value and call-by-need.
\end{varenumerate}\medskip
To conclude, while having a residual theory of useful sharing is
certainly both interesting and challenging, it is also certainly not
necessary in order to begin a theory of cost models for the
$\l$-calculus.  }

\section{The Proof, Made Abstract}
\label{sect:abstract-proof}
Here we describe the architecture of our proof, decomposing it, and
proving the implementation theorem from a few abstract properties. The
aim is to provide a tentative recipe for a general proof of invariance
for functional languages.\

We want to show that a certain \emph{abstract} strategy $\tostrat$ for
the $\l$-calculus provides a unitary and invariant cost model, \ie\
that the number of $\tostrat$ steps is a measure polynomially related
to the number of transitions on a Turing machine or a RAM.

In our case, $\tostrat$ will be \lo\ $\beta$-reduction $\toblo$. Such
a choice is natural, as $\toblo$ is normalizing, it produces standard
derivations, and it is an iteration of head reduction.

Because of size-explosion in the $\l$-calculus, we have to add
sharing, and our framework for sharing is the \emph{(shallow) linear
  substitution calculus}, \ugo{that plays the role of a very abstract
  intermediate machine between $\l$-terms and Turing machines. Our
  encoding will rather address an informal notion of an algorithm
  rather than Turing machines.  The algorithms will be clearly
  implementable with polynomial overhead but details of the
  implementation will not be discussed (see however
  \refsect{discussion})}.

In the LSC, $\toblo$ is implemented by \lo\ useful reduction
$\tolou$. \ben{We say that $\tolou$ is a \emph{partial strategy} of
  the LSC, because the useful restriction forces it to stop on compact
  normal forms, that in general are not normal forms of the LSC}. Let
us be abstract, and replace $\tolou$ with a general partial strategy
$\toes$ within the LSC. We want to show that $\toes$ is invariant with
respect to both $\tostrat$ and RAM. Then we need two theorems, which
together---when instantiated to the strategies $\toblo$ and
$\tolou$---yield the main result of the paper:
\begin{varenumerate}
\item 
  \emph{High-Level Implementation}: $\tostrat$ terminates iff
  $\toes$ terminates. Moreover, $\tostrat$ is implemented by $\toes$
  with only a polynomial overhead. Namely, $\tm\toes^k\tmtwo$ iff
  $\tm\tostrat^h\unf{\tmtwo}$ with $k$ polynomial in $h$ (our actual
  bound will be quadratic);
\item 
  \emph{Low-Level Implementation}: $\toes$ is implemented on
  a RAM with an overhead in time which is polynomial in both
  $k$ and the size of $\tm$.
\end{varenumerate}

\subsection{High-Level Implementation} The high-level half relies on the following notion.
\ben{
\begin{definition}[High-Level Implementation System]
Let $\tostrat$ be a deterministic strategy on $\l$-terms and $\toes$ a
partial strategy of the shallow \lsc. The pair $(\tostrat,\toes)$ is a
\deff{high-level implementation system} if whenever $\tm$ is a
LSC term it holds that:
\begin{varenumerate}
  \item 
    \emph{Normal Form}: if $\tm$ is a $\toes$-normal form then
    $\unf\tm$ is a $\tostrat$-normal form.
  \item 
    \emph{Strong Projection}: if $\tm\toes\tmtwo$ is a $\todb$ step then
    $\unf\tm\tostrat\unf\tmtwo$.
\end{varenumerate}
Moreover, it is \emph{locally bounded} if whenever $\tm$ is a
$\l$-term and $\deriv:\tm\toes^*\tmtwo$ then the length of a sequence of
$\ls$-steps from $\tmtwo$ is linear in the number $\sizedb\deriv$ of
(the past) $\db$-steps in $\deriv$.
\end{definition}
}
The normal form and projection properties address the
\emph{qualitative} part of the high-level implementation theorem, \ie\
the part about termination. The normal form property guarantees that
$\toes$ does not stop prematurely, so that when $\toes$ terminates
$\tostrat$ cannot keep going. The projection property guarantees that
termination of $\tostrat$ implies termination of $\toes$. It also states a stronger fact: \emph{$\tostrat$ steps can
  be identified with the $\db$-steps of the $\toes$ strategy}. \ben{Note that the fact that one $\todb$ step projects on exactly one $\tob$-step is a general property of the shallow LSC, given by \reflemma{db-projection}. The projection property then requires that the steps selected by the two strategies coincide up to unfolding.}

The \emph{local boundedness} property is instead used for the
\emph{quantitative} part of the theorem, \ie\ to provide the
polynomial bound. A simple argument indeed bounds the
\emph{global} number of $\ls$-steps in $\toes$ derivation with respect to the
number of $\db$-steps, that---by the identification of $\beta$ and
$\db$ redexes---is exactly the length of the associated $\tostrat$
derivation.

\begin{theorem}[High-Level Implementation]
  Let $\tm$ be an ordinary $\l$-term and $(\tostrat,\toes)$ a
  high-level implementation system. Then,
  \begin{varenumerate}
  \item \emph{Normalization}: $\tm$ is $\tostrat$-normalizing iff it is $\toes$-normalizing,
   \item \emph{Projection}: if $\deriv:\tm\toes^*\tmtwo$ then $\unf{\deriv}:\tm\tostrat^*\unf{\tmtwo}$,
  \item \emph{Overhead}: if the system is locally bounded, then $\deriv$ is at most quadratically longer than $\unf{\deriv}$, \ie\ $\size\deriv=O(\size{\unf{\deriv}}^2)$.
  \end{varenumerate}
\end{theorem}
\begin{proof}\hfill
\begin{varenumerate}
\item 
  $\Leftarrow$) Suppose that $\tm$ is $\toes$-normalizable and let
  $\deriv:\tm\toes^*\tmtwo$ a derivation to $\toes$-normal form. By
  the projection property there is a derivation
  $\tm\tostrat^*\unf{\tmtwo}$. By the normal form property
  $\unf{\tmtwo}$ is a $\tostrat$-normal form.\\  
  $\Rightarrow$) Suppose that $\tm$ is $\tostrat$-normalizable and let
  $\derivtwo:\tm\tostrat^k\tmtwo$ be the derivation to
  $\tostrat$-normal form (unique by determinism of
  $\tostrat$). Assume, by contradiction, that $\tm$ is not
  $\toes$-normalizable. Then there is a family of $\toes$-derivations
  $\deriv_i:\tm\toes^i\tmtwo_i$ with $i\in\nat$, each one extending
  the previous one. By the local boundedness, $\toes$ can make
  only a finite number of $\ls$ steps (more generally, $\tols$ is
  strongly normalizing in the \lsc). Then the sequence
  $\set{\sizedb{\deriv_i}}_{i\in\nat}$ is non-decreasing and
  unbounded. By the projection property, the family
  $\set{\deriv_i}_{i\in\nat}$ unfolds to a family of
  $\tostrat$-derivations $\set{\unf{\deriv_i}}_{i\in\nat}$ of
  unbounded length (in particular greater than $k$), absurd.
\item
  From \reflemma{ls-unfolding} ($\tols$ projects on =) and
  \reflemma{db-projection} (a single shallow $\todb$ projects on a
  single $\tob$ step) we obtain $\unf\deriv:\unf\tm\tob^*\unf\tmtwo$
  with $\size{\unf\deriv} = \sizedb\deriv$, and by the projection
  property the steps of $\unf{\deriv}$ are $\tostrat$ steps.
\item
  To show $\size\deriv=O(\size{\unf{\deriv}}^2)$ it is enough to show
  $\size\deriv=O(\sizedb{\deriv}^2)$. Now, $\deriv$ has the shape:
  $$\tm=\tmthree_1\todb^{a_1}\tmfour_1\tols^{b_1}\tmthree_2\todb^{a_2}\tmfour_2\tols^{b_2}\ldots\tmthree_k\todb^{a_k}\tmfour_k\tols^{b_k}\tmtwo$$
  By the local boundedness, we obtain $b_i\leq c\cdot\sum_{j=1}^i
  a_j$ for some constant $c$. Then:
$$\sizels\deriv=\sum_{i=1}^k b_i\leq\sum_{i=1}^k (c\cdot\sum_{j=1}^i a_j) = c\cdot\sum_{i=1}^k \sum_{j=1}^i a_j$$
  Note that $\sum_{j=1}^i a_j\leq \sum_{j=1}^k a_j=\sizedb\deriv$ and
  $k\leq\sizedb\deriv$. So
$$\sizels\deriv\leq c \cdot\sum_{i=1}^k\sum_{j=1}^i a_j \leq c \cdot\sum_{i=1}^k\sizedb\deriv \leq c\cdot\sizedb\deriv^2$$
  Finally, $\size\deriv=\sizedb\deriv+\sizels\deriv\leq
  \sizedb\deriv+c\cdot\sizedb\deriv^2=O(\sizedb\deriv^2)$.\qedhere
\end{varenumerate}
\end{proof}
Note that the properties of the implementation hold for \emph{all}
derivations (and not only for those reaching normal forms). In fact,
they even hold for derivations in strongly diverging terms. In this
sense, our cost model is robust.

\subsection{Low-Level Implementation.} For the low-level part we define three basic requirements.

\begin{definition}\label{def:lli}
\label{def:mech}
A partial strategy $\toes$ on \lsc\ terms is \deff{efficiently mechanizable} if given a
derivation $\deriv:\tm\toes^*\tmtwo$:
\begin{varenumerate}
\item \emph{No Size-Explosion}: $\size\tmtwo$ is polynomial in $\size\tm$ and $\size\deriv$;
\item \emph{Step}: every redex of $\tmtwo$ can be implemented in time polynomial in $\size\tmtwo$;
\item \emph{Selection}: the search for the next $\toes$ redex to reduce in
  $\tmtwo$ takes polynomial time in $\size\tmtwo$.
\end{varenumerate}
\end{definition}

The first two properties are natural. At first sight the
\emph{selection property} is always trivially verified: finding a
redex in $\tmtwo$ takes time linear in $\size\tmtwo$. However, our
strategy for ES will reduce only redexes satisfying a side-condition
whose na\"ive verification takes exponential time in
$\size\tmtwo$. Then one has to be sure that such a computation can be
done in polynomial time.

\begin{theorem}[Low-Level Implementation]
  Let $\toes$ be an efficiently mechanizable strategy, $\tm$ a $\l$-term, and $k$ a
  number of steps.  Then there is an algorithm that outputs
  $\strp{\toes}{k}{\tm}$, and which works in time polynomial in $k$
  and $\size\tm$.
\end{theorem}
\begin{proof}
  \ben{The algorithm for computing $\strp{\toes}{k}{\tm}$ is obtained
    by iteratively searching for the next $\toes$ redex to reduce and
    then reducing it, by using the algorithms given by the step and
    selection property.} The complexity is obtained by summing the
  polynomials given by the step and selection property, that are in
  turn composed with the polynomial of the no size-explosion
  property. Since polynomials are closed by sum and composition, the
  algorithm works in polynomial time.
\end{proof}

In \cite{DBLP:conf/rta/AccattoliL12}, we proved that \emph{head
  reduction} and \emph{linear head reduction} form a locally bounded
high-level implementation system and that linear head reduction is
efficiently mechanizable (but note that \cite{DBLP:conf/rta/AccattoliL12} does not employ the terminology we use here).
 
Taking $\toes$ as the \lo\ strategy $\tolo$ of the \lsc, almost does
the job. Indeed, $\tolo$ is efficiently mechanizable and $(\toblo,\tolo)$ is a
high-level implementation system. Unfortunately, it is not a locally
bounded implementation, because of the \emph{length-explosion} example
given in \refsect{hard}, and thus invariance does not hold. This is
why useful reduction is required.\medskip

\subsection{Efficient Mechanizability and Subterms} 
We have been very lax in the definition of efficiently mechanizable strategies. The
strategy that we will consider has the following additional property.
\begin{definition}[Subterm]
 A partial strategy $\toes$ on \lsc\ terms has the \emph{subterm property} if given a
derivation $\deriv:\tm\toes^*\tmtwo$ the terms duplicated along $\deriv$ are subterms of $\tm$.
\end{definition}

The subterm property in fact enforces \emph{linearity} in the no size-explosion and step properties, as the following immediate lemma shows.

\begin{lemma}[Subterm $\Rightarrow$ Linear No Size-Explosion + Linear Step]
  \label{l:subt-impl-nosize-step}
  If $\deriv:\tm\toes^*\tmtwo$ has the subterm property then 
  \begin{varenumerate}
  \item \label{p:subterm-2}
    \emph{Linear Size}: $\size{\tmtwo}\leq (\size\deriv+1)\cdot\size{\tm}$.    
  \item
    \emph{Linear Step}: every redex of $\tmtwo$ can be implemented in time linear in $\size\tm$;
  \end{varenumerate}
\end{lemma}

The subterm property is fundamental and common to most implementations
of functional languages
\cite{DBLP:conf/dagstuhl/Jones96,DBLP:conf/birthday/SandsGM02,DBLP:conf/rta/AccattoliL12,DBLP:conf/icfp/AccattoliBM14},
and for implementations and their complexity analysis it plays the
role of the subformula property in sequent calculus. It is sometimes
called \emph{semi-compositionality}
\cite{DBLP:conf/dagstuhl/Jones96}. We will show that every standard
derivation for the LSC has the subterm property. \ben{To the best of
  our knowledge, instead, no strategy of the $\l$-calculus has the
  subterm property. We are not aware of a proof of the nonexistence of
  strategies for $\beta$-reduction with the subterm property,
  though. For a fixed strategy, however, it is easy to build a
  counterexample, as $\beta$-reduction substitutes everywhere, in
  particular in the argument of applications that can later become a
  redex. The reason why the subterm property holds for many micro-step
  strategies, indeed, is precisely the fact that they do not
  substitute in such arguments, see also
  \refsect{subterm-via-standard}.}

The reader may wonder, then, why we did not ask the subterm property
of an efficiently mechanizable strategy. The reason is that we want to
provide a general abstract theory, and there may very well be
strategies that are efficiently mechanizable but that do not satisfy
the subterm property, or, rather, that satisfy only some relaxed form
of it.\ignore{ (\eg\ the size of duplicated subterms is
  \emph{bi}linear in $\size\tm$ \emph{and} $\size\deriv$) still
  providing efficient mechanizability. For instance, when extending
  call-by-need under abstractions one might decide to reduce the body
  of an abstraction before substituting it. This strategy---that to
  the best of our knowledge has never been studied in the
  literature---no longer enjoys the strong form of the subterm
  property, but we conjecture that it is nonetheless efficiently
  mechanizable.}

\ben{Let us also point out that in the subterm property the word
  \emph{subterm} conveys the good intuition but is slightly abused:
  since evaluation is up to $\alpha$-equivalence, a subterm of $\tm$
  is actually a subterm up to variable names, both free and
  bound. More precisely: define $\tmthree^-$ as $\tmthree$ in which
  all variables (including those appearing in binders) are replaced by
  a fixed symbol $\ast$. Then, we will consider $\tmtwo$ to be a
  subterm of $\tm$ whenever $\tmtwo^-$ is a subterm of $\tm^-$ in the
  usual sense. The key property ensured by this definition is that the
  size $\size\codetwo$ of $\codetwo$ is bounded by $\size\code$, which
  is what is actually relevant for the complexity analysis.  }

\section{Road Map}
\label{sect:road-map}
We need to ensure that \lou\ derivations are efficiently mechanizable
and form a high-level implementation system when paired with
\lob\ derivations. These are non-trivial properties, with subtle
proofs in the following sections. The following schema is designed to
help the reader to follow the master plan:
\begin{varenumerate}
  \item we will show that $(\toblo,\tolou)$ is a high-level implementation system, by showing
  \begin{varvarenumerate}
    \item the \emph{normal form} property in \refsect{normal-form}
  
    \item the \emph{projection property} in \refsect{new-projection}, by introducing \lou\ contexts;
  \end{varvarenumerate}
  \item we will prove the \emph{local boundedness} property of $(\toblo,\tolou)$ through a detour via standard derivations. The detour is in three steps:
  \begin{varvarenumerate}
    \item the introduction of \emph{standard derivations} in \refsect{standard}, that are shown to have the \emph{subterm property};
  
    \item the proof that \emph{\lou\ derivations are standard} in \refsect{subterm-via-standard}, and thus have the subterm property;
    
    \item the proof that \lou\ derivations have the \emph{local boundedness property}, that relies on the subterm and normal form properties;
  \end{varvarenumerate}
  
  \item we will prove that \lou\ derivations are efficiently mechanizable by showing:
  
  \begin{varvarenumerate}
    \item the no size-explosion and step properties, that at this point are actually already known to hold, because they follow from the subterm property (\reflemma{subt-impl-nosize-step});
  
    \item the \emph{selection property}, by exhibiting a polynomial
  algorithm to test whether a redex is useful or not, in \refsect{algorithm}.
  \end{varvarenumerate}
\end{varenumerate}\medskip
In \refsect{summing-up}, we will put everything together, obtaining an
implementation of the $\l$-calculus with a polynomial overhead, from
which invariance follows.


\section{The Normal Form Property}
\label{sect:normal-form}
 To show the normal form property we first have to generalize it to the relative unfolding in a context, in the next lemma, and then obtain it as a corollary.

   The statement of the lemma essentially says that for a useful normal form $\tmtwo$ in a context $\ctx$ the unfolding $\relunf\tmtwo\ctx$ either unfolds to a $\tob$-normal form or it has a useful substitution redex provided by $\ctx$. The second case is stated in a more technical way, spelling out the components of the redex, and will be used twice later on, in \reflemma{loudb-are-ilou} in \refsect{new-projection} (page \pageref{l:loudb-are-ilou}) and \refprop{local-bound} in \refsect{nested} (page \pageref{prop:local-bound}). To have a simpler look at the statement we suggest to ignore cases \ref{p:nf-charac-dupl} and \ref{p:nf-charac-creation} at a first reading.

\begin{lemma}[Contextual Normal Form Property]
\label{l:nf-charac}
Let $\ctxp\tmtwo$ be such that $\tmtwo$ is a useful normal form and $\ctx$ is a shallow context. Then either
\begin{varenumerate}
	\item $\relunf\tmtwo\ctx$ is a $\beta$-normal form, or
	\item $\ctx\neq \ctxhole$ and there exists a shallow context $\ctxtwo$ such that 
	\begin{varvarenumerate}
		\item $\tmtwo = \ctxtwop\var$, and
		\item $\ctxp\ctxtwo$ is the position of a useful $\ls$-redex of $\ctxp\tmtwo$, namely
		\begin{varvarvarenumerate}
			\item \label{p:nf-charac-dupl}\emph{Relative Duplication}: $\relunf\var{\ctxp\ctxtwo}$ has a $\beta$-redex, or
			\item \label{p:nf-charac-creation}\emph{Relative Creation}: $\ctxtwo$ is applicative and $\relunf\var{\ctxp\ctxtwo}$ is an abstraction.
		\end{varvarvarenumerate}
	\end{varvarenumerate}
\end{varenumerate}
\end{lemma}

\begin{proof}
	Note that:
	\begin{varitemize}
		\item If there exists $\ctxtwo$ as in case 2 then $\ctx \neq \ctxhole$, otherwise case \ref{p:nf-charac-dupl} or case \ref{p:nf-charac-creation} would imply a useful substitution redex in $\tmtwo$, while $\tmtwo$ is a useful normal form by hypothesis.
		\item
                  Cases 1 and 2 are mutually exclusive: if \ref{p:nf-charac-dupl} or \ref{p:nf-charac-creation} holds clearly $\relunf\tmtwo\ctx$ has a $\beta$-redex. We are only left to prove that one of the two always holds.
	\end{varitemize}
	By induction on $\tmtwo$. Cases:
	\begin{varenumerate}
		\item \emph{Variable} $\tmtwo = \var$. If $\relunf\var\ctx$ is a $\beta$-normal form nothing remains to be shown. Otherwise, $\relunf\var{\ctx}$ has a $\beta$-redex and we are in case \ref{p:nf-charac-dupl}, setting $\ctxtwo \defeq \ctxhole$ (giving $\relunf\var{\ctxp\ctxtwo}= \relunf\var{\ctx}$).
		
	\item \emph{Abstraction} $\tmtwo = \la\vartwo\tmthree$. Follows from the \ih\ applied to $\tmthree$ and $\la\var\ctxhole$.
	
	\item \emph{Application} $\tmtwo = \tmthree\tmfour$. Note that $\relunf\tmtwo\ctx =_{\reflemmaeqp{relunf-prop}{zero}} \relunf\tmthree\ctx \relunf\tmfour\ctx$ and that $\tmthree$ is a useful normal form, since $\tmtwo$ is. We can then apply the \ih\ to $\tmthree$ and $\ctxp{\ctxhole\tmfour}$. There are two cases:
	\begin{varvarenumerate}
		\item \emph{$\relunf\tmthree{\ctxp{\ctxhole\tmfour}} = \relunf{\tmthree}{\ctx}$ is a normal form}. Sub-cases:
		\begin{varvarvarenumerate}
			\item \emph{$\relunf{\tmthree}{\ctx}$ is an abstraction $\la\vartwo\tmfive$}. This is the interesting inductive case, because it is the only one where the \ih\ provides case 1 for $\tmthree$ but the case ends proving case 2, actually \ref{p:nf-charac-creation}, for $\tmtwo$. In the other cases of the proof (3(a)ii, 3(b), and 4), instead, the case of the statement provided by the \ih\ is simply propagated \emph{mutatis mutandis}.
			
			Note that $\tmthree$ cannot have the form $\sctxp{\la\vartwo\tmsix}$, because otherwise $\tmtwo$ would not be $\tolou$ normal. Then it follows that $\tmthree = \sctxp\var$ (as $\tmthree$ cannot be an application). For the same reason, $\unf\tmthree$ cannot have the form $\la\vartwo\tmsix$. Then $\unf\tmthree = \varthree$ for some variable $\varthree$ (possibly $\varthree=\var$).  Now take $\ctxtwo \defeq \sctx \tmfour$.  Note that  $\ctxtwo$ is applicative and that $\relunf\var{\ctxp{\sctx \tmfour}} =_{\reflemmaeqp{relunf-prop}{four}} \relunf{\sctxp\var}{\ctxp{\ctxhole\tmfour}} = \relunf{\sctxp\var}{\ctx} = \relunf\tmthree\ctx = \la\vartwo\tmfive$. So we are in case \ref{p:nf-charac-creation} and there is a useful $\tols$-redex of position $\ctxp\ctxtwo$.

			\item \emph{$\tmthree$ is not an abstraction}. Note that $\relunf\tmthree\ctx$ is neutral. Then the statement follows from the \ih\ applied to $\tmfour$ and $\ctxp{\tmthree\ctxhole}$. Indeed, if $\relunf\tmfour{\ctxp{\tmthree\ctxhole}} = \relunf\tmfour{\ctx}$ is a $\beta$-normal form then $\relunf\tmtwo\ctx = \relunf\tmthree\ctx \relunf\tmfour\ctx$ is a $\beta$-normal form and we are in case 1. While if exists $\ctxthree$ such that  $\tmfour = \ctxthreep\var$ verifies case 2 (with respect to $\ctxp{\tmthree\ctxhole}$) then $\ctxtwo \defeq \tmthree\ctxthree$ verifies case 2 with respect to $\ctx$.
		\end{varvarvarenumerate}
		
		\item \emph{exists $\ctxthree$ such that $\tmthree =\ctxthreep\var$ verifying case 2 of the statement with respect to $\ctxp{\ctxhole\tmfour}$}. Then case 2 holds for $\ctxtwo \defeq \ctxthree\tmfour$ with respect to $\ctx$.
	\end{varvarenumerate}
	
	\item \emph{Substitution} $\tmtwo = \tmthree\esub\vartwo\tmfour$. Note that $\relunf\tmtwo\ctx = \relunf{\tmthree\esub\vartwo\tmfour}{\ctx}=_{\reflemmaeqp{relunf-prop}{four}} \relunf{\tmthree}{\ctxp{\ctxhole\esub\vartwo\tmfour}}$. So we can apply the \ih\ to $\tmthree$ and $\ctxp{\ctxhole\esub\vartwo\tmfour}$, from which the statement follows.
	\end{varenumerate}
\end{proof}

\begin{corollary}[Normal Form Property]
\label{coro:normal-form}
Let $\tm$ be a useful normal form. Then $\unf\tm$ is a $\beta$-normal form.
\end{corollary}

\begin{proof}
Let us apply \reflemma{nf-charac} to $\ctx \defeq \ctxhole$ and $\tmtwo \defeq \tm$. Since $\ctx = \ctxhole$, case 2 cannot hold, so that $\relunf\tm\ctx = \unf\tm$ is a $\beta$-normal form.
\end{proof}
\ben{At the end of the next section we will obtain the converse implication (\refcorollaryp{converse-prop-high-lev}{normal}), as a corollary of the strong projection property.}

\ben{Let us close this section with a comment. The auxiliary lemma for the normal form property is of a technical nature. Actually, it is a strong improvement (inspired by \cite{usef-constr}) over the sequence of lemmas we provided in the technical report \cite{EV}, that followed a different (worse) proof strategy. At a first reading the lemma looks very complex and it is legitimate to suspect that we did not fully understand the property we proved. We doubt, however, the existence of a much simpler proof, and believe that---despite the technical nature---our proof is compact. There seems to be an inherent difficulty given by the fact that useful reduction is a global notion, in the sense that it is not a rewriting rule closed by evaluation contexts, but it is defined by looking at the whole term at once. Its global character seems to prevent a simpler proof.}

\section{The Projection Property}
\label{sect:new-projection}
We now turn to the proof of the projection property. We already know that a single shallow $\db$-step projects to a single $\beta$-step (\reflemma{db-projection}). Therefore it remains to be shown that $\tolou$ $\db$-steps project on \lob\ steps. We do it contextually, in three steps: 

\begin{varenumerate}
	\item giving a (non-inductive) notion of \lou\ \emph{context}, and proving that if a redex $\ctx$ is a $\tolou$ redex then $\ctx$ is a \lou\ context.
	\item providing that \lou\ contexts admit a \emph{inductive} formulation.
	\item proving that \emph{inductive} \lou\ contexts unfold to inductive \lob\ contexts, that is where \lob\ steps take place.
\end{varenumerate}

As for the normal form property, the proof strategy is inspired by \cite{usef-constr}, and improves the previous proof in the technical report \cite{EV}.

\subsection{\lou\ Contexts} Remember that a term is \emph{neutral} if it is $\beta$-normal and is not of the form $\la\var\tmtwo$ (\ie\ it is not an abstraction). 

 \begin{definition}[\lou\ Contexts]
 \
  A context $\ctx$ is \lou\ if 
    \begin{varenumerate}
     	\item \emph{Leftmost}: whenever $\ctx = \ctxtwop{\tm\ctxthree}$ then $\relunf\tm\ctxtwo$ is neutral, and
			\item \emph{Outermost}: whenever $\ctx = \ctxtwop{\la\var\ctxthree}$ then $\ctxtwo$ is not applicative.
		\end{varenumerate}
	\end{definition}

The next lemma shows that $\tolou$ redexes take place in \lou\ contexts. In the last sub-case the proof uses the generalized form of the normal form property (\reflemma{nf-charac}).

\begin{lemma}[The Position of a $\tolou$\ Step is a \lou\ Context]
\label{l:loudb-are-ilou}
Let $\ctx: \tm \to \tmtwo$ be a useful step. If $\ctx$ is a $\tolou$ step then $\ctx$ is \lou.
\end{lemma}

\begin{proof}
Properties in the definition of \lou\ contexts:
		\begin{varenumerate}
			\item \emph{Outermost}: if $\ctx = \ctxtwop{\la\var\ctxthree}$ then clearly $\ctxtwo$ is not applicative, otherwise there is a useful redex (the $\todb$ redex involving $\la\var\ctxthree$) containing $\ctx$, \ie\ $\ctx$ is not the \lou\ redex, that is absurd.
			\item \emph{Leftmost}: suppose that the leftmost property of \lou\ contexts is violated for $\ctx$, and let $\ctx = \ctxtwop{\tmthree\ctxthree}$ be such that $\relunf\tmthree\ctxtwo$ is not neutral. We have that $\tmthree$ is $\tolou$-normal. Two cases:
			\begin{varvarenumerate}
			\item \emph{$\relunf\tmthree\ctxtwo$ is an abstraction}. Then $\ctxtwo$ is the position of a useful redex and $\ctx$ is not the \lou\ step, absurd. 
			\item \emph{$\relunf\tmthree\ctxtwo$ contains a $\beta$-redex}. By the contextual normal form property (\reflemma{nf-charac}) there is a useful redex in $\tm$ having its position in $\tmthree$, and so $\ctx$ is not the position of the $\tolou$-redex, absurd.\qedhere
			\end{varvarenumerate}
		\end{varenumerate}
\end{proof}

\subsection{Inductive \lou\ Contexts} We introduce the alternative characterization of \lou\ contexts, avoiding relative unfoldings. \ben{We call it \emph{inductive} because it follows the structure of the context $\ctx$, even if in the last clause the hypothesis might be  syntactically bigger than the conclusion. Something, however, always decreases: in the last clause it is the number of ES. The lemma that follows is indeed proved by induction over the number of ES in $\ctx$ and $\ctx$ itself.}
 
 \begin{definition}[Inductive \lou\ Contexts]
  A context $\ctx$ is \emph{inductively \lou} (or \ilou), if a judgment about it can be derived by using the following inductive rules:    
  \begin{center}
$\begin{array}{cccccc}
	\AxiomC{}
	\RightLabel{(ax-\ilou)}
	\UnaryInfC{$\ctxhole$ is \ilou}
	\DisplayProof
	&
	\AxiomC{$\ctx$ is \ilou}
	\AxiomC{$\ctx\neq\sctxp{\la\var\ctxtwo}$}
	\RightLabel{(@l-\ilou)}
	\BinaryInfC{$\ctx \tm$ is \ilou}
	\DisplayProof \\\\
	
		\AxiomC{$\ctx$ is \ilou}
	\RightLabel{($\l$-\ilou)}
	\UnaryInfC{$\la\var\ctx$ is \ilou}
	\DisplayProof 
	&
		\AxiomC{$\unf\tm$ is neutral}
		\AxiomC{$\ctx$ is \ilou}
		\RightLabel{(@r-\ilou)}
		\BinaryInfC{$\tm \ctx$ is \ilou}
		\DisplayProof 
		 \\\\
	\multicolumn{2}{c}{	  
	\AxiomC{$\ctx\isub\var{\unf\tm}$ is \ilou}
	\RightLabel{(ES-\ilou)}
	\UnaryInfC{$\ctx \esub\var\tm$ is \ilou}
	\DisplayProof
	}
\end{array}$
\end{center}  
 \end{definition}
 
 Let us show that \lou\ contexts are inductive \lou\ contexts.

\begin{lemma}[\lou\ Contexts are \ilou]
  \label{l:ilou-eq-rlou}
 Let $\ctx$ be a context. If $\ctx$ is \lou\ then $\ctx$ is inductively \lou.
\end{lemma}

\begin{proof}
  By lexicographic induction on $(\esmeas\ctx,\ctx)$, where $\esmeas\ctx$ is the number of ES in $\ctx$. Cases of $\ctx$: 
 \begin{varenumerate}
  \item \emph{Empty} $\ctx = \ctxhole$. Immediate.
  
  \item \emph{Abstraction} $\ctx = \la\var\ctxtwo$. By \ih\ $\ctxtwo$ is \ilou. Then $\ctx$ is \ilou\ by rule ($\l$-\ilou).
  
  \item \emph{Left Application} $\ctx = \ctxtwo \tmtwo$. By \ih, $\ctxtwo$ is \ilou. Note that $\ctxtwo \neq \sctxp{\la\var\ctxthree}$ otherwise the outermost property of \lou\ contexts would be violated, since $\ctxtwo$ appears in an applicative context. Then $\ctx$ is \ilou\ by rule (@l-\ilou).
  
  \item \emph{Right Application} $\ctx = \tmtwo \ctxtwo$. Then $\ctxtwo$ is \lou\ and so $\ctxtwo$ is \ilou\ by \ih. By the leftmost property of \lou\ contexts $\unf\tmtwo = \relunf\tmtwo\ctxhole$ is neutral. Then  $\ctx$ is \ilou\ by rule (@r-\ilou).
  
  \item \emph{Substitution} $\ctx = \ctxtwo\esub\var\tmtwo$. We  prove that $\ctxtwo\isub\var{\unf\tmtwo}$ is \lou\ and obtain that $\ctx$ is \ilou\ by applying the \ih\ (first component decreases) and rule (ES-\ilou). There are two conditions to check:
    \begin{varvarenumerate}
     \item \emph{Outermost}: consider $\ctxtwo\isub\var{\unf\tmtwo} = \ctxthreep{\la\vartwo\ctxfour}$.  Note that the abstraction $\l\vartwo$ comes from an abstraction of $\ctxtwo$ to which $\isub\var{\unf\tmtwo}$ has been applied, because $\unf\tmtwo$ cannot contain a context hole. Then $\ctxtwo = \ctx_2\ctxholep{\la\vartwo\ctx_3}$ with $\ctx_2\isub\var{\unf\tmtwo} =\ctxthree$ and $ \ctx_3\isub\var{\unf\tmtwo} =\ctxfour$. By hypothesis $\ctx$ is \lou\ and so $\ctx_2\esub\var{\unf\tmtwo}$ is not applicative. Then $\ctx_2$ and thus $\ctx_2\isub\var{\unf\tmtwo}$ are not applicative.
     
     \item \emph{Leftmost}: consider $\ctxtwo\isub\var{\unf\tmtwo} = \ctxthreep{\tmthree\ctxfour}$. Note that the application $\tmthree\ctxfour$ comes from an application of $\ctxtwo$ to which $\isub\var{\unf\tmtwo}$ has been applied, because $\unf\tmtwo$ cannot contain a context hole. Then  $\ctxtwo = \ctx_2\ctxholep{\tmthree'\ctx_3}$ with $\ctx_2\isub\var{\unf\tmtwo} =\ctxthree$, $\tmthree'\isub\var{\unf\tmtwo} = \tmthree$, and $ \ctx_3\isub\var{\unf\tmtwo} =\ctxfour$. We have to show that $\relunf\tmthree\ctxthree$ is neutral. Note that 
	
     \begin{center}
      $\relunf\tmthree\ctxthree = \relunf{\tmthree'\isub\var{\unf\tmtwo}}{\ctx_2\isub\var{\unf\tmtwo}} =_{\reflemmaeqp{relunf-prop}{four}} \relunf{\tmthree'}{\ctx_2\esub\var{\tmtwo}}$
     \end{center}
     Since $\ctx = \ctx_2\ctxholep{\tmthree'\ctx_3}\esub\var{\tmtwo}$ is \lou\ by hypothesis, we obtain that $\relunf{\tmthree'}{\ctx_2\esub\var{\tmtwo}}$ is neutral.
    \end{varvarenumerate}
 \end{varenumerate}
\end{proof}

\subsection{Unfolding \lou\ steps} Finally, we show that inductive \lou\ contexts unfold to inductive \lob\ contexts (\refdef{ilob-ctx}, page \pageref{def:ilob-ctx}).

\begin{lemma}[\ilou\ unfolds to \ilob]
	\label{l:useful-projection}
If $\ctx$ is an \ilou\ context then $\unf\ctx$ is a \ilob\ context.
\end{lemma}

\begin{proof}

By lexicographic induction on $(\esmeas\ctx,\ctx)$. Cases of $\ctx$. 
 	\begin{varenumerate}
		\item \emph{Empty} $\ctx = \ctxhole$. Then $\unf\ctx = \unf\ctxhole = \ctxhole$ is \ilob.
		\item \emph{Abstraction} $\ctx = \la\var \ctxtwo$. By definition of unfolding $\unf\ctx = \la\var\unf\ctxtwo$, and by the \ilou\ hypothesis $\ctxtwo$ is \ilou. Then by \ih\ (second component), $\unf\ctxtwo$ is \ilob, and so $\unf\ctx$ is \ilob\ by rule ($\l$-\ilob).
		\item \emph{Left Application} $\ctx = \ctxtwo \tm$. By definition of unfolding, $\unf\ctx = \unf\ctxtwo \unf\tm$, and by the \ilou\ hypothesis $\ctxtwo$ is \ilou. Then by \ih\ (second component), $\unf\ctxtwo$ is \ilob. Moreover, $\ctxtwo\neq\sctxp{\la\var\ctxthree}$, that implies $\ctxtwo=\sctxp{\ctxthree\tmtwo}$, $\ctxtwo=\sctxp{\tmtwo\ctxthree}$, or $\ctxtwo=\sctx$. In all such cases we obtain $\unf\ctxtwo\neq \la\var\gctx$. Therefore, $\unf\ctx$ is \ilob\ by rule (@l-\ilob).
		\item \emph{Right Application} $\ctx = \tm \ctxtwo$. By definition of unfolding $\unf\ctx = \unf\tm \unf\ctxtwo $, and by the \ilou\ hypothesis $\ctxtwo$ is \ilou\ and $\unf\tm$ is neutral. Then by \ih\ (second component), $\ctxtwo$ is \ilob, and so $\ctx$ is \ilob\ by rule (@r-\ilob).
  		\item \emph{Substitution} $\ctx = \ctxtwo\esub\var\tm$. By definition of unfolding $\unf\ctx = \unf\ctxtwo\isub\var{\unf\tm}$, and by the \ilou\ hypothesis $\ctxtwo\isub\var{\unf\tm}$ \ilou. Then by \ih\ (first component) $\unf{\ctxtwo\isub\var{\unf\tm}}$ is \ilob, that is equivalent to the statement because $\unf\ctxtwo\isub\var{\unf\tm} =_{\reflemmaeqp{relunf-prop}{zero}} \unf{\ctxtwo\isub\var{\unf\tm}}$.\qedhere
  	\end{varenumerate}
\end{proof}

\noindent The projection property of \lou\ $\todb$ steps now follows easily:

\begin{theorem}[Strong Projection Property]
	\label{tm:projection}
	Let $\tm$ be a \lsc\ term and $\tm\tolou\tmtwo$ a $\todb$ step. Then $\unf{\tm}\toblo\unf{\tmtwo}$.
\end{theorem}

\begin{proof}
	 $\tolou$-steps take place in \lou\ contexts (\reflemma{loudb-are-ilou}), \lou\ contexts are inductive \lou\ contexts (\reflemma{ilou-eq-rlou}), that unfold to \ilob\ contexts (\reflemma{useful-projection}), which is where $\toblo$ steps take place (\reflemma{toblo-ilob}).
\end{proof}

\ben{The following corollary shows that in fact for our high-level implementation system the converse statements of the normal form and projection properties are also valid, even if they are not needed for the invariance result.

\begin{corollary}
\label{c:converse-prop-high-lev}\hfill 
 \begin{varenumerate}
  \item \label{p:converse-prop-high-lev-normal}
    \emph{Converse Normal Form}: if $\unf\tm$ is a $\beta$-normal form then
    $\tm$ is a useful normal form.
  \item 
    \emph{Converse Projection}: if $\unf\tm\toblo\tmtwo$ then there exists $\tmthree$ such that $\tm\tolou\tmthree$ with $\unf\tmthree=\unf\tm$ or $\unf\tmthree=\tmtwo$.
\end{varenumerate}
\end{corollary}

\begin{proof}\hfill
 \begin{varenumerate}
  \item Suppose that $\tm$ has a useful redex. By the strong projection property (\reftm{projection}) the $\tolou$ redex  in $\tm$ cannot be a $\db$-redex, otherwise $\unf\tm$ is not $\beta$-normal. Similarly, by \reflemma{useful-prom-on-beta} it cannot be a $\ls$-redex. Absurd.    
  \item By the normal form property (\refcoro{normal-form}) $\tm$ has a useful redex, otherwise $\unf\tm$ is $\beta$-normal, that is absurd. Then the statement follows from the strong projection property (\reftm{projection}) or from \reflemma{ls-unfolding}.\qedhere   
\end{varenumerate} 
\end{proof}
}

\ben{
\noindent For the sake of completeness, let us also point out that the converse statements of \reflemma{loudb-are-ilou} (\ie\ useful steps taking place in \lou\ contexts are $\tolou$ steps) and \reflemma{useful-projection} (inductive \lou\ contexts are \lou\ contexts) are provable. We omitted them to lighten the technical development of the paper.
}

\section{Standard Derivations and the Subterm Property}
\label{sect:standard}
Here we introduce standard derivations and show that they have the
subterm property.

\subsection{Standard derivation}
They are defined on top of the concept of residual of a redex. For the
sake of readability, we use the concept of residual without formally
defining it (see \cite{non-standard-preprint} for details).

\begin{definition}[Standard Derivation]
A derivation $\deriv:\ctx_1;\ldots;\ctx_n$ is \deff{standard} if
$\ctx_i$ is not the residual of a LSC redex $\ctxtwo\leftout\ctx_j$
for every $i\in\set{2,\ldots,n}$ and $j<i$.
\end{definition}
The same definition where terms are ordinary $\l$-terms and redexes are $\beta$-redexes gives the
ordinary notion of standard derivation in the theory of $\l$-calculus.

Note that any single reduction step is standard. Then,
notice that standard derivations select redexes in a left-to-right and
outside-in way, but they are not necessarily \lo. For instance, the
derivation
$$
(\ap{(\l\var.\vartwo)}\vartwo)\esub\vartwo\varthree \tols
(\ap{(\l\var.\varthree)}\vartwo)\esub\vartwo\varthree \tols
(\ap{(\l\var.\varthree)}\varthree)\esub\vartwo\varthree
$$
is standard even if the \lo\ redex (\ie\ the $\db$-redex on $\var$) is
not reduced. The extension of the derivation with
$(\ap{(\l\var.\varthree)}\varthree)\esub\vartwo\varthree\todb
\varthree\esub\var\varthree\esub\vartwo\varthree$ is not
standard. Last, note that the position of a $\ls$-step
(\refdef{redex-position}, page \pageref{def:redex-position}) is given
by the substituted occurrence and not by the ES, that is
$(\ap\var\vartwo)\esub\var\tmtwo\esub\vartwo\tm \tols
(\ap\var\tm)\esub\var\tmtwo\esub\vartwo\tm \tols
(\ap\tmtwo\tm)\esub\var\tmtwo\esub\vartwo\tm$ is not standard. We have
the following expected result.

\begin{theorem}[\cite{non-standard-preprint}]
\label{thm:lo-standard}
 \lo\ derivations (of the LSC) are standard.
\end{theorem}

The \emph{subterm property} states that at any point of a
derivation $\deriv:\tm\to^*\tmtwo$ only subterms of the initial term
$\tm$ are duplicated. Duplicable subterms are identified by \emph{boxes}, and we need a technical lemma about them.

\subsection{Boxes, Invariants, and Subterms}
A \emph{box} is the argument of an application or the content of an
explicit substitution.  In the graphical representation of $\l$-terms
xwith ES, our boxes correspond to explicit boxes for promotions.

\begin{definition}[Box Context, Box Subterm]
  Let $\tm$ be a term. \deff{Box contexts} (that are not necessarily
  shallow) are defined by the following grammar, where $\gctx$ is a
  general context (\ie\ not necessarily shallow):
  \begin{center}
    $\begin{array}{rcl} 
      \bctx &\grameq& \ap\tm\ctxhole\midd
      \tm\esub\var\ctxhole\midd \gctxp\bctx.
    \end{array}$
  \end{center}
  A \deff{box subterm} of $\tm$ is a term $\tmtwo$
  such that $\tm=\bctxp\tmtwo$ for some box context $\bctx$.
\end{definition}

In the simple case of linear head reduction
\cite{DBLP:conf/rta/AccattoliL12}, the subterm property follows from
the invariant that evaluation never substitutes in box-subterms, and so
ES---that are boxes---contain and thus substitute only subterms of the
initial term. For linear \lo\ reduction, and more generally for
standard derivations in the LSC, the property is a bit more subtle to
establish. Substitutions can in fact act inside boxes, but a more
general invariant still holds: whenever a standard derivation
substitutes/evaluates in a box subterm $\tmtwo$, then $\tmtwo$ will no
longer be substituted. This is a consequence of selecting redexes
according to $\leftout$. Actually, the invariant is stated the other
way around, saying that the boxes on the right---that are those
involved in later duplications---are subterms of the initial term.

\begin{lemma}[Standard Derivations Preserve Boxes on Their Right]\label{l:subterm-aux}
  Let $\deriv:\tm_0\to^k\tm_k\to\tm_{k+1}$ be a standard derivation
  and let $\ctx_k$ be the last contracted redex, $k\geq 0$, and
  $\bctx\prefix\tm_{k+1}$ be a box context
  such that $\ctx_k\leftout\bctx$. Then the box subterm $\tmtwo$ identified
  by $\bctx$ (\ie\ such that $\tm_{k+1}=\bctxp{\tmtwo}$) is a box subterm
  of $\tm_0$.
\end{lemma}

\begin{proof}
  By induction on $k$. If $k=0$ the statement trivially
  holds. Otherwise, consider the redex
  $\ctx_{k-1}:\tm_{k-1}\to\tm_k$. By \ih\ the statement holds wrt box
  contexts $\bctx\prefix\tm_k$ such that
  $\ctx_{k-1}\leftout\bctx$. The proof analyses the position of
  $\ctx_k$ with respect to the position $\ctx_{k-1}$, often
  distinguishing between the left-to-right $\leftright$ and the
  outside-in $\outin$ suborders of $\leftout$. Cases:
  \begin{varenumerate}
  \item
    \emph{$\ctx_k$ is equal to $\ctx_{k-1}$}. Clearly if
    $\ctx_{k-1}=\ctx_k\leftright\bctx$ then the statement holds
    because of the \ih\ (reduction does not affect boxes on the right
    of the hole of the position). We need to check the box contexts
    such that $\ctx_k\outin\bctx$. Note that $\ctx_{k-1}$ cannot be a
    $\todb$ redex, because if
    $\tm_{k-1}=\ctx_{k-1}\ctxholep{\sctxp{\l\var.\tmthree}\tmfour}
    \todb \ctx_{k-1}\ctxholep{\sctxp{\tmthree\esub\var\tmfour}}=\tm_k$
    then $\ctx_{k-1}=\ctx_k$ is not the position of any redex in
    $\tm_k$. Hence, $\ctx_{k-1}$ is a $\tols$ step and there are two
    cases. If it substitutes a:
    \begin{varvarenumerate}
    \item \emph{Variable}, \ie\ the sequence of steps $\ctx_{k-1};\ctx_k$ is
      
      \begin{center}
	$\tm_{k-1} = \ctx_{k-1}\ctxholep\var \tols \ctx_{k-1}\ctxholep\vartwo \tols \ctx_{k-1}\ctxholep\tmthree = \tm_{k+1}$
      \end{center}
      
      Then all box subterms of $\tmthree$ trace back to box subterms that also appear on the right of $\ctx_{k-1}$ in $\tm_k$ and so they are box subterms of $\tm_0$ by \ih
    \item \emph{$\db$-redex}, \ie\ the sequence of steps $\ctx_{k-1};\ctx_k$ is
      
      \begin{center}
	$\tm_{k-1} = \ctx_{k-1}\ctxholep\var \tols \ctx_{k-1}\ctxholep{\sctxp{\l\vartwo.\tmthree}\tmfour} \tols  \ctx_{k-1}\ctxholep{\sctxp{\tmthree\esub\vartwo\tmfour}} = \tm_{k+1}$			 
      \end{center}
      
      Then all box subterms of $\sctxp{\tmthree\esub\vartwo\tmfour}$ trace back to box subterms of $\sctxp{\l\vartwo.\tmthree}\tmfour$, hence they are in $\tm_k$, and so they are box subterms of $\tm_0$ by \ih
    \end{varvarenumerate}
  \item
    \emph{$\ctx_{k-1}$ is internal to $\ctx_k$},
    \ie\ $\ctx_k\outin\ctx_{k-1}$ and $\ctx_k\neq\ctx_{k-1}$. This
    case is only possible if $\ctx_k$ has been created \emph{upwards}
    by $\ctx_{k-1}$, otherwise the derivation would not be
    $\leftout$-standard.  There are only two possible cases of
    creations \emph{upwards}:
    \begin{varvarenumerate}
    \item
      $\db$ creates $\db$, \ie\ $\ctx_{k-1}$ is 
      \begin{center}
	$\tm_{k-1}=\ctx_{k-1}\ctxholep{\sctxp{\l\vartwo.\sctxtwop{\la\varthree\tmthree}}\tmfour}\todb\ctx_{k-1}\ctxholep{\sctxp{\sctxtwop{\la\varthree\tmthree}\esub\vartwo\tmfour}}=\tm_k$
      \end{center}
      and $\ctx_{k-1}$ is applicative, that is
      $\ctx_{k-1}=\ctx_k\ctxholep{\sctxthreep\cdot\tmfive}$, so that
      $\ctx_k$ is
      \begin{center}
	$\tm_k=\ctx_k\ctxholep{\sctxthreep{\sctxp{\sctxtwop{\l\varthree.\tmthree}\esub\vartwo\tmfour}}\tmfive} \todb \ctx_k\ctxholep{\sctxthreep{\sctxp{\sctxtwop{\tmthree\esub\varthree\tmfive}\esub\vartwo\tmfour}}}=\tm_{k+1}$
      \end{center}
      The box subterms of $\sctxthree$ and $\tmfive$ (including
      $\tmfive$ itself) are box subterms of $\tm_k$ with box context
      $\bctx$ such that $\ctx_{k-1}\leftright\bctx$ and so they are
      box subterms of $\tm_0$ by \ih The other box subterms of
      $\sctxthreep{\sctxp{\sctxtwop{\tmthree\esub\varthree\tmfive}\esub\vartwo\tmfour}}$
      are instead box subterms of
      $\sctxp{\l\vartwo.\sctxtwop{\l\varthree.\tmthree}}\tmfour$,
      \ie\ of box context $\bctx$ such that $\ctx_{k-1}\outin\bctx$,
      and so they are box subterms of $\tm_0$ by \ih
    \item
      $\lssym$ creates $\db$, \ie\ $\ctx_{k-1}$ is
      \begin{center}
	$\tm_{k-1}=\ctx_{k-1}\ctxholep{\var}\todb\ctx_{k-1}\ctxholep{\sctxp{\l\vartwo.\tmthree}}=\tm_k$ 
      \end{center}      
      and $\ctx_{k-1}$ is applicative,
      \ie\ $\ctx_{k-1}=\ctx_k\ctxholep{\sctxtwop\cdot\tmfour}$ so that
      $\ctx_k$ is
      \begin{center}
	$\ctx_k\ctxholep{\sctxtwop{\sctxp{\l\vartwo.\tmthree}}\tmfour} \todb \ctx_k\ctxholep{\sctxtwop{\sctxp{\tmthree\esub\vartwo\tmfour}}}$			 
      \end{center}
      The box subterms of $\sctxtwo$ and $\tmfour$ (including
      $\tmfour$ itself) are box subterms of the ending term of
      $\ctx_{k-1}$ whose box context $\bctx$ is
      $\ctx_{k-1}\leftright\bctx$ and so they are box subterms of
      $\tm_0$ by \ih The other box subterms of
      $\sctxtwop{\sctxp{\tmthree\esub\vartwo\tmfour}}$ are also box
      subterms of
      $\sctxp{\l\vartwo.\sctxtwop{\l\varthree.\tmthree}}\tmfour$ and
      so they are box subterms of $\tm_0$ by \ih
    \end{varvarenumerate}
  \item
    \emph{$\ctx_k$ is internal to $\ctx_{k-1}$},
    \ie\ $\ctx_{k-1}\outin\ctx_k$ and $\ctx_k\neq\ctx_{k-1}$. Cases of
    $\ctx_{k-1}$:
    \begin{varvarenumerate}
    \item \label{p:subterm}
      \emph{$\db$-step}, \ie\ $\ctx_{k-1}$ is
      \begin{center}
	$\tm_{k-1}= \ctx_{k-1}\ctxholep{\sctxp{\l\var.\tmthree}\tmfour} \todb \ctx_{k-1}\ctxholep{\sctxp{\tmthree\esub\var\tmfour}}=\tm_k$			 
      \end{center}
      Then the hole of $\ctx_k$ is inside
      $\sctxp{\tmthree\esub\var\tmfour}$. Box subterms identified by a
      box context $\bctx$ such that $\ctx_k\leftright\bctx$ in
      $\tm_{k+1}$ are also box subterms of $\tm_k$, and so the
      statement follows from the \ih For box subterms identified by a
      box context $\bctx$ of $\tm_{k+1}$ such that $\ctx_k\outin\bctx$
      we have to analyze $\ctx_k$. Suppose that $\ctx_k$ is a:
      \begin{varitemize}
      \item
        \emph{$\db$-step}. Note that in a root $\db$-step (\ie\ at
        top-level) all the box subterms of the reduct are box subterms
        of the redex. In this case the redex is contained in
        $\sctxp{\tmthree\esub\var\tmfour}$ and so by \ih\ all such box
        subterms are box subterms of $\tm_0$.	
      \item
        \emph{$\lssym$-step}, \ie\ $\ctx_k$ has the form
        $\tm_k=\ctx_k\ctxholep{\var} \tols
        \ctx_k\ctxholep{\tmfive}=\tm_{k+1}$. In $\tm_k$, $\tmfive$ is
        identified by a box context $\bctx$ such that
        $\ctx_k\leftright\bctx$. From $\ctx_{k-1}\leftout\ctx_k$ we
        obtain $\ctx_{k-1}\leftout\bctx$ and so all box subterms of
        $\tmfive$ are box subterms of $\tm_0$ by \ih
      \end{varitemize}  
    \item
      \emph{$\lssym$-step}: $\ctx_{k-1}$ is
      $\tm_{k-1}=\ctx_{k-1}\ctxholep{\var} \todb
      \ctx_{k-1}\ctxholep{\tmfive}=\tm_k$. It is analogous to the
      $\db$-case: $\ctx_k$ takes place inside $\tmfive$, whose box
      subterms are box subterms of $\tm_0$, by \ih If $\ctx_k$ is a
      $\db$-redex then it only rearranges constructors in $\tmfive$
      without changing box subterms, otherwise it substitutes
      something coming from a substitution that is on the right of
      $\ctx_{k-1}$ and so whose box subterms are box subterms of
      $\tm_0$ by \ih
    \end{varvarenumerate}    
  \item
    \emph{$\ctx_k$ is on the left of $\ctx_{k-1}$},
    \ie\ $\ctx_{k-1}\leftright\ctx_k$. For $\bctx$ such that
    $\ctx_k\leftright\bctx$, the statement follows from the \ih,
    because there is a box context $\bctxtwo$ in $\tm_k$ such that
    $\ctx_{k-1}\leftright\bctxtwo$ and identifying the same box
    subterm of $\bctx$. For $\bctx$ such that $\ctx_k\outin\bctx$, we
    reason as in case \ref{p:subterm}.\qedhere
  \end{varenumerate}
\end{proof}
\noindent From the invariant, one easily obtains the subterm property.
\begin{corollary}[Subterm]
  \label{coro:subterm}
  Let $\deriv:\tm\to^k\tmtwo$ be a standard derivation. Then every
  $\tols$-step in $\deriv$ duplicates a subterm of $\tm$.
\end{corollary}
\begin{proof}
  By induction on $k$. If $k=0$ the statement is evidently
  true. Otherwise, by \ih\ in $\deriv:\tm\to^{k-1}\tmthree$ every
  $\tols$-step duplicated a subterm of $\tm$. If the next step is a
  $\db$-step the statement holds, otherwise it is a $\ls$-step that by
  \reflemma{subterm-aux} duplicates a subterm of $\tmtwo$ which is a
  box subterm, and so a subterm, of $\tm$.
\end{proof}

\subsection{Technical Digression: Shallowness and Standardization}
 In \cite{non-standard-preprint} it is shown that in the full
\lsc\ standard derivations are complete, \ie\ that whenever
$\tm\to^*\tmtwo$ there is a standard derivation from $\tm$ to
$\tmtwo$. The shallow fragment does not enjoy such a standardization
theorem, as the residuals of a shallow redex need not be shallow. This
fact however does not clash with the technical treatment in this
paper.  The shallow restriction is indeed compatible with
standardization in the sense that:
\begin{varenumerate}
\item 
  \emph{The linear \lo\ strategy is shallow}: if the initial term is a
  $\l$-term then every redex reduced by the linear \lo\ strategy is
  shallow (every non-shallow redex $\ctx$ is contained in a
  substitution, and every substitution is involved in an outer redex
  $\ctxtwo$);
\item 
  \emph{$\leftout$-ordered shallow derivations are standard}: any
  strategy picking shallow redexes in a left-to-right and outside-in
  fashion does produce standard derivations (it follows from the easy
  fact that a shallow redex $\ctx$ cannot turn a non-shallow redex
  $\ctxtwo$ such that $\ctxtwo\leftout\ctx$ into a shallow redex).
\end{varenumerate}
Moreover, the only redex swaps we will consider
(\reflemma{useless-pers}) will produce shallow residuals.\medskip

\subsection{Shallow Terms}
Let us conclude the section with a further invariant of standard
derivations. It is not needed for the invariance result, but it sheds
some light on the shallow subsystem under study. Let a term be
\deff{shallow} if its substitutions do not contain substitutions. The
invariant is that if the initial term is a $\l$-term then standard
shallow derivations involve only shallow terms.

\begin{lemma}[Shallow Invariant]
\label{l:sh-inv}
Let $\tm$ be a $\l$-term and $\deriv:\tm\to^k\tmtwo$ be a standard
derivation. Then $\tmtwo$ is a shallow term.
\end{lemma}

\begin{proof}
  By induction on $k$. If $k=0$, the statement is evidently
  true. Otherwise, by \ih\ every explicit substitution in $\tmthree$,
  where $\deriv:\tm\to^{k-1}\tmthree$, contains a $\l$-term. We distinguish the two cases concerning the sort of the next step $\tmthree\to\tmtwo$:
  \begin{varenumerate}
  \item \emph{$\ls$-step}. By the subterm property and the fact that $\tm$ has no ES, the step duplicates
    a term without substitutions, and---since reduction is shallow---it does not put the duplicated term in a
    substitution. Therefore, every substitution of $\tmtwo$ corresponds uniquely to a substitution of $\tmthree$ with the same content. Then $\tmtwo$ is a shallow term by \ih
  \item \emph{$\db$-step}. It is easily seen that the argument of the
    $\db$-step is on the right of the previous step, so that by
    \reflemma{subterm-aux} it contains a (box) subterm of $\tm$. Then,
    the substitution created by the $\db$-step contains a subterm of
    $\tm$, that is an ordinary $\l$-term by hypothesis. The step does
    not affect any other substitution, because reduction is shallow,
    and so $\tmtwo$ is a shallow term.\qedhere
  \end{varenumerate}
\end{proof}

\noindent In this paper we state many properties relative to derivations whose
initial term is a $\l$-term. The shallow invariant essentially means
that all these properties may be generalized to (standard) derivations
whose initial term is shallow. There is, however, a subtlety that
justifies our formulation with respect to
$\l$-terms. \reflemma{sh-inv}, indeed, does not hold if one simply
assume that $\tm$ is a shallow term. Consider for instance
$(\la\var\var)(\vartwo\esub\vartwo\varthree)$, that is shallow and
that reduces in one step (thus via a standard derivation) to
$\var\esub\var{\vartwo\esub\vartwo\varthree}$, which is not
shallow. The subtlety is that the position $\ctx$ of the first step of
the standard derivation has to be a $\leftout$-majorant of the
position of any ES in the term. For the sake of simplicity, we
prefered to assume that the initial term has no ES.

Note also that this \reflemma{sh-inv} is the only
point of this section relying on the assumption that reduction is
shallow (the hypothesis of the derivation being standard is also necessary, consider
$(\l\var.\var) ((\l\vartwo.\vartwo)\varthree)\todb (\l\var.\var)
(\vartwo\esub\vartwo\varthree)\todb\var\esub\var{\vartwo\esub\vartwo\varthree}$).

\section{\lou\ Derivations are Standard}
\label{sect:subterm-via-standard}
\emph{Notation}: to avoid ambiguities, in this section we use
$\redex,\redextwo,\redexthree$ for redexes,
$\redex',\redextwo',\redexthree'$ for their residuals, and
$\ctx,\ctxtwo,\ctxthree$ for shallow contexts.\medskip

\lo\ derivations are standard (\refthm{lo-standard}), and this is
expected. A priori, instead, \lou\ derivations may not be standard, if
the reduction of a useful redex $\redex$ could turn a useless redex
$\redextwo\leftout\redex$ into a useful redex.  Luckily, this is not
possible, \ie\ uselessness is stable by reduction of
$\leftout$-majorants, as proved by the next lemma.

We first need to recall two properties of the standardization order
$\leftout$ relative to residuals, called \emph{linearity} and
\emph{enclave}. They are two of the axioms of the axiomatic theory of
standardization developed by \mellies\ in his PhD thesis
\cite{phdmellies}, that in turn is a refinement of a previous
axiomatization by Gonthier, \levy, and
\mellies\ \cite{DBLP:conf/lics/GonthierLM92} (that did not include the
enclave axiom). The axioms of \mellies' axiomatic theory have been
proved to hold for the LSC by Accattoli, Bonelli, Kesner, and Lombardi
in \cite{non-standard-preprint}. The two properties essentially
express that if $\redex\leftout\redextwo$ then $\redextwo$ cannot act
on $\redex$. Their precise formulation follows.

\begin{lemma}[\cite{non-standard-preprint}]
If $\redex\leftout\redextwo$ then  
\begin{varenumerate}
	\item \emph{Linearity}: $\redex$ has a unique residual $\redex'$ after $\redextwo$;
	\item \emph{Enclave}: two cases 
		\begin{enumerate}
			\item \emph{Creation}: if $\redextwo$ creates a redex $\redexthree$ then $\redex'\leftout \redexthree$;
			\item \emph{Nesting}: If $\redextwo\leftout \redexthree$ and $\redexthree'$ is a residual of $\redexthree$ after $\redextwo$ then $\redex'\leftout \redexthree'$.
		\end{enumerate}
\end{varenumerate} 
\end{lemma}

Now we can prove the key lemma of the section.

\begin{lemma}[Useless Persistence]
	\label{l:useless-pers}
	Let $\redex:\tm\tols\tmtwo$ be a useless redex, $\redextwo:\tm\to\tmfour$ be a useful redex such that $\redex\leftout\redextwo$, and $\redex'$ be the unique residual of $\redex$ after $\redextwo$ (uniqueness follows from the just recalled property of \emph{linearity} of $\leftout$). Then
	\begin{varenumerate}
		\item \label{p:useless-pers-one} $\redex'$ is shallow and useless;
		\item if $\redextwo$ is \lou\ and $\redexthree$ is the \lou\ redex in $\tmfour$ then $\redex'\leftout\redexthree$.
	\end{varenumerate}
\end{lemma}

\begin{proof}\hfill
	\begin{varenumerate}
	\item
	Let $\redex:\ctxtwop{\ctxp\var\esub\var\tmthree} \tols \ctxtwop{\ctxp\tmthree\esub\var\tmthree}$. According to \refdef{useful-redex} (page \pageref{def:useful-redex}), a $\ls$-redex is useless when it is not useful. Then, uselessness of $\redex$ implies that $\relunf{\tmthree}{\ctxtwo}$ is a normal $\l$-term (otherwise the \emph{relative duplication} clause in the definition of useful redexes would hold) and if $\relunf{\tmthree}{\ctxtwo}$ is an abstraction then $\ctxtwop{\ctx\esub\var\tmthree}$ is not an applicative context (otherwise \emph{relative creation} would hold). 
	
	Note that $\ls$-steps cannot change the useless nature of $\redex$. To change it, in our case, they should be able to change the abstraction/normal nature of $\relunf{\tmthree}{\ctxtwo}$ or to change the applicative nature of $\ctxtwop{\ctx\esub\var\tmthree}$, but both changes are impossible: unfoldings, and thus $\relunf{\tmthree}{\ctxtwo}$, cannot be affected by $\ls$-steps (formally, an omitted generalization of \reflemma{relunf-prop} is required), and $\ls$-steps cannot provide/remove arguments to/from context holes. So, in the following we suppose that $\redextwo$ is a $\db$-redex. 
	
	By induction on $\ctxtwo$, the external context of $\redex$. Cases:
	\begin{varvarenumerate}
		\item \emph{Empty context $\ctxhole$}. Consider $\redextwo$, that necessarily takes place in the context $\ctx$,
		\begin{center}
		 $\redextwo: \ctxp\var\esub\var\tmthree \to \ctxtwop\var\esub\var\tmthree$
		\end{center}

		The only way in which the residual $\redex':\ctxtwop\var\esub\var\tmthree \tols \ctxtwop\tmthree\esub\var\tmthree$ of $\redex$ can be useful is if $\redextwo$ turned the non-applicative context $\ctx$ into an applicative context $\ctxtwo$, assuming that $\unf{\tmthree}$ is an abstraction. It is easily seen that this is possible only if $\redextwo\leftout\redex$, against hypothesis. Namely, only if $\ctx=\ctxthreep{\sctxp{\l\vartwo.\sctxtwo}\tmfour}$ and $\ctxthree$ is applicative, so that $\redextwo$ is 
			\begin{center}
			$\ctxthreep{\sctxp{\l\vartwo.\sctxtwop\var}\tmfour}\esub\var\tmthree \todb
			\ctxthreep{\sctxp{\sctxtwop\var\esub\vartwo\tmfour}}\esub\var\tmthree $			 
			\end{center}

		with $\ctxtwo=\ctxthreep{\sctxp{\sctxtwo\esub\vartwo\tmfour}}$ applicative context.

		\item \emph{Inductive cases}:
		\begin{varvarvarenumerate}
			\item \emph{Abstraction}, \ie\ $\ctxtwo=\l\vartwo.\ctxthree$. Both redexes $\redex$ and $\redextwo$ take place under the outermost abstraction, so the statement follows from the \ih
			
			\item \emph{Left of an application}, \ie\ $\ctxtwo=\ctxthree\tmfive$. Note that $\redextwo$ cannot be the eventual root $\db$-redex (\ie\ if $\ctxthree$ is of the form $\sctxp{\l\vartwo.\ctxfour}$ then $\redextwo$ is not the $\db$-redex involving $\l\vartwo$ and $\tmfive$), because this would contradict $\redex\leftout\redextwo$. If the redex $\redextwo$ takes place in $\ctxthreep{\ctxp\var\esub\var\tmthree}$ then we use the \ih Otherwise $\redextwo$ takes place in $\tmfive$, the two redexes are disjoint, and commute. Evidently, the residual $\redex'$ of $\redex$ after $\redextwo$ is still shallow and useless.
			
			\item \emph{Right of an application}, \ie\ $\ctxtwo=\tmfive\ctxthree$. Since $\redex\leftout\redextwo$, $\redextwo$ necessarily takes place in $\ctxthree$, and the statement follows from the \ih
			
			\item \emph{Substitution}, \ie\ $\ctxtwo=\ctxthree\esub\vartwo\tmfive$. Both redexes $\redex$ and $\redextwo$ take place under the outermost explicit substitution $\esub\vartwo\tmfive$, so the statement follows from the \ih
		\end{varvarvarenumerate}
	\end{varvarenumerate}

	\item Assume that $\redextwo$ is \lou. By \refpoint{useless-pers-one}, the unique residual $\redex'$ of any useless redex $\redex\leftout\redextwo$ is useless, so that the eventual next \lou\ redex $\redexthree$ either has been created by $\redextwo$ or it is the residual of a redex $\redexthree^*$ such that $\redextwo\leftout\redexthree^*$. The enclave property guarantees that $\redex'\leftout\redexthree$.\qedhere
\end{varenumerate}
\end{proof}

\noindent Now an easy iterated application of the previous lemma shows that \lou\ derivations are standard.

\begin{proposition}
\label{prop:lou-der-are-standard}
Every \lou\ derivation is standard.
\end{proposition}

\begin{proof}
	By induction on the length $k$ of a \lou\ derivation $\deriv$. If $k = 0$ then the statement trivially holds. If $k>0$ then $\deriv$ writes as $\derivtwo; \redex$ where $\derivtwo$ by \ih\ is standard. Let $\derivtwo$ be $\redex_1;\ldots\redex_k$ and $\redex_i:\tm_{i}\to\tm_{i+1}$ with $i\in\set{1,\ldots,k}$. If $\derivtwo; \redex$ is not standard there is a term $\tm_i$ and a redex $\redextwo$ of $\tm_i$ such that 
	\begin{varenumerate}
		\item $\redex$ is a residual of $\redextwo$ after $\redex_i;\ldots;\redex_k$;
		\item $\redextwo\leftout\redex_i$.
	\end{varenumerate}
	Since $\redex_i$ is \lou, $\redextwo$ is useless. Then, iterating the application of \reflemma{useless-pers} to the sequence $\redex_i;\ldots;\redex_k$, we obtain that $\redex$ is useless, which is absurd. Then $\deriv = \derivtwo; \redex$ is standard.
\end{proof}

We conclude by applying \refcoro{subterm}.
\begin{corollary}[Subterm]
\label{coro:subterm-for-lou}
\lou\ derivations have the subterm property.
\end{corollary}

\section{The Local Boundedness Property, via Outside-In Derivations}
\label{sect:nested}
In this section we show that \lou\ derivations have the local
boundedness property. We introduce yet another abstract property, the notion of \emph{outside-in derivation}, and show that together with the subterm property it implies local boundedness. We conclude by showing that \lou\ derivations are outside-in.

\begin{definition}[Outside-In Derivation]
Two $\ls$-steps $\tm\tols\tmtwo\tols\tmthree$ are \deff{outside-in} if the second one substitutes on the subterm substituted by the first one, \ie\ if there exist $\ctx$ and $\ctxtwo$ such that the two steps have the compact form $\ctxp{\var}\tols\ctxp{\ctxtwop{\vartwo}}\tols\ctxp{\ctxtwop{\tmtwo}}$. A derivation is outside-in if any two consecutive substitution steps are outside-in.
\end{definition}

For instance, the first of the following two sequences of steps is outside-in while the second is not:
\begin{align*}
  (\ap\var\vartwo)\esub\var{\ap\vartwo\tm}\esub\vartwo\tmtwo&\tols(\ap{(\ap\vartwo\tm)}\vartwo)\esub\var{\ap\vartwo\tm}\esub\vartwo\tmtwo\\
    &\tols(\ap{(\ap\tmtwo\tm)}\vartwo)\esub\var{\ap\vartwo\tm}\esub\vartwo\tmtwo;\\
  (\ap\var\vartwo)\esub\var{\ap\vartwo\tm}\esub\vartwo\tmtwo&\tols(\ap{(\ap\vartwo\tm)}\vartwo)\esub\var{\ap\vartwo\tm}\esub\vartwo\tmtwo\\
    &\tols(\ap{(\ap\vartwo\tm)}\tmtwo)\esub\var{\ap\vartwo\tm}\esub\vartwo\tmtwo.
\end{align*}
The idea is that outside-in derivations ensure the local boundedness property because 
\begin{varenumerate}
 \item no substitution can be used twice in a outside-in sequence $\tmtwo\tols^k\tmthree$, and
 \item $\tols$ steps do not change the number of substitutions, because they duplicate terms without ES by the subterm property.
\end{varenumerate}
 Therefore, $k$ is necessarily bounded by the number of ES in $\tmtwo$---noted $\esmeas{\tmtwo}$---which in turn is bounded by the number of preceding $\db$-steps. The next lemma formalizes this idea.

\begin{lemma}[Subterm + Outside-In $\Rightarrow$ Local Boundedness]
\label{l:nested-trace}
Let $\tm$ be a $\l$-term, $\deriv:\tm\to^n \tmtwo\tols^k\tmthree$ be a derivation with the subterm property and whose suffix $\tmtwo\tols^k\tmthree$ is outside-in. Then $k\leq\sizedb{\deriv}$.
\end{lemma}

\begin{proof}
Let $\tmtwo=\tmtwo_0\tols\tmtwo_1\tols\ldots\tols \tmtwo_k=\tmthree$
be the outside-in suffix of $\deriv$ and $\tmtwo_i\tols\tmtwo_{i+1}$
one of its steps, for $i\in\set{0,\ldots,k-2}$. Let us use $\ctx_i$
for the external context of the step, \ie\ the context such that
$\tmtwo_i=\ctx_i\ctxholep{\ctxtwop{\var}\esub{\var}{\tmfour}}\tols
\ctx_i\ctxholep{\ctxtwop{\tmfour}\esub{\var}{\tmfour}}=\tmtwo_{i+1}$. The
following outside-in step $\tmtwo_{i+1}\tols\tmtwo_{i+2}$ substitutes
on the substituted occurrence of $\tmfour$. By the subterm property,
$\tmfour$ is a subterm of $\tm$ and so it is has no ES. Then the ES
acting in $\tmtwo_{i+1}\tols\tmtwo_{i+2}$ is on the right of
$\esub{\var}{\tmfour}$, \ie\ the external context $\ctx_{i+1}$ is a
prefix of $\ctx_i$, in symbols $\ctx_{i+1}\outin\ctx_i$.  Since the
derivation $\tmtwo_0\tols\tmtwo_1\tols\ldots\tols \tmtwo_k$ is
outside-in we obtain a sequence
$\ctx_k\outin\ctx_{k-1}\outin\ldots\outin\ctx_0$ of contexts of
$\tmtwo$. In particular, every $\ctx_i$ corresponds to a different
explicit substitution in $\tmtwo$, and so $k\leq\esmeas{\tmtwo}$. Now,
we show that $\esmeas\tmtwo=\sizedb\deriv$, that will conclude the
proof.  The subterm property has also another consequence. Given that
only ordinary $\l$-terms are duplicated, no explicit substitution
constructor is ever duplicated by $\lssym$-steps in $\deriv$: if
$\tmthree\tols\tmfour$ is a step of $\deriv$ then
$\esmeas{\tmthree}=\esmeas{\tmfour}$. Every $\db$-step, instead,
introduces an explicit substitution,
\ie\ $\esmeas\tmtwo=\sizedb\deriv$.
\end{proof}

Since we know that \lou\ derivations have the subterm property,
(\refcoro{subterm-for-lou}), what remains to be shown is that they are
outside-in.

\begin{proposition}
\label{prop:local-bound}
\lou\ derivations are outside-in.
\end{proposition}

\ben{Note that while in \reflemma{nested-trace} the hypothesis that
  the initial term of the derivation is a $\l$-term (relaxable to a
  shallow term) is essential, here---as well as for the the subterm
  property---such an hypothesis is not needed}. Note also that the
last subcase of the proof uses the generalized form of the normal form
property (\reflemma{nf-charac}).

\begin{proof}
We prove the following implication: if the reduction step
$\ctxp{\var}\tols\ctxp{\tmtwo}$ is \lou, and the \lou\ redex $\ctxtwo$
in $\ctxp\tmtwo$ is a $\ls$-redex then $\ctx$ and $\ctxtwo$ are
outside-in, \ie\ $\ctx\outin\ctxtwo$ or $\ctx=\ctxtwo$.
Two cases, depending on \emph{why} the reduction
step $\ctxp{\var}\tols\ctxp{\tmtwo}$ is \opt:
\begin{varenumerate}
\item
   \emph{Relative Creation}, \ie\ $\ctx$ is applicative and $\relunf{\tmtwo}{\ctx}$ is an
  abstraction. Two sub-cases:
  \begin{varvarenumerate}
  \item
    \emph{$\tmtwo$ is an abstraction (in a substitution context)}. Then the \lou\ redex in $\ctxp{\tmtwo}$ is the $\db$-redex having $\tmtwo$ as abstraction, and there is nothing to prove (because $\ctxtwo$ is not a $\ls$-redex).
  \item
    \emph{$\tmtwo$ is not an abstraction}. Then it must a variable $\varthree$ (because it is a $\lambda$-term),
    and $\relunf{\varthree}{\ctx}$ is an abstraction. But then $\ctxp{\tmtwo}$
    is simply $\ctxp{\varthree}$ and the given occurrence of $\varthree$ marks another
    \opt\ substitution redex, \ie\ $\ctxtwo = \ctx$, that is the \lou\ redex because $\ctx$ already was the position of the \lou\ redex at the preceding step.
  \end{varvarenumerate}

\item  \emph{Relative Duplication}, \ie\ $\relunf{\tmtwo}{\ctx}$ is not an abstraction or $\ctx$ is not applicative, but $\relunf{\tmtwo}{\ctx}$ contains a $\beta$-redex. Two sub-cases:
  \ben{
  \begin{varvarenumerate}
  \item \emph{$\tmtwo$ contains a useful redex $\ctxtwo$}. Then the position of the \lou\ redex $\ctxthree$ in $\ctxp\tmtwo$ (that is not necessarily $\ctxp\ctxtwo$) is in $\tmtwo$. Two cases:  
  \begin{varvarvarenumerate}
    		\item \emph{$\ctxthree$ is a $\db$-redex}. Then there is nothing to prove, because the \lou\ redex is not a $\ls$-redex.
		\item \emph{$\ctxthree$ is a $\ls$-redex}. Then the two steps are outside-in, because $\ctx$ is a prefix of $\ctxthree$.
    \end{varvarvarenumerate}
  \item
    \emph{$\tmtwo$ is a useful normal form}. Since $\relunf{\tmtwo}{\ctx}$ does contain a $\beta$-redex, we can apply the contextual normal form property (\reflemma{nf-charac}) and obtain that there exists a useful $\ls$-redex $\ctxtwo$ in $\ctxp\tmtwo$ such that $\ctx\outin\ctxtwo$. Then the $\leftout$-minimum of these redexes is the \lou\ redex in $\ctxp\tmtwo$, and $\ctx$ and $\ctxtwo$ are outside-in as redexes.\qedhere
  \end{varvarenumerate}
  }
\end{varenumerate} 
\end{proof}

\begin{corollary}[Local Boundedness Property]
 \lou\ derivations \ben{(starting on $\l$-terms)} have the local boundedness property.
\end{corollary}

\begin{proof}
 By \refcoro{subterm-for-lou} \lou\ derivations have the subterm property and by \refprop{local-bound} they are outside-in. The initial term is a $\l$-term, and so \reflemma{nested-trace} (subterm + outside-in $\Rightarrow$ local boundedness) provides the local boundedness property.
\end{proof}

At this point, we proved all the required properties for the implementation theorems but for the selection property for \lou\ derivations, addressed by the next section.

\section{The Selection Property, or Computing Functions in Compact Form}\label{sect:properties}
\label{sect:algorithm}
\newcommand{\alg}{\mathcal{A}}
\newcommand{\algone}{\mathcal{A}}
\newcommand{\algtwo}{\mathcal{B}}
\newcommand{\algthree}{\mathcal{C}}
\newcommand{\fun}{f}
\newcommand{\funone}{f}
\newcommand{\funtwo}{g}
\newcommand{\funequal}{f_{=}}
\newcommand{\nattm}{\mathbb{T}}
\newcommand{\natvar}[1]{\mathsf{var}(#1)}
\newcommand{\natlam}{\mathsf{lam}}
\newcommand{\natapp}{\mathsf{app}}
\newcommand{\natval}{n}
\newcommand{\bool}{\mathbb{B}}
\newcommand{\btrue}{\mathsf{true}}
\newcommand{\bfalse}{\mathsf{false}}
\newcommand{\bval}{b}
\newcommand{\funnat}{\mathit{nature}}
\newcommand{\funred}{\mathit{redex}}
\newcommand{\funav}{\mathit{apvars}}
\newcommand{\funfv}{\mathit{freevars}}
\newcommand{\vars}{\mathcal{VARS}}
\newcommand{\ifnempty}[3]{#1\Downarrow_{#2,#3}}
\newcommand{\len}[1]{|#1|}
\newcommand{\eq}{\mathit{alpha}}
\renewcommand{\setone}{V}
\renewcommand{\settwo}{W}
\newcommand{\targetset}{A}

This section proves the selection property for \lou\ derivations,
which is the missing half of the proof that they are efficiently mechanizable,
\ie\ that they enjoy the low-level implementation theorem. The proof
consists in providing a polynomial algorithm for testing the
usefulness of a substitution step. The subtlety is that the test has
to check whether a term of the form $\relunf{\tm}{\ctx}$ contains a
$\beta$-redex, or whether it is an abstraction, without explicitly
computing $\relunf{\tm}{\ctx}$ (which, of course, takes exponential
time in the worst case). If one does not prove that this test can be done
in time polynomial in (the size of) $\tm$ and $\ctx$, then firing
\emph{a single} reduction step can cause an exponential blowup!

Our algorithm consists in the simultaneous computation of four
functions on terms in compact form, two of which will provide the
answer to our problem. We need some abstract preliminaries about
computing functions in compact form.

A function $\funone$ from $n$-tuples of $\lambda$-terms to
a set $\targetset$ is said to have \emph{arity} $n$, and we write 
$\funone:n\rightarrow\targetset$ in this case. 
The function $\funone$ is said to be:
\begin{varitemize}
\item
  \emph{Efficiently computable} if there is a polynomial time
  algorithm $\alg$ such that for every $n$-uple of $\lambda$-terms
  $(\tm_1,\ldots,\tm_n)$, the result of
  $\alg(\tm_1,\ldots,\tm_n)$ is precisely $\funone(\tm_1,\ldots,\tm_n)$.
\item
  \emph{Efficiently computable in compact form} if there is a polynomial
  time algorithm $\alg$ such that for every $n$-uple of LSC terms
  $(\tm_1,\ldots,\tm_n)$, the result of
  $\alg(\tm_1,\ldots,\tm_n)$ is precisely $\funone(\unf{\tm_1},\ldots,\unf{\tm_n})$.
\item
  \begin{sloppypar}
    \emph{Efficiently computable in compact form relative to a context} if there is a polynomial
    time algorithm $\alg$ such that for every $n$-uple of pairs of LSC terms and contexts
    $((\tm_1,\ctx_1),\ldots,(\tm_n,\ctx_n))$, the result of
    $\alg((\tm_1,\ctx_1),\ldots,(\tm_n,\ctx_n)))$ is precisely 
    $\funone(\relunf{\tm_1}{\ctx_1},\ldots,\relunf{\tm_n}{\ctx_n})$.
  \end{sloppypar}
\end{varitemize}\medskip
An example of such a function is $\eq:2\rightarrow\bool$, which given two
$\lambda$-terms $\tm$ and $\tmtwo$, returns $\btrue$ if $\tm$ and
$\tmtwo$ are $\alpha$-equivalent and $\bfalse$
otherwise. In~\cite{DBLP:conf/rta/AccattoliL12}, $\eq$ is shown to be
efficiently computable in compact form, via a dynamic programming
algorithm $\algtwo_\eq$ taking as input two LSC terms and computing,
for every pair of their subterms, whether the (unfoldings) are
$\alpha$-equivalent or not. Proceeding bottom-up, as usual in dynamic
programming, permits to avoid the costly task of computing unfoldings
explicitly, which takes exponential time in the worst-case. More
details about $\algtwo_\eq$ can be found
in~\cite{DBLP:conf/rta/AccattoliL12}.

Each one of the functions of our interest takes values in one of the following sets:
\begin{align*}
	\vars &= \mbox{ the set of finite sets of variables}\\
  \bool&=\{\btrue,\bfalse\}\\
    \nattm&=\{\natvar{\var}\mid\mbox{ $\var$ is a variable}\}\cup\{\natlam,\natapp\}
\end{align*}
Elements of $\nattm$ represent the \emph{nature} of a term. The functions we need are:
\begin{varitemize}
\item
  $\funnat:1\rightarrow\nattm$, which returns the nature of the input term;
\item
  $\funred:1\rightarrow\bool$, which returns $\btrue$ if the input term
  contains a redex and $\bfalse$ otherwise;
\item
  $\funav:1\rightarrow\vars$, which returns the set of variables that have
  a free occurrence in applicative position in the input term;
\item
  $\funfv:1\rightarrow\vars$, which returns the set of free variables occurring
  in the input term.
\end{varitemize}\medskip
Note that they all have arity 1 and that showing $\funred$ and $\funnat$ to be \emph{efficiently computable in compact form relative to a context}
is precisely what is required to prove the efficiency of useful reduction. 

The four functions above can all be proved to be efficiently
computable (in the three meanings). It is convenient to do so by
giving an algorithm computing the product function
$\funnat\times\funred\times\funav\times\funfv:1\rightarrow\nattm\times\bool\times\vars\times\vars$
(which we call $\funtwo$) compositionally, on the structure of the
input term, because the four function are interrelated (for example,
$\tm\tmtwo$ has a redex, \ie\ $\funred(\tm\tmtwo)=\btrue$, if $\tm$ is
an abstraction, \ie\ if $\funnat(\tm)=\natlam$). The algorithm
computing $\funtwo$ on terms is $\alg_\funtwo$ and is defined in
Figure~\ref{fig:explicit}.
\begin{figure*}
\begin{center}
\fbox{
\begin{minipage}{.97\textwidth}
\begin{align*}
  \alg_\funtwo(\var)&=(\natvar{\var},\bfalse,\emptyset,\{\var\});\\
  \alg_\funtwo(\l\var.\tm)&=(\natlam,\bval_\tm,\setone_\tm-\{\var\},\settwo_\tm-\{\var\})\\
  &\mbox{where }\alg_\funtwo(\tm)=(\natval_\tm,\bval_\tm,\setone_\tm,\settwo_\tm);\\
  \alg_\funtwo(\tm\tmtwo)&=(\natapp,\bval_\tm\vee\bval_\tmtwo\vee(\natval_\tm=\natlam),\setone_\tm\cup\setone_\tmtwo\cup\{\var\mid\natval_\tm=\natvar{\var}\},\settwo_\tm\cup\settwo_\tmtwo)\\
  &\mbox{where }\alg_\funtwo(\tm)=(\natval_\tm,\bval_\tm,\setone_\tm,\settwo_\tm)\mbox{ and }\alg_\funtwo(\tmtwo)=(\natval_\tmtwo,\bval_\tmtwo,\setone_\tmtwo,\settwo_\tmtwo);
\end{align*}
\end{minipage}}
\end{center}
\caption{Computing $\funtwo$.}\label{fig:explicit}
\end{figure*}
 
 The interesting case in the algorithms for the two compact cases is the one for ES, that makes use of a special notation: given two sets of variables $\setone,\settwo$ and a variable $\var$, $\ifnempty{\setone}{\var}{\settwo}$ is defined to be $\setone$ if $\var\in\settwo$ and the empty set $\emptyset$ otherwise. The algorithm $\algtwo_\funtwo$ computing $\funtwo$ on LSC terms is defined in Figure~\ref{fig:implicit}.
\begin{figure*}
\begin{center}
\fbox{
\begin{minipage}{.97\textwidth}
\begin{align*}
  \algtwo_\funtwo(\var)&=(\natvar{\var},\bfalse,\emptyset,\{\var\});\\
  \algtwo_\funtwo(\l\var.\tm)&=(\natlam,\bval_\tm,\setone_\tm-\{\var\},\settwo_\tm-\{\var\})\\
  &\mbox{where }\algtwo_\funtwo(\tm)=(\natval_\tm,\bval_\tm,\setone_\tm,\settwo_\tm);\\
  \algtwo_\funtwo(\tm\tmtwo)&=(\natapp,\bval_\tm\vee\bval_\tmtwo\vee(\natval_\tm=\natlam),\setone_\tm\cup\setone_\tmtwo\cup\{\var\mid\natval_\tm=\natvar{\var}\},\settwo_\tm\cup\settwo_\tmtwo)\\
  &\mbox{where }\algtwo_\funtwo(\tm)=(\natval_\tm,\bval_\tm,\setone_\tm,\settwo_\tm)\mbox{ and }\algtwo_\funtwo(\tmtwo)=(\natval_\tmtwo,\bval_\tmtwo,\setone_\tmtwo,\settwo_\tmtwo);\\    
  \algtwo_\funtwo( \tm\esub\var\tmtwo)&=(\natval,\bval,\setone,\settwo)\\
  &\mbox{where }\algtwo_\funtwo(\tm)=(\natval_\tm,\bval_\tm,\setone_\tm,\settwo_\tm)\mbox{ and }\algtwo_\funtwo(\tmtwo)=(\natval_\tmtwo,\bval_\tmtwo,\setone_\tmtwo,\settwo_\tmtwo)\mbox{ and:}\\
  &\qquad\natval_\tm=\natvar{\var}\Rightarrow\natval=\natval_\tmtwo;\quad\natval_\tm=\natvar{\vartwo}\Rightarrow\natval=\natvar{\vartwo};\\
  &\qquad\natval_\tm=\natlam\Rightarrow\natval=\natlam;\quad\natval_\tm=\natapp\Rightarrow\natval=\natapp;\\
  &\qquad\bval=\bval_\tm\vee(\bval_\tmtwo\wedge\var\in\settwo_\tm)\vee((\natval_\tmtwo=\natlam)\wedge(\var\in\setone_\tm));\\
  &\qquad\setone=(\setone_\tm-\{\var\})\cup\ifnempty{\setone_\tmtwo}{\var}{\settwo_\tm}\cup\;\{\vartwo\mid\natval_\tmtwo=\natvar{\vartwo}\wedge\var\in\setone_\tm\};\\
  &\qquad\settwo=(\settwo_\tm-\{\var\})\cup\ifnempty{\settwo_\tmtwo}{\var}{\settwo_\tm}
\end{align*}
\end{minipage}}
\end{center}
\caption{Computing $\funtwo$ in compact form.}\label{fig:implicit}
\end{figure*}
The algorithm computing $\funtwo$ on pairs in the form $(\tm,\ctx)$ (where $\tm$ is a LSC term and $\ctx$ is
a shallow context) is defined in Figure~\ref{fig:implicitcontext}.
\begin{figure*}
\begin{center}
\fbox{
\begin{minipage}{.97\textwidth}
\begin{align*}
  \algthree_\funtwo(\tm,\ctxhole)&=\algtwo_\funtwo(\tm);\\
  \algthree_\funtwo(\tm,\l\var.\ctx)&=\algthree_\funtwo(\tm,\ctx);\\
  \algthree_\funtwo(\tm,\ctx\tmtwo)&=\algthree_\funtwo(\tm,\ctx);\\
  \algthree_\funtwo(\tm,\tmtwo\ctx)&=\algthree_\funtwo(\tm,\ctx);\\
  \algthree_\funtwo(\tm,\ctx\esub\var\tmtwo)&=(\natval,\bval,\setone,\settwo)\\
  &\mbox{where }\algthree_\funtwo(\tm,\ctx)=(\natval_{\tm,\ctx},\bval_{\tm,\ctx},\setone_{\tm,\ctx},\settwo_{\tm,\ctx})
   \mbox{ and }\algtwo_\funtwo(\tmtwo)=(\natval_\tmtwo,\bval_\tmtwo,\setone_\tmtwo,\settwo_\tmtwo)\mbox{ and:}\\
  &\qquad\natval_{\tm,\ctx}=\natvar{\var}\Rightarrow\natval=\natval_\tmtwo;
  \quad\natval_{\tm,\ctx}=\natvar{\vartwo}\Rightarrow\natval=\natvar{\vartwo};\\
  &\qquad\natval_{\tm,\ctx}=\natlam\Rightarrow\natval=\natlam;
  \quad\natval_{\tm,\ctx}=\natapp\Rightarrow\natval=\natapp;\\
  &\qquad\bval=\bval_{\tm,\ctx}\vee(\bval_\tmtwo\wedge\var\in\settwo_{\tm,\ctx})\vee((\natval_\tmtwo=\natlam)\wedge(\var\in\setone_{\tm,\ctx}));\\
  &\qquad\setone=(\setone_{\tm,\ctx}-\{\var\})\cup\ifnempty{\setone_\tmtwo}{\var}{\settwo_{\tm,\ctx}}
  \cup\;\{\vartwo\mid\natval_\tmtwo=\natvar{\vartwo}\wedge\var\in\setone_{\tm,\ctx}\};\\
  &\qquad\settwo=(\settwo_{\tm,\ctx}-\{\var\})\cup\ifnempty{\settwo_\tmtwo}{\var}{\settwo_{\tm,\ctx}}
\end{align*}
\end{minipage}}
\end{center}
\caption{Computing $\funtwo$ in compact form, relative to a context.}\label{fig:implicitcontext}
\end{figure*}

First of all, we need to convince ourselves about the \emph{correctness} of the proposed algorithms: do they really
compute the function $\funtwo$? Actually, the way the algorithms are defined, namely by primitive recursion on the
input terms, helps very much here: a simple induction suffices to prove the following:
\begin{proposition}
  The algorithms $\alg_\funtwo$, $\algtwo_\funtwo$, and $\algthree_\funtwo$ are all correct: for every
  $\lambda$-term $\tm$, for every LSC term $\tmtwo$ and for every context $\ctx$, we have
  \begin{varenumerate}
   \item $\alg_\funtwo(\tm)=\funtwo(\tm)$;
   \item $\algtwo_\funtwo(\tmtwo)=\funtwo(\unf{\tmtwo})$;
   \item $\algthree_\funtwo(\tmtwo,\ctx)=\funtwo(\relunf{\tmtwo}{\ctx})$.
  \end{varenumerate}  
\end{proposition}
\begin{proof}[Proof Sketch.]\hfill
    \begin{varenumerate}
    \item
      The equation $\alg_\funtwo(\tm)=\funtwo(\tm)$ can be proved by induction on the structure of $\tm$. An interesting
      case:
      \begin{varitemize}
      \item
        If $\tm=\tmtwo\tmthree$, then we know that:
        \begin{align*}
          \alg_\funtwo(\tmtwo\tmthree)&=(\natapp,\bval_\tmtwo\vee\bval_\tmthree\vee(\natval_\tmtwo=\natlam),\setone_\tmtwo\cup\setone_\tmthree\cup\{\var\mid\natval_\tmtwo=\natvar{\var}\},\settwo_\tmtwo\cup\settwo_\tmthree)\\
          &\mbox{where }\alg_\funtwo(\tmtwo)=(\natval_\tmtwo,\bval_\tmtwo,\setone_\tmtwo,\settwo_\tmtwo)\mbox{ and }\alg_\funtwo(\tmthree)=(\natval_\tmthree,\bval_\tmthree,\setone_\tmthree,\settwo_\tmthree);
        \end{align*}
        Now, first of all observe that $\funred(\tm)=\btrue$ if and only if there is a redex in $\tmtwo$ or 
        a redex in $\tmthree$ or if $\tmtwo$ is a $\lambda$-abstraction. Moreover, the variables occurring in
        applicative position in $\tm$ are those occurring in applicative position in either $\tmtwo$ or in
        $\tmthree$ or $\var$, if $\tmtwo$ is $\var$ itself. Similarly, the variables occurring free in $\tm$
        are simply those occurring free in either $\tmtwo$ or in $\tmthree$. The thesis can be synthesized easily
        from the inductive hypothesis.
      \end{varitemize}
    \item
      The equation $\algtwo_\funtwo(\tmtwo)=\funtwo(\unf{\tmtwo})$ can be proved by induction on the structure of $\tmtwo$, using
      the correctness of $\alg$.
    \item
      The equation $\algthree_\funtwo(\tmtwo,\ctx)=\funtwo(\relunf{\tmtwo}{\ctx})$ can be proved by induction on the structure of $\ctx$, using
      the correctness of $\algtwo$.
    \end{varenumerate}
    This concludes the proof.
  \end{proof}
The way the algorithms above have been defined also helps while proving that they work in bounded time, e.g., the number of recursive
calls triggered by $\alg_\funtwo(\tm)$ is linear in $\len{\tm}$ and each of them takes polynomial time. As a consequence, we can also easily bound the complexity of the three algorithms at hand.
\begin{proposition}[Selection Property]\label{prop:polytimealgos}
  The algorithms $\alg_\funtwo$, $\algtwo_\funtwo$,
  and $\algthree_\funtwo$ all work in polynomial time. Thus the \lou\ strategy has the selection property.
\end{proposition}
  \begin{proof}
    The three algorithms are defined by primitive recursion. More
    specifically:
    \begin{varitemize}
    \item Any call $\alg_\funtwo(\tm)$ triggers at most $\len{\tm}$
      calls to $\alg_\funtwo$;
    \item
      Any call $\algtwo_\funtwo(\tm)$ triggers at most $\len{\tm}$ calls to $\algtwo_\funtwo$;
    \item
      Any call $\algthree_\funtwo(\tm,\ctx)$ triggers at most $\len{\tm}+\len{\ctx}$ calls to $\algtwo$ and
      at most $\len{\ctx}$ calls to $\algthree$;
    \end{varitemize}
    Now, the amount of work involved in any single call (not counting the, possibly recursive, calls) 
    is itself polynomial, simply because the tuples produced in output are made of objects whose
    size is itself bounded by the length of the involved terms and contexts.
\end{proof}
\ugo{What Proposition~\ref{prop:polytimealgos} implicitly tells us is
  that the usefulness of a given redex in an LSC term $\tm$ can be
  checked in polynomial time in the size of $\tm$. The Selection
  Property (Definition~\ref{def:lli}) then holds for LOU derivations:
  the next redex to be fired is the \lo\ useful one (of course, finding the \lo\ useful redex among useful redexes can trivially be done in polynomial time).}

\section{Summing Up}
\label{sect:summing-up}

The various ingredients from the previous sections can be combined
so as to obtain the following result:
\ugo{
\begin{theorem}[Polynomial Implementation of $\l$]\label{theo:invariance}
  There is an algorithm which takes as input a $\l$-term $\tm$ and
  a natural number $n$ and
  which, in time polynomial in $m=\min\{n,\nos{\toblo}{\tm}\}$ and $\size{\tm}$,
  outputs an \lsc\ term $\tmtwo$ such that $\tm\rightarrow^m\unf{\tmtwo}$.
\end{theorem}}
Together with the linear implementation of Turing machines in the
$\l$-calculus given in \cite{DBLP:conf/rta/AccattoliL12}, we obtain
our main result.
\begin{theorem}[Invariance]
 The $\l$-calculus is a reasonable model in the sense of the weak invariance thesis.
\end{theorem}

As we have already mentioned, the algorithm witnessing the invariance
of the $\l$-calculus does \emph{not} produce a $\l$-term, but a
useful normal form, \ie\ a compact representation (with ES) of a $\l$-term. 
Theorem~\ref{theo:invariance}, together with the fact that
equality of terms can be checked efficiently \emph{in compact form}
entail the following formulation of invariance, akin in spirit to,
\eg, Statman's Theorem~\cite{DBLP:journals/tcs/Statman79a}:
\begin{corollary}\label{coro:statmanlike}
  There is an algorithm which takes as input two $\l$-terms $\tm$ and
  $\tmtwo$ and checks whether $\tm$ and $\tmtwo$ have the same
  normal form in time polynomial in $\nos{\toblo}{\tm}$,
  $\nos{\toblo}{\tmtwo}$, $\size{\tm}$, and $\size{\tmtwo}$.
\end{corollary}

\noindent If one instantiates Corollary~\ref{coro:statmanlike} to the
case in which $\tmtwo$ is a (useful) normal form, one obtains that checking
whether the normal form of any term $\tm$ is equal to (the unfolding of) $\tmtwo$ can be
done in time polynomial in $\nos{\toblo}{\tm}$, $\size{\tm}$, and
$\size{\tmtwo}$. This is particularly relevant when the size of
$\tmtwo$ is constant, \eg, when the $\l$-calculus computes decision
problems and the relevant results are truth values.

Please observe that whenever one (or both) of the involved terms are
\emph{not} normalizable, the algorithms above (correctly) diverge.

\section{Discussion}
\label{sect:discussion}

\emph{Applications}. \ben{One might wonder what is the practical
  relevance of our invariance result, since functional programming
  languages rely on weak evaluation, for which invariance was already
  known. The main application of strong evaluation is in the design of
  proof assistants and higher-order logic programming, typically for
  type-checking in frameworks with dependent types as the Edinburgh
  Logical Framework or the Calculus of Constructions, as well as for
  unification modulo $\beta\eta$ in simply typed frameworks like
  $\l$-Prolog. Of course, in these cases the language at work is not
  as minimalistic as the $\l$-calculus, it is often typed, and other
  operations (\eg\ unification) impact on the complexity of
  evaluation. Nonetheless, the strong $\l$-calculus is always the core
  language, and so having a reasonable cost model for it is a
  necessary step for complexity analyses of these frameworks. Let us
  point out, moreover, that in the study of functional programming
  languages there is an emerging point of view, according to which the
  theoretical study of the language should be done with respect to
  strong evaluation, even if only weak evaluation will be implemented,
  see \cite{DBLP:conf/esop/SchererR15}.  We also believe that our work
  may be used to substantiate the practical relevance of some
  theoretical works. There exists a line of research attempting to
  measure the number of steps to evaluate a term by looking to its
  denotational interpretations (\eg\ relational semantics/intersection
  types in
  \cite{DBLP:journals/corr/abs-0905-4251,DBLP:journals/tcs/CarvalhoPF11,DBLP:journals/corr/BernadetL13,DBLP:conf/lics/LairdMMP13}
  and game semantics in
  \cite{DBLP:conf/popl/Ghica05,DBLP:conf/csl/LagoL08,DBLP:conf/tlca/Clairambault13})
  with the aim of providing abstract formulations of complexity
  properties. The problem of this literature is that either the
  measured strong strategies do not provide reliable complexity measures, or they only address head/weak reduction. In
  particular, the number of \lo\ steps to normal form---\ie\ our cost
  model---has never been measured with denotational tools. This is
  particularly surprising, because head reduction is the strategy
  arising from denotational considerations (this is the leading theme
  of Barendregt's book \cite{Barendregt84}) and the \lo\ strategy is
  nothing but iterated head reduction. We expect that our result will
  be the starting point for revisiting the quantitative analyses of
  $\beta$-reduction based on denotational semantics.  }

\emph{Mechanizability vs Efficiency}.  Let us stress that the study of
invariance is about \emph{mechanizability} rather than
\emph{efficiency}. One is not looking for the smartest or shortest
evaluation strategy, but rather for one that can be reasonably
implemented. The case of \levy's optimal evaluation, for instance,
hides the complexity of its implementation in the cleverness of its
definition. A \levy-optimal derivation, indeed, can be even shorter
than the shortest sequential strategy, but---as shown by Asperti and
Mairson \cite{DBLP:conf/popl/AspertiM98}---its definition hides
hyper-exponential computations, so that optimal derivations do not
provide an invariant cost model. The leftmost-outermost strategy, is a
sort of \emph{maximally unshared} normalizing strategy, where redexes
are duplicated whenever possible and unneeded redexes are never
reduced, somehow dually with respect to optimal derivations. It is
exactly this \emph{inefficiency} that induces the subterm property,
the key point for its mechanizability. It is important to not confuse
two different levels of sharing: our \lou\ derivations share
\emph{subterms}, but not \emph{computations}, while \levy's optimal
derivations do the opposite. By sharing computations optimally, they
collapse the complexity of too many steps into a single one, making
the number of steps an unreliable measure.

\emph{Inefficiencies}. This work is foundational in spirit and only
deals with polynomial bounds, and in particular it does not address an
efficient implementation of useful sharing. There are three main
sources of inefficiency:
\begin{varenumerate}
\item \emph{Call-by-Name Evaluation Strategy}: for a more efficient evaluation
  one should at least adopt a call-by-need policy, while many would
  probably prefer to switch to call-by-value altogether. Both
  evaluations introduce some sharing of \emph{computations} with
  respect to call-by-name, as they evaluate the argument before it is substituted (call-by-need) or the $\beta$-redex is fired (call-by-value). Our choice of call-by-name, however, comes
  from the desire to show that even the good old $\l$-calculus with
  normal order evaluation is invariant, thus providing a simple cost
  model for the working theoretician.
\item
  \emph{High-Level Quadratic Overhead}: in the micro-step evaluation
  presented here the number of substitution steps is at most quadratic
  in the number of $\beta$-steps, as proved in the high-level
  implementation theorem. Such a bound is strict, as there exist
  degenerate terms that produce these quadratic substitution
  overhead---for instance, the micro-step evaluation of the
  paradigmatic diverging term $\Omega$, but the degeneracy can also be
  adapted to terminating terms.
\item
  \emph{Low-Level Separate Useful Tests}: for the low-level
  implementation theorem we provided a separate \emph{global} test for
  the usefulness of a substitution step. It is natural to wonder if an
  abstract machine can implement it \emph{locally}. The idea we
  suggested in \cite{DBLP:conf/csl/AccattoliL14} is that some
  additional labels on subterms may carry information about the
  unfolding in their context, allowing to decide usefulness in linear
  time, and removing the need of running a global check.
\end{varenumerate}\medskip
These inefficiencies have been addressed by Accattoli and Sacerdoti
Coen in two studies \cite{DBLP:conf/wollic/AccattoliC14,usef-constr},
complementary to ours. In~\cite{DBLP:conf/wollic/AccattoliC14}, they
show that (in the much simpler weak case) call-by-value and
call-by-need both satisfy an high-level implementation theorem and
that the quadratic overhead is induced by potential chains of renaming
substitutions, sometimes called \emph{space leaks}. Moreover, in
call-by-value and call-by-need the quadratic overhead can be reduced
to \emph{linear} by simply removing variables from values. The same
speed-up can be obtained for call-by-name as well, if one slightly
modifies the micro-step rewriting rules (see the long version of
\cite{DBLP:conf/wollic/AccattoliC14}---that at the time of writing is
submitted and can only be found on Accattoli's web page---that builds
on a result of \cite{DBLP:conf/birthday/SandsGM02}).   In
\cite{usef-constr}, instead, the authors address the possibility of
local useful tests, but motivated by
\cite{DBLP:conf/wollic/AccattoliC14}, they rather do it for a weak
call-by-value calculus generalized to evaluate open terms, that is the
evaluation model used by the abstract machine at work in the Coq proof
assistant \cite{DBLP:conf/icfp/GregoireL02}. Despite being a weak
setting, open terms force to address useful sharing along the lines of
what we did here, but with some simplifications due to the weak
setting. The novelty of \cite{usef-constr} is an abstract machine
implementing useful sharing, and studied via a \emph{distillation},
\ie\ a new methodology for the representation of abstract machines in
the LSC \cite{DBLP:conf/icfp/AccattoliBM14}. Following the mentioned
suggestion, the machine uses simple labels to check usefulness
\emph{locally} and---surprisingly---the check takes \emph{constant
  time}. Globally, the machine is proved to have an overhead that is
\emph{linear} both in the number of $\beta$-steps and the size of the
initial term. Interestingly, that work builds on the schema for
usefulness that we provided here, showing that the our approach, and
in particular useful sharing, are general enough to encompass more
efficient scenarios. But there is more. At first sight call-by-value
seemed to be crucial in order to obtain a linear overhead, but the
tools of \cite{usef-constr}---a posteriori---seem to be adaptable to
call-by-name, with a slight slowdown: useful tests are checked in
linear rather than constant time (linear in the size of the initial
term). For call-by-need with open terms, the same tools seem to apply,
even if we do not yet know if useful tests are linear or constant.
Generally speaking, our result can be improved along two superposing
axes. One is to refine the invariant strategy so as to include as much
sharing of computations as possible, therefore replacing call-by-name
with call-by-value or call-by-need with open terms, or under
abstractions. The other axe is to refine the overhead in implementing
micro-step useful evaluation (itself splitting into two high-level and
low-level axes), which seems to be doable in (bi)linear time more or
less independently of the strategy.

\emph{On Non-Deterministic $\beta$-Reduction.}  This paper only deals
with the cost of reduction induced by the natural, but inefficient
leftmost-outermost strategy. The invariance of full $\beta$-reduction,
\ie\ of the usual non-deterministic relation allowed to reduce
$\beta$-redexes in any order, would be very hard to obtain, since it
would be equivalent to the invariance of the cost model induced by the
optimal \emph{one-step} deterministic reduction strategy, which is
well known to be even non-recursive~\cite{Barendregt84}. Note that, a
priori, \emph{non-recursive} does not imply \emph{non-invariant}, as
there may be an algorithm for evaluation polynomial in the steps of
the optimal strategy and that does not simulate the strategy itself---the
existence of such an algorithm, however, is unlikely. 
The optimal \emph{parallel} reduction strategy is instead recursive
but, as mentioned in the introduction, the number of its steps to
normal form is well known \emph{not} to be an invariant cost
model~\cite{DBLP:conf/popl/AspertiM98}.


\section{Conclusions}
This work can be seen as the last tale in the long quest for an
invariant cost model for the $\lambda$-calculus.  In the last ten
years, the authors have been involved in various works in which
\emph{parsimonious} time cost models have been shown to be invariant
for more and more general notions of reduction, progressively relaxing
the conditions on the use of
sharing~\cite{DBLP:journals/tcs/LagoM08,DBLP:journals/corr/abs-1208-0515,DBLP:conf/rta/AccattoliL12}. None
of the results in the literature, however, concerns reduction to
normal form as we do here.

We provided the first full answer to a long-standing open problem: the
$\l$-calculus is indeed a reasonable machine, if the length of the
leftmost-outermost derivation to normal form is used as cost model.

To solve the problem we developed a whole new toolbox: an abstract
deconstruction of the problem, a theory of useful derivations, a
general view of functions efficiently computable in compact form, and
a surprising connection between standard and efficiently mechanizable
derivations. Theorem after theorem, an abstract notion of machine
emerges, hidden deep inside the $\l$-calculus itself. While such a
machine is subtle, the cost model turns out to be the simplest and
most natural one, as it is unitary, machine-independent, and justified
by the standardization theorem, a classic result apparently unrelated
to the complexity of evaluation.

This work also opens the way to new studies. Providing an invariant
cost model, \ie\ a metric for efficiency, it gives a new tool to
compare different implementations, and to guide the development of
new, more efficient ones. As discussed in the previous section,
Accattoli and Sacerdoti Coen presented a call-by-value abstract
machine for useful sharing having only a linear overhead
\cite{usef-constr}, \ben{that actually on open $\l$-terms is
  asymptotically faster than the abstract machine at work in the Coq
  proof assistant, studied in \cite{DBLP:conf/icfp/GregoireL02}}. Such
a result shows that useful sharing is not a mere theoretical tool, and
justifies a finer analysis of the invariance of $\l$-calculus.

Among the consequences of our results, one can of course mention that
proving systems to characterize time complexity classes equal or
larger than $\mathbf{P}$ can now be done merely by deriving bounds on
the \emph{number} of leftmost-outermost reduction steps to normal
form. This could be useful, for instance, in the context of
\emph{light
  logics}~\cite{DBLP:conf/csl/GaboardiR07,DBLP:journals/lmcs/CoppolaLR08,DBLP:journals/iandc/BaillotT09}. The
kind of bounds we obtain here are however more \emph{general} than
those obtained in implicit computational complexity, because we deal
with a universal model of computation.

While there is room for finer analyses, we consider the understanding
of time invariance essentially achieved. However, the study of cost
models for $\l$-terms is far from being over. Indeed, the study of
space complexity for functional programs has only made its very first
steps
\cite{DBLP:conf/lics/Schopp07,DBLP:conf/popl/GaboardiMR08,DBLP:conf/esop/LagoS10,DBLP:conf/csl/Mazza15},
and not much is known about invariant \emph{space} cost models.

\end{document}